\documentclass[12pt]{article}

\usepackage[title]{appendix}

\usepackage[round,authoryear]{natbib}
\def\citeinb #1{\citeauthor{#1}, \citeyear{#1}}

\usepackage{xcolor}
\definecolor{myblue}{RGB}{51,51,178}

\usepackage[utf8]{inputenc}
\usepackage[margin=1in]{geometry}
\usepackage[onehalfspacing]{setspace}
\usepackage{etoolbox}
\AtBeginEnvironment{proof}{\small}   
\usepackage{graphicx}
\graphicspath{ {./images/} }					

\usepackage[font=footnotesize]{caption}
\usepackage{float}

\usepackage{subcaption}							
\usepackage{rotating}							

\usepackage{adjustbox}
\newsavebox{\figbox}
\newlength{\notewidth}
\usepackage{tikz}
\usetikzlibrary{arrows.meta,calc,positioning}

\tikzset{
  control/.style={blue, thick},
  treated/.style={red, thick},
  benchmark/.style={red!75!black, thick, densely dotted},
  axis/.style={->, line width=0.9pt},
  bias/.style = {<->, >={Stealth[length=1.2mm,width=1mm]}, thin},
  every node/.style={font=\small}
}

\usepackage{pgfplots}							
\pgfplotsset{compat=1.18}
\pgfplotsset{flexsaving/.style={
    width=0.92\linewidth, height=0.55\linewidth,
    xmin=-2.2, xmax=5.2, ymin=0.25, ymax=0.65,
    xtick={-2,-1,0,1,2,3,4,5}, ytick={0.3,0.4,0.5,0.6},
    xlabel={Event time (year $-$ 2003)},
    ylabel={Flexible-term saving ratio},
    axis lines=left, tick align=outside,
    tick label style={font=\small}, label style={font=\small},
    scaled y ticks=false,
    yticklabel style={/pgf/number format/fixed,/pgf/number format/precision=2},
    legend style={at={(0.5,1.04)},anchor=south,
                  legend columns=2,draw=none,font=\small},
    legend cell align=left,
}}

\usepackage{empheq}								

\usepackage[hidelinks]{hyperref}

\usepackage[shortlabels]{enumitem}				

\usepackage{amsthm}								
\theoremstyle{definition}
\newtheorem{assumption}{Assumption}

\numberwithin{definition}{subsection}			
\newtheorem{lemma}{Lemma}
\newtheorem{theorem}{Theorem}
\newtheorem{remark}{Remark}[section]

\newcommand{\abs}[1]{\left\lvert#1\right\rvert}							
\newcommand{\norm}[1]{\left\lVert#1\right\rVert}						

\usepackage{amsmath, amsfonts, amssymb, amsthm, mathtools, bm}

\usepackage{bbm}													
\newcommand{\1}{\mathbbm{1}}
\newcommand{\E}{\mathbbm{E}}

\usepackage{mathrsfs}												

\usepackage{mathtools}

\usepackage{hyperref}

\hypersetup{
	colorlinks=true, 													
	linktoc=all,        												
	linkcolor=blue,  													
	citecolor=blue,
	urlcolor=blue
}

\usepackage{booktabs}							
\usepackage{caption}
\usepackage{multirow}
\usepackage{threeparttable}						
\usepackage{makecell}  

\newcommand{\bs}{\mathbf{S}}

\newcommand{\G}{\mathbf{G}}
\newcommand{\Ical}{\mathcal{I}}

\newcommand{\isum}{ \sum_{i=1}^{n} }

\newcommand{\plim}{\xrightarrow{p}}
\newcommand{\V}{\mathbb{V}}

\newcommand{\X}{\mathbf{X}}
\newcommand{\Xcal}{\mathcal{X}}

	\title{Matched Triple-Differences: A Framework for Covariate Adjustment \thanks{I thank Douglas Steigerwald, Max Farrell, and Gonzalo Vazquez-Bare for valuable comments and suggestions.}}
    \author{Yihong Liu\thanks{Department of Economics, UC Santa Barbara.}}
    \date{October 2nd, 2026}

\begin{document}
\maketitle
\vspace{-1em}

\begin{abstract}
A common empirical strategy in triple-differences (DDD) is to include either covariate trends or covariate levels to a three-way fixed effects (3WFE) regression. This strategy is typically motivated by the conditional parallel gaps assumption which assumes that deviations from parallel trends are similar among units with comparable observed covariates. We formally study both 3WFE specifications and show that, in general, neither consistently estimates the average treatment effect on the treated (ATT). Our diagnosis has two parts. First, we show that the OLS estimands of both specifications fail to satisfy the covariate balancing condition that is sufficient for them to equal the ATT. Second, we characterize the additional assumptions for each OLS estimand to equal the ATT. To address these limitations, we propose a matched triple-differences framework that accommodates a general class of matching procedures, including nearest-neighbor matching and kernel matching. Within this framework, we construct a class of consistent matching estimators, establish their asymptotic normality, and provide a consistent variance estimator that accounts for the variability introduced by the matching step. Our empirical application shows that the proposed matched triple-differences estimator can yield estimates and qualitative conclusions that differ from those obtained using the 3WFE regression.
\end{abstract} \vspace{1em}
\newpage    	
\bigskip
\noindent 

\section{Introduction}
Difference-in-differences (DiD) is one of the most widely used causal inference methods in empirical economics, and its identification relies on the parallel trends assumption. In practice, researchers may worry that this assumption is too restrictive to provide a reasonable approximation to reality. This concern is particularly relevant when treatment requires a unit to satisfy two criteria simultaneously: being geographically exposed to a policy and belonging to the subpopulation eligible for it. In a DiD comparing the treated subgroup (units that are both exposed and eligible) with the unexposed eligible subgroup, geography-specific shocks may violate parallel trends (top two boxes in Figure \ref{figure:foursubgroupintro}). Similarly, subgroup-specific shocks may invalidate a DiD comparing the treated subgroup with the exposed ineligible subgroup (left two boxes in Figure \ref{figure:foursubgroupintro}). Moreover, both types of shocks may affect units differently depending on their characteristics. These concerns motivate triple-differences (DDD) designs with covariates. Their identifying assumption, the conditional parallel gaps assumption, accommodates both geography- and subgroup-specific shocks and allows the effects of these shocks to vary with covariates.

To implement these designs, researchers typically incorporate covariates through one of two three-way fixed effects (3WFE) regression specifications: one includes covariate-specific trends, and the other includes covariate levels. We refer to these as covariate-trend and covariate-level 3WFE regressions, respectively. These specifications are widely used in empirical work: Table \ref{table:ddd_literature} lists $18$ papers published between $2010$ and $2026$ in top economics journals (e.g., AER and QJE) that use at least one of the two specifications. Yet these specifications have received limited attention in the methodological and theoretical literature. In this paper, we formally analyze both specifications, characterize their limitations, and propose valid alternative procedures for estimation and inference based on matching.

\begin{figure}[htbp]
\centering
\begin{tikzpicture}[scale=0.9, transform shape,
    box/.style={draw, rounded corners, minimum width=4cm, minimum height=1.5cm,
                align=center, font=\small},
    treated/.style={box, very thick, fill=gray!15, font=\small\bfseries},
    lab/.style={font=\small\bfseries},
    did/.style={<->, >=stealth, thick, shorten <=2pt, shorten >=2pt},
    shock/.style={font=\footnotesize\itshape, text=red!70!black, align=center}
]
\node[lab] at (-3.8, 3.05) {Exposed};
\node[lab] at ( 3.8, 3.05) {Unexposed};
\node[lab, rotate=90] at (-6.5,  1.85) {Eligible};
\node[lab, rotate=90] at (-6.5, -1.85) {Ineligible};

\node[treated] (A) at (-3.8,  1.85) {Exposed Eligible\\(Treated)};
\node[box]     (B) at ( 3.8,  1.85) {Unexposed Eligible};
\node[box]     (C) at (-3.8, -1.85) {Exposed Ineligible};
\node[box]     (D) at ( 3.8, -1.85) {Unexposed Ineligible};

\draw[did] (A) -- node[above, lab]  {DiD 1}
                  node[below, shock] {geography shock} (B);

\draw[did] (A) -- node[left, lab]   {DiD 2}
                  node[right, shock] {subgroup shock} (C);
\end{tikzpicture}
\caption{Limitations of the two DiD comparisons}
\label{figure:foursubgroupintro}
\end{figure}
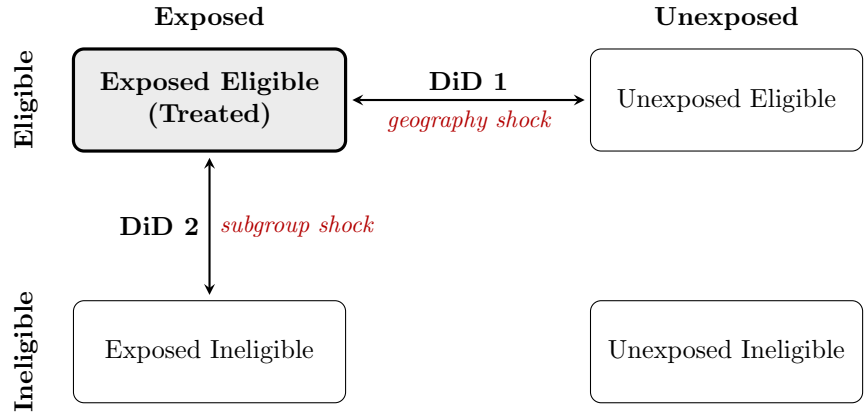

Throughout, we consider a triple-differences design with multiple periods and simultaneous adoption of treatment, and the parameters of interest are the average treatment effects on the treated subgroup (ATT) in each post-treatment period. Our starting point is a weighting representation of the identification result: the ATT can be expressed as a weighting estimand in which the treated subgroup is left unweighted and each of the three untreated subgroups is reweighted by a density ratio. These density ratios align the covariate distribution of each untreated subgroup with that of the treated subgroup. This representation implies a sufficient covariate balancing condition: a weighting estimand equals the ATT if its weighting profile (i)~balances the covariate distribution of \emph{each} untreated subgroup with that of the treated subgroup, and (ii)~attains this balance while leaving the treated subgroup unweighted.

We then characterize the limitations of the commonly used covariate-trend and covariate-level 3WFE regressions. For the covariate-trend specification, the corresponding OLS estimand can be expressed as a weighting estimand that reweights all four subgroups with weights proportional to linear projection residuals, known in the literature as ``implied regression weights" \citep{Chatt_Zubi_2023_biometrika}. This weighting profile fails to satisfy the covariate balancing condition in three repects. First, the pooling issue, it targets covariate balance between the treated subgroup and a pooled untreated group merging the three untreated subgroups, rather than balancing each untreated subgroup separately with the treated subgroup. Second, the mean-only balance issue, the implied regression weights attain only mean balance rather than balance of the full covariate distribution: two subgroups can still have different covariate distributions even when their covariate means are balanced. Third, the treated-reweighting issue, these weights also reweight the treated subgroup. The covariate-trend OLS estimand equals the ATT if treatment effects are homogeneous in the covariates and, in addition, either the conditional mean functions of the three untreated subgroups are linear with a common slope or the generalized propensity scores take a specific inverse-linear form.

For the covariate-level specification, the corresponding OLS estimand equals the unconditional triple-differences estimand, which assigns a weight of one to all four subgroups. This weighting profile also fails to satisfy the covariate balancing condition because it leaves all three untreated subgroups unweighted. The covariate-level OLS estimand equals the ATT if and only if the parallel gaps assumption holds unconditionally, in which case there is no need to include covariates in the triple-differences design. 

To address these limitations, we propose a matched triple-differences framework for valid estimation and inference. The estimation part of our framework is a two-step procedure. The first step conducts three pairwise matchings, each between the treated subgroup and one of the three untreated subgroups. The second step estimates a weighted 3WFE regression using the weights generated in the matching step. The resulting matched triple-differences estimator is designed to satisfy the covariate balancing condition: matching reweights each untreated subgroup toward the covariate distribution of the treated subgroup while leaving the treated subgroup unweighted. The inference part of our framework is based on a set of novel regularity conditions on the matching procedures. These conditions allow us to abstract from any specific matching method and establish asymptotic theory for a general class of matching estimators. Notably, this class includes nearest-neighbor matching, with either a fixed or a diverging number of neighbors, as well as kernel matching. To account for the uncertainty introduced by the matching step, we propose a unified variance estimator that is consistent for every matching procedure covered by our framework. 

We illustrate our results with simulations and an empirical application. In our empirical application, we revisit \cite{Cai_2016_AEJ}, who studies the effect of a tobacco insurance program on the borrowing and saving behavior of the treated households. \cite{Cai_2016_AEJ} uses triple-differences as the main empirical design, and her main regression specification is similar to the covariate-level 3WFE regression. We find that our matched triple-differences estimates can differ in magnitude from those produced by her regression specification and lead to qualitatively different results for some outcomes of interest. Specifically, the estimates from her regression specification imply that the program’s effect on net saving switches sign over time (positive in years $0$ and $3$ and negative in year $5$), whereas our matched triple-differences estimates show no
significant effect in any post-treatment year.

\subsection{Related Literature and Contribution}
This paper contributes to four strands of literature. First, it contributes to the literature on triple-differences designs and the 3WFE method. Triple-differences was first proposed by \citet{Gruber_1994_AER}. \citet{Olden_Moen_2022_EJ} formalize the triple-differences estimator in a two-period setting without covariates. Recent work has begun to examine the validity of 3WFE methods and to study triple-differences designs with covariates. Three closely related papers are \citet{Strezhnev_2023_wp}, \citet{Leventer_2025_wp}, and \citet{Ortiz-Villavicencio_SantAnna_2025_working}. \citet{Strezhnev_2023_wp} studies the limitations of 3WFE regressions \emph{without} covariates under staggered adoption, where identification relies on an unconditional parallel gaps assumption. Our paper instead focuses on 3WFE regressions with covariates, where identification relies on a conditional parallel gaps assumption. We show that these regressions can fail even under simultaneous adoption, a setting in which 3WFE without covariates is valid under the unconditional parallel gaps assumption.

Both \citet{Leventer_2025_wp} and \citet{Ortiz-Villavicencio_SantAnna_2025_working} study triple-differences designs with covariates. \citet{Leventer_2025_wp} examines an empirical practice distinct from the 3WFE approach studied in our paper. Specifically, the paper shows that treating triple differences with covariates as the difference between two conditional DiD comparisons is invalid. \citet{Ortiz-Villavicencio_SantAnna_2025_working} also examine the practice of differencing two conditional DiD comparisons. In addition, they provide Monte Carlo evidence that the two 3WFE specifications considered in our paper can be biased in a two-period setting. Our paper differs from and complements these studies in two respects. First, we provide an analytical diagnosis that characterizes the distinct limitations of the two commonly used 3WFE specifications with covariates. This diagnosis explains the specific reasons why each specification can fail and characterizes the restrictions under which it is valid, thereby complementing the simulation evidence in \citet{Ortiz-Villavicencio_SantAnna_2025_working}. Second, the estimators proposed by \citet{Leventer_2025_wp} and \citet{Ortiz-Villavicencio_SantAnna_2025_working} are based on inverse-probability weighting and doubly robust estimation, whereas our proposed estimator is based on fully nonparametric matching methods.

Second, this paper contributes to the literature on cross-sectional matching methods. Earlier work focuses on asymptotic theory and inference for specific matching methods, such as nearest-neighbor matching (\citeinb{Abadie-Imbens_2006_ECMA}; \citeinb{Lin-Ding-Han_2023_ECMA}; \citeinb{Abadie-Spiess_2022_JASA}) and kernel matching (\citeinb{Heckman-etal_1998_ECMA}). More recently, \cite{Lin_Han_2025_JOE} and \cite{Meng_et_al_2026_wp} have developed general asymptotic theory and inference procedures that cover a class of matching methods. As in these two papers, our asymptotic framework also abstracts from a specific matching method but differs from them in two aspects. First, our framework is based on regularity conditions on the matching procedure that are non-nested with those in \cite{Lin_Han_2025_JOE} and \cite{Meng_et_al_2026_wp}. Second, we develop a different variance estimator. Our variance estimator builds on that of \cite{Abadie-Imbens_2006_ECMA} but extends it to the triple-differences setting and to a class of matching procedures, which requires a different proof strategy.

Third, this paper adds to the line of work that combines matching with panel data, which has primarily focused on the difference-in-differences setting. \citet{Heckman-Ichimura-Todd_1997_ReStud} propose a difference-in-differences matching estimator for a two-period setting based on kernel matching and local linear regression; \citet{Liu_Vazquez_Bare_2026_wp} study nearest-neighbor matching combined with difference-in-differences designs in settings with multiple periods and staggered adoption. We extend this literature to triple-differences designs: we lay out how to combine matching with triple-differences designs and develop asymptotic theory and inference for a class of matching methods.

Finally, our work is broadly related to the literature on interpreting OLS estimands in regression specifications with covariates \citep{Sloczynski_2022_Restat,Sloczynski_2026_Restud,Blandol_et_al_2026_Restud}. Specifically, the implied regression weights of the covariate-trend specification are closely related to those studied by \citet{Chatt_Zubi_2023_biometrika} in the cross-sectional setting and \citet{Caetano-Callaway_2026_wp} in the difference-in-differences setting. However, we show that the implied regression weights involve a pooling issue that has no counterpart in either the cross-sectional or the simultaneous-adoption DiD setting.
 
The remainder of the paper is organized as follows. Section \ref{section:econometricssetup} presents the motivating
example, introduces the setup and notation, and establishes the identification result and the covariate balancing condition. Section \ref{section:empiricaldiagnose} diagnoses the two 3WFE regression specifications commonly used in empirical practice. Section \ref{section:matchedtriplediffereces} introduces the matched triple-differences framework, establishes the consistency and asymptotic normality of the resulting estimators, and develops the variance estimator for inference. Section \ref{section:simulation} presents simulation studies, Section \ref{section:empiricalapplication} illustrates our methods by revisiting \citet{Cai_2016_AEJ}, and Section \ref{section:conclusion} concludes.

\section{Setup, Estimand, and Identification}\label{section:econometricssetup}
\subsection{Motivating Example}
As a motivating example, we use the study by \citet{Cai_2016_AEJ} to illustrate a triple-differences design and explain how conditioning on covariates can make its identifying assumption more credible. \citet{Cai_2016_AEJ} uses household-level panel data from $12$ tobacco-producing counties in Jiangxi Province, China, to study how the provision of a tobacco insurance program affects rural households' borrowing and saving behavior. In these counties, some households (hereafter, tobacco households) rely on tobacco as a main source of income, but tobacco production is vulnerable to adverse weather events. In 2003, one county (hereafter, the exposed county) collaborated with the People's Insurance Company of China (PICC) to design and launch a tobacco production insurance program that protected tobacco households against losses from weather-related disasters.

The key feature of this program is that a household is treated if and only if it satisfies two criteria simultaneously: exposure and eligibility. A household is exposed if it resides in the county offering the insurance program, and eligible if it grows tobacco. Households outside the exposed county and non-tobacco households are untreated. 

To estimate the effect of the insurance program on tobacco households in the exposed county, a triple-differences design is more appealing than a difference-in-differences design because it accounts for both county- and industry-specific shocks, such as changes in local economic conditions and in the tobacco market, respectively. To see this, consider a difference-in-differences comparison of tobacco households inside and outside the exposed county. This comparison accounts for industry-specific shocks common to tobacco households across counties. However, it may be confounded by county-specific shocks that affect households' financial decisions differently in the exposed and unexposed counties, in which case the parallel trends assumption fails. A triple-differences design adds a second difference-in-differences comparison, between non-tobacco households inside and outside the exposed county, to capture these county-specific shocks. Differencing the two comparisons removes this bias caused by the violation of the parallel trends assumption\footnote{Alternatively, one could start from a difference-in-differences comparison of tobacco and non-tobacco households within the exposed county, which accounts for county-specific shocks. This comparison, however, may be confounded by industry-specific shocks. The triple-differences design removes this bias symmetrically, by differencing out the analogous comparison in the unexposed counties.}. Figure \ref{figure:fourhouseholds} illustrates the two difference-in-differences comparisons underlying this triple-differences design.

Moreover, county- and industry-specific shocks are likely to affect households differently depending on their characteristics. For example, households of different sizes may face different budget constraints and therefore adjust their borrowing and saving differently in response to the same shock. Consequently, deviations from parallel trends may depend on household covariates. In this setting, conditioning on covariates makes the triple-differences identification strategy more credible, because households with similar covariates tend to experience similar deviations from parallel trends.

\begin{figure}[htbp]
\centering
\begin{tikzpicture}[
    box/.style={
        draw,
        rounded corners,
        minimum width=4.2cm,
        minimum height=1.8cm,
        align=center,
        font=\small
    },
    treated/.style={
        box,
        very thick,
        fill=gray!15
    }
]

\node[font=\small\bfseries] at (-3.5,3.25) {Exposed};
\node[font=\small\bfseries] at ( 3.5,3.25) {Unexposed};

\node[font=\small\bfseries, rotate=90] at (-7.0, 1.5) {Tobacco};
\node[font=\small\bfseries, rotate=90] at (-7.0,-1.5) {Non-tobacco};

\node[treated] (A) at (-3.5,1.5)
    {\textcolor{red}{\textbf{Exposed Tobacco}}\\
     \textcolor{red}{\textbf{(Treated)}}};
\node[box] (B) at (3.5,1.5)
    {Unexposed Tobacco\\ (Untreated)};
\node[box] (C) at (-3.5,-1.5)
    {Exposed Non-tobacco\\ (Untreated)};
\node[box] (D) at (3.5,-1.5)
    {Unexposed Non-tobacco\\ (Untreated)};

\draw[<->, thick]
    (A.east) -- node[above, font=\small]
    {$\mathbf{DiD_1}$} (B.west);

\draw[<->, thick]
    (C.east) -- node[below, font=\small]
    {$\mathbf{DiD_2}$} (D.west);

\node[align=center, font=\small] at (0,0)
    {$\mathbf{DDD}=\mathbf{DiD_1}-\mathbf{DiD_2}$};

\end{tikzpicture}
\caption{Four types of households and the two DiD comparisons}
\label{figure:fourhouseholds}
\end{figure}
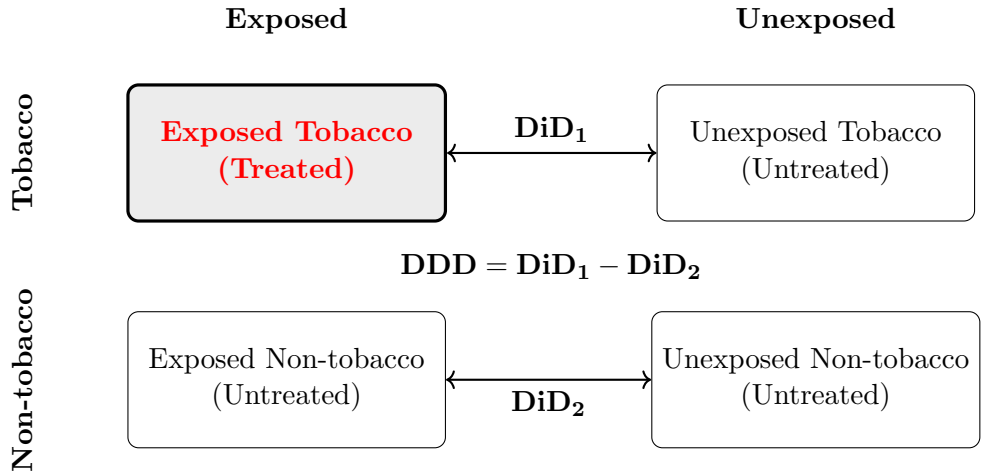

\newpage
\noindent
\subsection{Econometric Setup and Identification}
Consider a balanced panel that follows $n$ units over $T$ periods, $t=1,\dots,T$. The treatment is implemented simultaneously in a single period $t^*$, with $2\leq t^*\leq T$, so that $t<t^*$ are the pre-treatment periods and $t\geq t^*$ are the post-treatment periods. Each unit is characterized by two time-invariant binary indicators: an exposure indicator \(G\) and an eligibility indicator \(S\). Specifically,
\begin{itemize}
    \item the \emph{exposure} indicator \(G \in \{0,1\}\) indicates whether the unit belongs to a group (e.g., a county) that is exposed to the treatment in the post-treatment periods;
    \item the \emph{eligibility} indicator \(S \in \{0,1\}\) indicates whether the unit belongs to a subpopulation (e.g., households growing tobacco) that is eligible for the treatment.
\end{itemize}
These two indicators partition the population into four cross-sectional subgroups displayed in Figure \ref{figure:foursubgroup}. The key feature of the triple-differences design is that only units satisfying both criteria are treated: a unit must belong to the exposed group and to the eligible subpopulation. The treatment is an absorbing state: once a unit gets treated, it remains treated through period $T$. Let $D\in\{0,1\}$ indicate whether the unit is treated in the post-treatment periods, and let \(D_{t} \in \{0,1\}\) indicate whether the unit is treated in period $t$, so that
\[D=GS, \quad D_{t}=GS\1(t\geq t^*)   \]
In Figure~\ref{figure:foursubgroup}, the treated subgroup is the top-left cell, \((G=1,S=1)\); we refer to the remaining three cells as the \emph{untreated subgroups}.

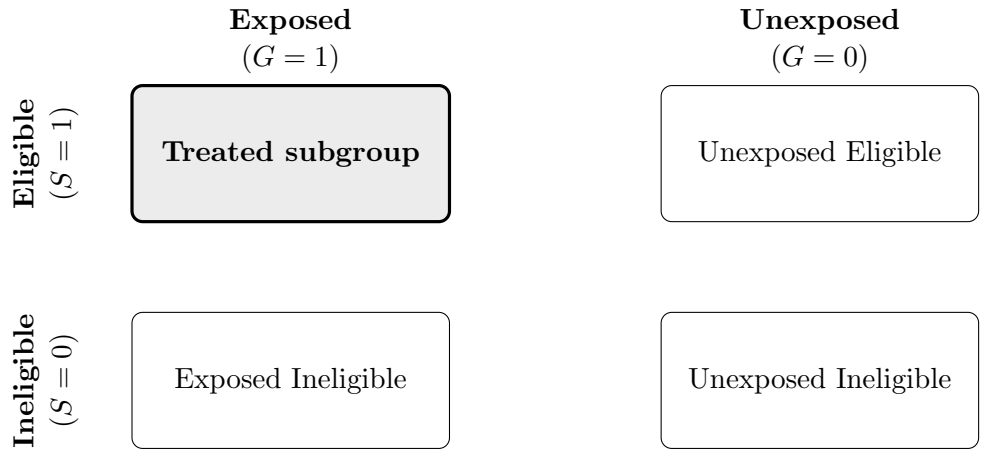
\begin{figure}[htbp]
\centering
\begin{tikzpicture}[
    box/.style={
        draw,
        rounded corners,
        minimum width=4.2cm,
        minimum height=1.8cm,
        align=center,
        font=\small
    },
    treated/.style={
        box,
        very thick,
        fill=gray!15
    }
]

\node[font=\small\bfseries] at (-3.5,3.25) {Exposed};
\node[font=\small] at (-3.5,2.75) {\((G=1)\)};

\node[font=\small\bfseries] at (3.5,3.25) {Unexposed};
\node[font=\small] at (3.5,2.75) {\((G=0)\)};

\node[font=\small\bfseries, rotate=90] at (-7.0,1.5) {Eligible};
\node[font=\small, rotate=90] at (-6.5,1.5) {\((S=1)\)};

\node[font=\small\bfseries, rotate=90] at (-7.0,-1.5) {Ineligible};
\node[font=\small, rotate=90] at (-6.5,-1.5) {\((S=0)\)};

\node[treated] (A) at (-3.5,1.5) {\textbf{Treated subgroup}};
\node[box]     (B) at ( 3.5,1.5) {Unexposed Eligible};
\node[box]     (C) at (-3.5,-1.5) {Exposed Ineligible};
\node[box]     (D) at ( 3.5,-1.5) {Unexposed Ineligible};

\end{tikzpicture}
\caption{Four cross-sectional subgroups}
\label{figure:foursubgroup}
\end{figure}

We adopt the potential-outcome framework in which the potential outcome depends on the entire treatment sequence. Let $Y_{t}(d_{1},\dots,d_{T})$ denote the potential outcome in period $t$ under treatment sequence $(d_{1},\dots,d_{T})$, and let $\mathbf{0}_s$ and $\mathbf{1}_s$ denote $s$-dimensional vectors of zeros and ones, respectively. Because adoption is simultaneous and the treatment is absorbing, only two treatment paths can occur in our setting: treated units follow $(\mathbf{0}_{t^*-1},\mathbf{1}_{T-t^*+1})$ and untreated units follow $\mathbf{0}_{T}$. We therefore abbreviate
\[ Y_{t}(1)\equiv Y_{t}(\mathbf{0}_{t^*-1},\mathbf{1}_{T-t^*+1}), \quad  Y_{t}(0)\equiv Y_{t}(\mathbf{0}_{T})  \]
 As is common in the causal inference literature on panel data (\citeinb{Chaise-Dhaut_2022_EJ}; \citeinb{Roth-etal_2023_JoE}), we assume SUTVA and consistency to link the observed and potential outcomes. As a result,
\begin{align*}
    Y_{t}= Y_{t}(1)D_t+Y_{t}(0)(1-D_t)
\end{align*}
Our parameters of interest are the average treatment effects on the treated subgroup (ATT) $l$ periods after the initial treatment period,
\[
\tau_l=\E\!\left[Y_{t^*+l}(1)-Y_{t^*+l}(0)\mid D=1\right] = \E\!\left[Y_{t^*+l}(1)-Y_{t^*+l}(0)\mid G=1,\ S=1\right], \quad l=0,1,\dots,T-t^*
\]

To identify $\tau_l$, we make the following three assumptions: 
\begin{assumption}\label{assumption:noanticipation}
    (No Anticipation) For all $t=1,\dots,T,$ $Y_{t}(d_1,\dots,d_T)=Y_{t}(d_1,\dots,d_t)$
\end{assumption}
The no-anticipation assumption requires that future treatments do not affect current outcomes: the potential outcome in period $t$ depends only on the treatment path up to period $t$.  

Let $X \in \mathbbm{R}^q$ denote a vector of time-invariant, pre-treatment covariates of dimension $q$.
\begin{assumption} \label{assumption:overlap}
    (Strong Overlap) There exists \(\epsilon>0\) such that    
        \[
        \Pr(G=g,S=s\mid X=x)>\epsilon
        \]
        for every \((g,s)\in\{0,1\}^2\) and almost every \(x\in \mathcal{X}\), where \(\mathcal X\) denotes the support of \(X\).
\end{assumption}
The strong overlap assumption requires that, at almost every covariate value $x$, each of the four subgroups occurs with probability bounded away from zero: within each covariate stratum, the population contains units from all four subgroups of Figure~\ref{figure:foursubgroup}.\footnote{Assumption \ref{assumption:overlap} is stronger than what identification of $\tau_l$ requires: it would suffice to impose overlap almost everywhere on the support of $X$ conditional on the treated subgroup, rather than on the whole support $\mathcal{X}$. We maintain the stronger version for simplicity. }

The key identifying assumption of the triple-differences design replaces the familiar parallel trends assumption of the difference-in-differences designs with a parallel gaps assumption, following the recent triple-differences literature (\citeinb{Olden_Moen_2022_EJ}; \citeinb{Ortiz-Villavicencio_SantAnna_2025_working}).

\begin{assumption}\label{assumption:parallelgap}
    (Conditional Parallel Gaps) For all $t=2,\dots,T$,
    \begin{align*}
    &\E[Y_{t}(0)-Y_{t-1}(0)\mid G=1,\ S=1,\ X] - \E[Y_{t}(0)-Y_{t-1}(0)\mid G=0,\ S=1,\ X]\\
     &=   \E[Y_{t}(0)-Y_{t-1}(0)\mid G=1,\ S=0,\ X] - \E[Y_{t}(0)-Y_{t-1}(0)\mid G=0,\ S=0,\ X],\ a.s.
    \end{align*}
\end{assumption}
Figure~\ref{figure:ddd_parallel_gaps} illustrates the intuition behind the conditional parallel gaps assumption. The assumption allows each of the two difference-in-differences comparisons used in the triple-differences design to violate the conditional parallel trends assumption, and allows the size of this violation to vary with the covariates $X$. It requires only that, at each covariate value $X=x$, the violation be the same in both comparisons. In our motivating example, the conditional parallel gaps assumption allows for both county- and industry-specific shocks, and allows these shocks to affect households' borrowing and saving behavior differently depending on their characteristics. It requires, however, that county-specific shocks generate the same deviation from parallel trends for tobacco and non-tobacco households with similar covariates, or, equivalently, that industry-specific shocks generate the same deviation from parallel trends for households inside and outside the exposed county with similar covariates.
 
\begin{figure}[htbp]
\centering
\begin{subfigure}[t]{0.48\textwidth}
\centering
\caption*{\large\bfseries Eligible $(S=1)$}
\begin{tikzpicture}[x=0.6cm,y=0.7cm]
  \draw[axis] (0,0) -- (8.6,0) node[right] {\scriptsize Time};
  \draw[axis] (0,0) -- (0,6.2) node[above] {\scriptsize Outcome};
  \draw (1,-0.08) -- (1,0.08) node[below=3pt, font=\scriptsize] {Pre};
  \draw (5,-0.08) -- (5,0.08) node[below=3pt, font=\scriptsize] {Post};
  \coordinate (C0)  at (1,1.2);
  \coordinate (C1)  at (5,2.4);
  \coordinate (T0)  at (1,3.0);
  \coordinate (TB1) at (5,4.2);
  \coordinate (TA1) at (5,5.7);
  \draw[control]   (C0) -- (C1);
  \draw[treated]   (T0) -- (TA1);
  \draw[benchmark] (T0) -- (TB1);
  \fill[blue]         (C0)  circle (1.5pt);
  \fill[blue]         (C1)  circle (1.5pt);
  \fill[red]          (T0)  circle (1.5pt);
  \fill[red]          (TA1) circle (1.5pt);
  \fill[red!75!black] (TB1) circle (1.3pt);
  \draw[bias] ($(TB1)+(0.35,0.05)$) -- ($(TA1)+(0.35,-0.05)$);
  \node[font=\scriptsize, inner sep=1pt, anchor=west]
    at ($(TB1)!0.5!(TA1)+(0.55,0)$) {$\delta_1(x)$};
  \node[anchor=west, font=\scriptsize] at ($(TA1)+(0.55,0)$) {Exposed $(G{=}1,S{=}1)$};
  \node[anchor=west, font=\scriptsize] at ($(TB1)+(0.55,0)$) {Trend under CPT};
  \node[anchor=west, font=\scriptsize] at ($(C1) +(0.55,0)$) {Unexposed $(G{=}0,S{=}1)$};
\end{tikzpicture}
\end{subfigure}
\hfill
\begin{subfigure}[t]{0.48\textwidth}
\centering
\caption*{\large\bfseries Ineligible $(S=0)$}
\begin{tikzpicture}[x=0.6cm,y=0.7cm]
  \draw[axis] (0,0) -- (8.6,0) node[right] {\scriptsize Time};
  \draw[axis] (0,0) -- (0,6.2) node[above] {\scriptsize Outcome};
  \draw (1,-0.08) -- (1,0.08) node[below=3pt, font=\scriptsize] {Pre};
  \draw (5,-0.08) -- (5,0.08) node[below=3pt, font=\scriptsize] {Post};
  \coordinate (C0)  at (1,1.0);
  \coordinate (C1)  at (5,1.5);
  \coordinate (T0)  at (1,2.5);
  \coordinate (TB1) at (5,3.0);
  \coordinate (TA1) at (5,4.5);
  \draw[control]   (C0) -- (C1);
  \draw[treated]   (T0) -- (TA1);
  \draw[benchmark] (T0) -- (TB1);
  \fill[blue]         (C0)  circle (1.5pt);
  \fill[blue]         (C1)  circle (1.5pt);
  \fill[red]          (T0)  circle (1.5pt);
  \fill[red]          (TA1) circle (1.5pt);
  \fill[red!75!black] (TB1) circle (1.3pt);
  \draw[bias] ($(TB1)+(0.35,0.05)$) -- ($(TA1)+(0.35,-0.05)$);
  \node[font=\scriptsize, inner sep=1pt, anchor=west]
    at ($(TB1)!0.5!(TA1)+(0.55,0)$) {$\delta_0(x)$};
  \node[anchor=west, font=\scriptsize] at ($(TA1)+(0.55,0)$) {Exposed $(G{=}1,S{=}0)$};
  \node[anchor=west, font=\scriptsize] at ($(TB1)+(0.55,0)$) {Trend under CPT};
  \node[anchor=west, font=\scriptsize] at ($(C1) +(0.55,0)$) {Unexposed $(G{=}0,S{=}0)$};
\end{tikzpicture}
\end{subfigure}
\caption{Conditional parallel gaps}
\label{figure:ddd_parallel_gaps}
\vspace{0.3em}
\begin{minipage}{0.95\textwidth}
\footnotesize
\textit{Note:} This figure illustrates the conditional parallel gaps assumption at a fixed covariate value \(X=x\). In each panel, the dotted line is the trend the exposed units would follow if conditional parallel trends (CPT) held. In the left panel, \(\delta_1(x)\) is the deviation from conditional parallel trends in untreated outcomes between eligible exposed units, \((G=1,S=1)\), and eligible unexposed units, \((G=0,S=1)\). In the right panel, \(\delta_0(x)\) is the corresponding deviation between ineligible exposed units, \((G=1,S=0)\), and ineligible unexposed units, \((G=0,S=0)\). The assumption allows conditional parallel trends to be violated in the two DiD comparisons used in the triple-differences design, but requires the violation to be the same across the two DiD comparisons: \(\delta_1(x)=\delta_0(x)\).
\end{minipage}
\end{figure}
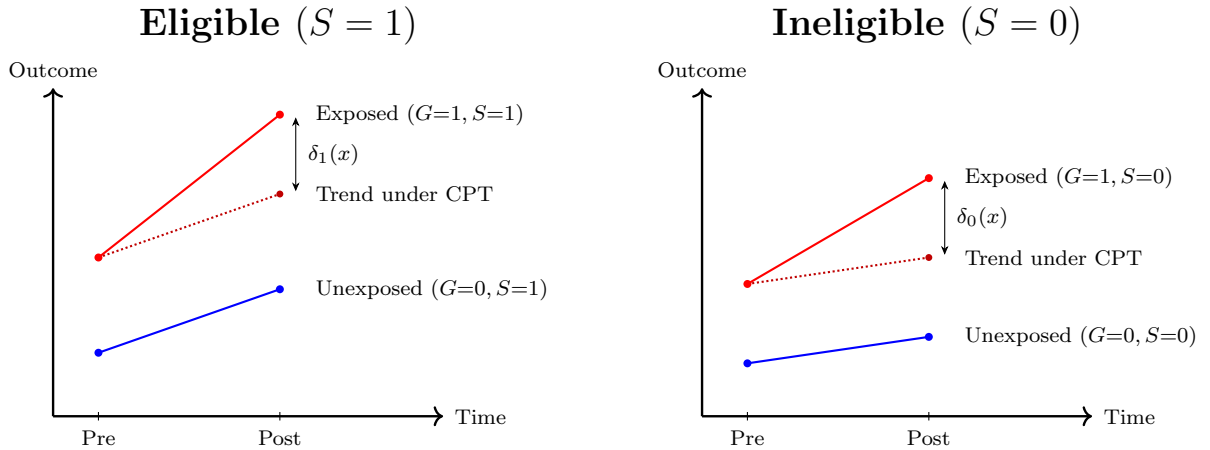
Although the conditional parallel gaps assumption allows violations of conditional parallel trends, we emphasize that the former is not a weaker version of the latter. Mathematically, neither condition implies the other. A triple-differences design should therefore not be interpreted as a mechanical robustness check on a DiD design: the two designs rest on distinct identifying assumptions, and their credibility depends on the empirical context. Introducing the second difference-in-differences comparison improves credibility when the second comparison is subject to the same parallel-trends violation as the first, but it can introduce bias conditional parallel trends hold but conditional parallel gaps fail.

\begin{theorem}
\label{theorem:identification}
Under Assumptions \ref{assumption:noanticipation},
\ref{assumption:overlap}, and \ref{assumption:parallelgap}, $\tau_l$ is identified for every $l=0,1,\dots,T-t^*$,
\begin{equation*}
\begin{aligned}
\tau_l
&=
\E[\Delta Y_l\mid G=1,S=1] 
-\E\!\left[
    \frac{f_{11}(X)}{f_{01}(X)}
    \Delta Y_l
    \mid G=0,S=1
\right] \\
&\quad
-\E\!\left[
    \frac{f_{11}(X)}{f_{10}(X)}
    \Delta Y_l
    \mid G=1,S=0
\right] 
+\E\!\left[
    \frac{f_{11}(X)}{f_{00}(X)}
    \Delta Y_l
    \mid G=0,S=0
\right] ,
\end{aligned}
\end{equation*}
where $\Delta Y_l\equiv Y_{t^*+l}-Y_{t^*-1}$ and \(f_{gs}(x)\) denotes the conditional density of \(X\) given
$(G,S)=(g,s)$.
\end{theorem}
Theorem \ref{theorem:identification} gives a weighting representation of $\tau_l$. The treated subgroup is left unweighted, and each untreated subgroup $(g,s)$ is reweighted by the density ratio $f_{11}(X)/f_{gs}(X)$. The density ratio is a balancing weight: it aligns the covariate distribution of the untreated subgroup $(g,s)$ with that of the treated subgroup. The identification result thus pins down a balancing weight profile:
\begin{equation}
\label{eq:balancing-weights}
w_{11}(X)=1,
\qquad
w_{gs}(X)=\frac{f_{11}(X)}{f_{gs}(X)}
\quad\text{for }(g,s)\in\{(0,1),(1,0),(0,0)\}.
\end{equation}
This profile yields a sufficient condition for a weighting estimand to equal the ATT, which we call the
\emph{covariate balancing condition}: the estimand’s weighting profile (i)balances the covariate distribution of each untreated subgroup with that ofthe treated subgroup, and (ii) attains this balance while leaving the treated subgroup unweighted. Consequently, an estimator that converges to such a weighting estimand is consistent for the ATT. This condition plays a key role in both our diagnosis of empirical practice in Section \ref{section:empiricaldiagnose} and our construction of the matching estimator in Section \ref{section:matchedtriplediffereces}.

\section{Diagnosis of Empirical Practice}\label{section:empiricaldiagnose}

\noindent
In this section, we diagnose the two leading empirical practices in the applied literature. Consider a sample $\{Y_{i1},Y_{i2},\dots, Y_{iT}, G_i,S_i,X_i\}_{i=1}^n$, drawn independently from the same distribution.
\begin{assumption}\label{assumption:iid}
    (Random Sampling) $\{Y_{i1},Y_{i2},\dots, Y_{iT}, G_i,S_i,X_i\}_{i=1}^n$ is a sample of independent and identically distributed observations from a population $(Y_{1},Y_{2},\dots,Y_{T},G,S,X)$.
\end{assumption}
As documented in our survey of $18$ empirical papers (Table \ref{table:ddd_literature}), empirical practice in the literature largely follows two types of 3WFE regression specifications. In the first specification, which we call the covariate-trend 3WFE regression, the covariates $X_i$ enter with covariate specific trends, allowing for time varying coefficients:
\begin{equation}
\label{eq:spec1}
Y_{it}
= \lambda_{G_iS_i}
+ \lambda_{G_i t}
+ \lambda_{S_i t}
+ \sum_{l\neq -1}
    \tau_l^{CT}\, G_iS_i\1(t-t^*=l)
+ \sum_{r=1}^T X_i'\theta_r\1(t=r)
+ u_{it}.
\end{equation}
In the second specification, which we call the covariate-level 3WFE regression, the covariates $X_i$ enter only in  levels, with time-invariant coefficients:
\begin{equation}
\label{eq:spec2}
Y_{it}
= \lambda_{G_iS_i}
+ \lambda_{G_i t}
+ \lambda_{S_i t}
+ \sum_{l\neq -1}
    \tau_l^{CL}\, G_iS_i\1(t-t^*=l)
+ X_i'\theta
+ u_{it}.
\end{equation}
In both specifications, $\lambda_{G_iS_i}$, $\lambda_{G_i t}$, and $\lambda_{S_i t}$ are the subgroup, exposure-by-time, and eligibility-by-time fixed effects, respectively. The OLS coefficients $\widehat \tau_l^{CT}$ and $\widehat \tau_l^{CL}$ on the triple interaction term $ G_iS_i\1(t-t^*=l)$ are often interpreted as estimates of $\tau_l$, the ATT $l$ periods after the treatment. In a balanced panel, replacing the subgroup fixed effect $\lambda_{G_iS_i}$ with an individual fixed effect $\lambda_i$ yields identical OLS estimates in both the covariate-trend and covariate-level specifications. Because more than two-thirds the papers in our survey use subgroup fixed effects, we focus our diagnosis on the specifications with subgroup fixed effects; since the estimates are identical, the diagnosis below also applies to the specifications with individual fixed effects.

For each specification, our diagnosis consists of two parts. First, we rewrite the corresponding OLS estimand as a weighting estimand and determine whether its weighting profile satisfies the covariate balancing condition. Second, if it does not, we characterize the additional assumptions under which the OLS estimand equals the ATT. Our results show that common empirical practice generally fails to consistently estimate the ATT: the induced weighting estimands do not correctly align the covariate distributions of the untreated subgroups with that of the treated subgroup, and they require additional restrictive assumptions to equal the ATT.

\subsection{Covariate-trend 3WFE Regression}
Recall that $D=GS$ is the treatment indicator. Let $L(D\mid G,S,X)$ denote the population linear projection of $D$ onto $(1,G,S,X)$, let $\kappa\equiv\E[(D-L(D\mid G,S,X))^2]$ denote the mean squared projection residual, and let $p_{gs}=P(G=g,S=s)$ denote the subgroup proportions. Let $\tau_l(X)=\E\!\left[Y_{t^*+l}(1)-Y_{t^*+l}(0)\mid G=1,\ S=1,X\right]  $ denote the conditional ATT, and let $ L_0(\Delta Y_l|G,S,X)$ denote the conditional linear projection of $\Delta Y_l$ on $(1,G,S,X)$ in the untreated population $D=0$.
\begin{lemma}\label{lemma:covariatetrend}
\textbf{(I)} Under Assumptions~\ref{assumption:overlap} and \ref{assumption:iid},
\begin{align*}
\widehat \tau_l^{CT} &\xrightarrow{p}\tau_{l}^{CT}\\
&=\E\!\left[\frac{p_{11}\left(1-L(D\mid G,S,X)\right)}{\kappa}\, \Delta Y_l \;\middle|\; G=1,S=1\right]
-\E\!\left[\frac{p_{01}\,L(D\mid G,S,X)}{\kappa}\, \Delta Y_l \;\middle|\; G=0,S=1\right] \\
&\quad-\E\!\left[\frac{p_{10}\,L(D\mid G,S,X)}{\kappa}\, \Delta Y_l \;\middle|\; G=1,S=0\right]
+\E\!\left[\frac{-p_{00}\,L(D\mid G,S,X)}{\kappa}\, \Delta Y_l \;\middle|\; G=0,S=0\right].
\end{align*}
\textbf{(II)} $\tau_l^{CT}=\tau_l$ if the treatment effect is \emph{homogeneous} in $X$, i.e., $\tau_l(X)=\tau_l$ and
\emph{either} of the following holds:

\begin{enumerate}
    \item Linear conditional means with a common slope.
    \[
        \E[\Delta Y_l \mid G=g,S=s,X]
        = L_0(\Delta Y_l|g,s,X) =  \gamma_{gs} + X'\gamma,
        \qquad (g,s)\in\{(0,0),(0,1),(1,0)\}.
    \]
    \item An inverse linear functional form for generalized propensity scores. Define $e_{gs}(X)=P(G=g,S=s|X)$.
    \begin{align*}
        e_{gs}(X)&=\frac{[(-1)^{g+s}L(D|g,s,X)]^{-1}}{\sum_{g=0}^1\sum_{s=0}^1[gs-(-1)^{g+s}L(D|g,s,X)]^{-1}}\\
        &=
        \frac{
                 [            \lambda_{gs} +(-1)^{g+s}X'\lambda ]^{-1}
        }{
            \displaystyle
            \sum_{g=0}^{1}\sum_{s=0}^{1}
             [            \lambda_{gs} +(-1)^{g+s}X'\lambda ]^{-1}
        }.
    \end{align*}
        
\end{enumerate}
\end{lemma}
Part (I) of Lemma \ref{lemma:covariatetrend} shows that the OLS estimator from the covariate-trend 3WFE regression converges to a weighting estimand, in which each of the four subgroups is reweighted by weights proportional to the linear projection residual $D-L(D\mid G,S,X)$. This weighting profile is the triple-differences analogue of the implied regression weights studied by \citet{Chatt_Zubi_2023_biometrika} in cross-sectional settings and by \citet{Caetano-Callaway_2026_wp} in DiD settings. In those two settings, the implied regression weights have an appealing property: they balance the covariate means of the treated and untreated groups. In the triple-differences setting, however, the implied regression weights induced by specification \eqref{eq:spec1} fail the covariate balancing condition in three respects. 

First, they target the wrong balance structure (pooling issue). Lemma~\ref{lemma:regressionweight} in the appendix shows that the implied regression weights balance covariate means between the treated subgroup ($D=1$) and the pooled untreated group ($D=0$). In a triple-differences design this is the wrong balance structure because the untreated group ($D=0$) now pools three distinct untreated subgroups, whereas the covariate balancing condition requires aligning the covariate distribution of \emph{each} of them with the treated subgroup separately. Second, the implied regression weights attain only mean balance rather than the balance of the full covariate distribution required by the covariate balancing conditions (mean only balance issue): Two subgroups can still have different covariate distributions even when their covariate means are balanced. Third, the implied regression weights also reweight the treated subgroup (treated-reweighting issue). The covariate balancing condition leaves the treated subgroup unweighted  while the implied regression weights profile reweights the treated subgroup by the weight $p_{11}(1-L(D\mid G,S,X))/\kappa$.

Part (II) of lemma \ref{lemma:covariatetrend} provides the additional sufficient assumptions for $\tau_l^{CT}$ to equal the ATT. The first is that the treatment effect is homogeneous with respect to the covariates. The second is that either the conditional mean functions of the three untreated subgroups are linear in $X$ with a common slope, differing only in their intercepts, or the generalized propensity scores admit the specific inverse linear functional form. Both requirements are restrictive: the first rules out treatment effect heterogeneity in $X$, and the second imposes a parametric form on either the conditional mean functions or the generalized propensity scores.

\begin{remark}
Specification \eqref{eq:spec1} can be enriched in two ways: (i) by further interacting the covariates with the subgroup indicators $G_i$ and $S_i$, and (ii) by including higher-order terms of the covariates. The first enrichment can attain mean balance between each untreated subgroup and the treated subgroup, addressing the pooling issue. The second enrichment can attain balance of higher-order moments of the covariates, partly addressing the mean-only balance issue. However, the enriched specifications still do not satisfy the covariate balancing condition: first, balance of higher-order moments among the four subgroups does not imply balance of the full covariate distribution, and second, the implied regression weights continue to reweight the treated subgroup. Hence, the enriched specifications do not consistently estimate $\tau_l$ in general.
\end{remark}

\subsection{Covariate-level 3WFE Regression}
\begin{lemma}\label{lemma:covariatelevel}
\textbf{(I)}  Under Assumptions~\ref{assumption:overlap} and \ref{assumption:iid}, 
 \begin{align*}
\widehat\tau_l^{CL}&\xrightarrow{p}\tau_l^{CL}\\ 
&=\E[\Delta Y_l \mid G=1,S=1] -\E[\Delta Y_l \mid G=0,S=1]- \E[\Delta Y_l \mid G=1,S=0] + \E[\Delta Y_l \mid G=0,S=0]
\end{align*}
\textbf{(II)} $\tau_l^{CL}=\tau_l$ \emph{if and only if} the parallel gaps assumption holds \textbf{unconditionally}: 
    \begin{align*}
        &\E[\Delta Y_l(0)\mid G=1,\ S=1] - \E[\Delta Y_l(0)\mid G=0,\ S=1]\\
     &=   \E[\Delta Y_l(0)\mid G=1,\ S=0] - \E[\Delta Y_l(0)\mid G=0,\ S=0]
    \end{align*}
\end{lemma}
Lemma \ref{lemma:covariatelevel} shows that the OLS estimator from the covariate-level 3WFE regression specification converges to the unconditional triple-differences estimand, $\tau_l^{CL}$. This estimand uses the constant weight profile $\widetilde w_{gs}(X)=1$ for all four subgroups and fails the sufficient covariate balancing condition because none of the three untreated subgroups is reweighted. The necessary and sufficient condition for $\tau_l^{CL} $ to equal the ATT is that the parallel gaps assumption holds unconditionally, in which case there is no need to include covariates in the triple-differences design.

In addition, although this 3WFE regression still produces OLS coefficient on $X$, in a balanced panel $\widehat{\tau}_l^{CL}$ is numerically identical to the estimator from the 3WFE regression without covariates.\footnote{In the alternative specification with individual fixed effects, discussed at the beginning of this section, the time-invariant covariates are absorbed by the fixed effects.} Therefore, including covariates in this way neither helps the identification nor improves the efficiency of estimation.

\section{Matched Triple-Differences: Estimation and Inference}\label{section:matchedtriplediffereces}
The diagnosis in Section \ref{section:empiricaldiagnose} shows that the OLS estimators from the two commonly used 3WFE specifications are, in general, inconsistent for the ATT. In this section, we propose a matched triple-differences framework for valid estimation and inference. The core of our framework consists of two parts: (i) an estimation procedure for constructing matching estimators in the triple-differences design; and (ii) a general asymptotic theory and inference that cover a class of matching procedures, such as the commonly used nearest-neighbor matching with a fixed or diverging number of neighbors and kernel matching.

\subsection{Estimation}
Estimation in our framework follows a two-step procedure. Let $\mathcal I_{gs}\equiv\left\{i: G_i=g,\, S_i=s\right\}$, and $n_{gs}=\abs{\mathcal I_{gs}}$, $(g,s)\in\{0,1\}^2$, denote the four subgroups and their sizes. In the first step, we conduct three pairwise covariate matchings; in each, the treated subgroup $\mathcal{I}_{11}$ is matched with one of the three untreated subgroups $\mathcal{I}_{gs}$, $(g,s)\neq(1,1)$. The goal of this step is to obtain a matching weight for every unit in each of the three untreated subgroups. In Figure \ref{figure:matching_workflow_proposed}, we visualize the three pairwise matchings and the resulting matching weights.

To be specific, consider the pairwise covariate matching between the treated subgroup $\mathcal{I}_{11}$ and the untreated subgroup $\mathcal{I}_{gs}$. Each treated unit $i$ is matched to its \emph{matched set} $\mathcal{M}(i)$ which consists of units in the untreated subgroup $\mathcal{I}_{gs}$ with covariates similar to $X_i$. Each untreated unit $j\in \mathcal{M}(i)$ receives a non-negative \emph{pairwise weight} $w_{ij}(\mathcal I_{11},\mathcal I_{gs})$, which measures the contribution of untreated unit $j$ to the matched set of treated unit $i$. The pairwise weights are computed by a prespecified matching algorithm (e.g., nearest-neighbor or kernel matching). An untreated unit $j$ may appear in the matched sets of several treated units and can therefore receive several pairwise weights. Summing all the pairwise weights that unit $j$ receives yields the \emph{matching weight} $w_{j}(\mathcal I_{11},\mathcal I_{gs})$, which measures how intensely untreated unit $j$ is used in the pairwise matching: 
\[
w_{j}(\mathcal I_{11},\mathcal I_{gs}) \equiv \sum_{i\in\mathcal I_{11}} w_{ij}(\mathcal I_{11},\mathcal I_{gs}),
\qquad j\in\mathcal I_{gs},
\]
In Figure \ref{figure:pairwisematching}, we illustrate, through a simplified example, how the matching weights are obtained in the pairwise matching. Carrying out all three pairwise matchings gives each untreated unit $j\in \mathcal{I}_{gs}$ a matching weight $w_{j}(\mathcal I_{11},\mathcal I_{gs})$ and leads to the following weighting profile:
\begin{align*}
w_i=
\begin{cases}
1 & i \in \mathcal I_{11}, \\[2pt]
 w_i\!\left(\mathcal I_{11},\,\mathcal I_{gs}\right) & i \in \mathcal I_{gs},\  (g,s)\neq(1,1)  \\
\end{cases}
\end{align*}
In the second step, we run a weighted 3WFE regression using the weighting profile we obtained in the first step:

\begin{equation}
\label{eq:ddd-regression}
\sqrt{w_i}Y_{it}
=
\sqrt{w_i}\left(\lambda_{G_iS_i}
+\lambda_{G_i t}
+\lambda_{S_i t}
+\sum_{l\neq-1}
    \tau_l^{match} G_iS_i\1(t-t^*=l)
+u_{it}\right).
\end{equation}
The coefficient on the triple interaction term, $\widehat{\tau}_l^{match}$, is our matched triple-differences estimator. It can then be equivalently written as
\begin{align*}
    \widehat \tau_l^{match}&= \frac{1}{n_{11}}
\sum_{i\in\mathcal{I}_{11}}
\Delta Y_{il}
-\frac{1}{n_{11}}
\sum_{i\in\mathcal{I}_{01}}w_i\!\left(\mathcal I_{11},\,\mathcal I_{01}\right)\Delta Y_{il}\\
&-\left(\frac{1}{n_{11}}
\sum_{i\in\mathcal{I}_{10}}w_i\!\left(\mathcal I_{11},\,\mathcal I_{10}\right)\Delta Y_{il}
-\frac{1}{n_{11}}
\sum_{i\in\mathcal{I}_{00}}w_i\!\left(\mathcal I_{11},\,\mathcal I_{00}\right)\Delta Y_{il}\right).
\end{align*}
Our matched triple-differences estimator is closely connected to the covariate balancing condition in Section \ref{section:econometricssetup}. To see this connection, note that in $\widehat{\tau}_l^{match}$, the treated subgroup is unweighted, while the matching weights $w_i\!\left(\mathcal I_{11},\,\mathcal I_{gs}\right)$ approximate the density ratio, so that each untreated subgroup is reweighted to have the approximately the same covariate distribution as the treated subgroup. 

\begin{figure}[htbp]
\centering
\savebox{\figbox}{%
\begin{tikzpicture}[
    >={Latex[length=2.4mm,width=2mm]},
    node distance=1.6cm and 3.4cm,        
    box/.style={
        draw, line width=0.6pt, rounded corners=3pt,
        minimum width=4.8cm, minimum height=1.7cm,
        align=center, font=\small
    },
    treated/.style={box, line width=1.2pt, fill=gray!20},
    header/.style={font=\small\bfseries, align=center},
    match/.style={->, line width=0.7pt},
    lab/.style={font=\footnotesize, inner sep=2.5pt}
]
\node[treated]         (A) {Treated subgroup $(\mathcal{I}_{11})$\\[2pt] $w_i = 1$};
\node[box, right=of A] (B) {Unexposed Eligible $(\mathcal{I}_{01})$\\[2pt] $w_i = w_i(\mathcal{I}_{11},\mathcal{I}_{01})$};
\node[box, below=of A] (C) {Exposed Ineligible $(\mathcal{I}_{10})$\\[2pt] $w_i = w_i(\mathcal{I}_{11},\mathcal{I}_{10})$};
\node[box, below=of B] (D) {Unexposed Ineligible $(\mathcal{I}_{00})$\\[2pt] $w_i = w_i(\mathcal{I}_{11},\mathcal{I}_{00})$};
\node[header, above=0.45cm of A] {Exposed\\ $(G_i = 1)$};
\node[header, above=0.45cm of B] {Unexposed\\ $(G_i = 0)$};
\node[header, rotate=90, anchor=south] at ([xshift=-0.45cm]A.west) {Eligible\\ $(S_i = 1)$};
\node[header, rotate=90, anchor=south] at ([xshift=-0.45cm]C.west) {Ineligible\\ $(S_i = 0)$};
\draw[match] (A) -- node[lab, above]         {pairwise matching 1} (B);
\draw[match] (A) -- node[lab, left]          {pairwise matching 2} (C);
\draw[match] (A) -- node[lab, above, sloped] {pairwise matching 3} (D);
\end{tikzpicture}}
\begin{adjustbox}{max width=\linewidth}\usebox{\figbox}\end{adjustbox}
\caption{The matched triple-differences procedure}
\label{figure:matching_workflow_proposed}
\vspace{0.35em}
\setlength{\notewidth}{\wd\figbox}%
\ifdim\notewidth>\linewidth \setlength{\notewidth}{\linewidth}\fi
\begin{minipage}{\notewidth}
\footnotesize
\raggedright
\emph{Note.} The matched triple-differences procedure performs three separate pairwise matchings. In each pairwise matching, every unit in the treated subgroup \((\mathcal{I}_{11})\) is matched to similar units in one untreated subgroup, based on observed covariates $X_i$.
\end{minipage}
\end{figure}
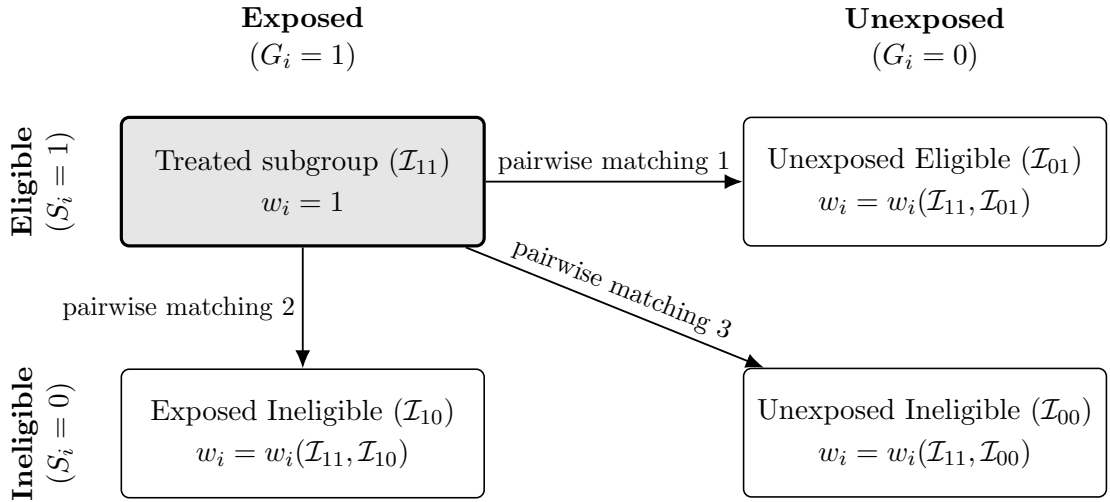

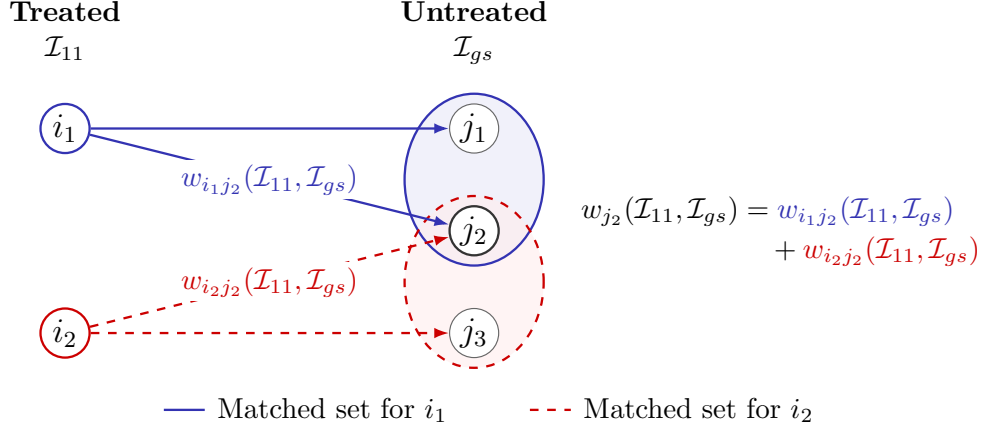
\begin{figure}[htbp]
\centering
\sbox{\figbox}{%
    \resizebox{!}{0.25\textheight}{%
    \begin{tikzpicture}[
        x=1cm,y=1cm,
        every node/.style={font=\footnotesize},
        unit/.style={
            circle,draw=black!60,fill=white,
            minimum size=6mm,inner sep=0pt,
            font=\normalsize
        },
        pair/.style={
            -{Latex[length=1.8mm,width=1.3mm]},thick
        },
        wt/.style={
            fill=white,inner sep=1.5pt,
            font=\footnotesize
        },
        eq/.style={
            anchor=west,inner sep=0pt,
            font=\footnotesize
        }
    ]
        \fill[myblue!7]
            (5,.625) ellipse (.85cm and 1.05cm);
        \fill[red!7,opacity=.60]
            (5,-.625) ellipse (.85cm and 1.05cm);
 
        \draw[myblue,thick]
            (5,.625) ellipse (.85cm and 1.05cm);
        \draw[red!80!black,thick,dashed]
            (5,-.625) ellipse (.85cm and 1.05cm);
 
        \node[font=\footnotesize\bfseries,align=center,anchor=north]
            at (0,2.95) {Treated\\$\mathcal I_{11}$};
        \node[font=\footnotesize\bfseries,align=center,anchor=north]
            at (5,2.95) {Untreated\\$\mathcal I_{gs}$};
 
        \node[unit,draw=myblue,thick]
            (i1) at (0,1.25) {$i_1$};
        \node[unit,draw=red!80!black,thick]
            (i2) at (0,-1.25) {$i_2$};
 
        \node[unit] (j1) at (5,1.25) {$j_1$};
        \node[unit,draw=black!80,thick]
            (j2) at (5,0) {$j_2$};
        \node[unit] (j3) at (5,-1.25) {$j_3$};
 
        \draw[pair,myblue] (i1) -- (j1);
        \draw[pair,red!80!black,dashed] (i2) -- (j3);
 
        \draw[pair,myblue] (i1) --
            node[wt,text=myblue,pos=.5]
            {$w_{i_1j_2}(\mathcal{I}_{11},\mathcal{I}_{gs})$} (j2);
 
        \draw[pair,red!80!black,dashed] (i2) --
            node[wt,text=red!80!black,pos=.5]
            {$w_{i_2j_2}(\mathcal{I}_{11},\mathcal{I}_{gs})$} (j2);
 
        \node[eq] at (6.3,0) {
            $\begin{aligned}
                w_{j_2}(\mathcal{I}_{11},\mathcal{I}_{gs})
                &=\textcolor{myblue}{w_{i_1j_2}(\mathcal{I}_{11},\mathcal{I}_{gs})}\\[-1pt]
                &\quad+\textcolor{red!80!black}{w_{i_2j_2}(\mathcal{I}_{11},\mathcal{I}_{gs})}
            \end{aligned}$
        };
 
        \coordinate (leg) at ([yshift=-.5cm]current bounding box.south);
        \node[anchor=east,inner sep=0pt] (leg1)
            at ([xshift=-.5cm]leg) {Matched set for $i_1$};
        \draw[myblue,thick]
            ([xshift=-.65cm]leg1.west) -- ([xshift=-.15cm]leg1.west);
 
        \draw[red!80!black,thick,dashed]
            ([xshift=.5cm]leg) -- ([xshift=1cm]leg);
        \node[anchor=west,inner sep=0pt]
            at ([xshift=1.15cm]leg) {Matched set for $i_2$};
    \end{tikzpicture}%
    }%
}%
\begin{adjustbox}{max width=\linewidth}\usebox{\figbox}\end{adjustbox}
\caption{Construction of the matching weights}
\label{figure:pairwisematching}
\vspace{0.35em}
\setlength{\notewidth}{\wd\figbox}%
\ifdim\notewidth>\linewidth \setlength{\notewidth}{\linewidth}\fi
\begin{minipage}{\notewidth}
\footnotesize
\raggedright
\emph{Note.} This figure illustrates how to obtain the matching weight $w_{j_2}(\mathcal{I}_{11},\mathcal{I}_{gs})$ for untreated unit $j_2\in\mathcal{I}_{gs}$ in the pairwise matching between the treated subgroup $\mathcal{I}_{11}$ and untreated subgroup $\mathcal{I}_{gs}$.
\end{minipage}
\end{figure}

\subsection{Asymptotic Theory}
The key novelty of our asymptotic theory is a set of regularity conditions on the pairwise covariate matching used in estimation. These conditions allow our framework to abstract from any specific matching method and to provide a general asymptotic theory for a class of estimators constructed from different matching methods. To avoid clutter, we suppress the dependence of the weights on the pair of subgroups, writing $w_{ij}$ and $w_j$ and we write $\widehat{\tau}_l$ for $\widehat{\tau}_l^{match}$ when no confusion arises.
 
\begin{assumption}[Regularity of the pairwise covariate matching]
\label{assumption:matching}
For each $(g,s)\in\{(0,1),(1,0),(0,0)\}$, the pairwise covariate matching between the treated subgroup $\mathcal I_{11}$ and untreated subgroup $\mathcal I_{gs}$ satisfies the following conditions.
\begin{enumerate}[label=(\roman*)]
\item The pairwise weights $\left\{w_{ij}:i\in\mathcal I_{11},\,j\in\mathcal I_{gs} \right\}$ are constructed using only the subgroup indicators and observed covariates of the treated and untreated subgroups, $(G_i,S_i,X_i)_{i\in \mathcal{I}_{11}\cup \mathcal{I}_{gs}}$. They do not depend on the order in which the observations are supplied: permuting the observations within $\mathcal I_{11}$ and $\mathcal I_{gs}$ permutes the corresponding pairwise weights accordingly. Moreover, for each $i\in\mathcal I_{11}$,
\[
w_{ij}\geq 0 \quad\text{for every }j\in\mathcal I_{gs},
\qquad
\sum_{j\in\mathcal I_{gs}} w_{ij} = 1.
\]
\item \textbf{(Locality)}
\[
\E\!\left[
\sum_{j\in\mathcal I_{gs}} w_{ij}\left\lVert X_j-X_i\right\rVert
\,\middle|\, i\in\mathcal I_{11}
\right] = o(1).
\]

\item \textbf{(Bounded Moments)} For some $\delta>0$,
\[
 \E\!\left[ w_{j}^{4+\delta} \,\middle|\, j\in\mathcal I_{gs} \right] =O(1).
\]

\item \textbf{(Stability)} Let $w_{j}^{(1)}$ denote the matching weight of unit $j$ computed after replacing the first observation $(Y_{11},\dots, Y_{1T}, G_1,S_1,X_1 ) $ with an independent copy $(Y_{1'1},\dots, Y_{1'T}, G_{1'},S_{1'},X_{1'} )$.
\[
\E\!\left[\Bigg(\sum_{j\in\mathcal I_{gs}\setminus\{1\}}
\bigl|\,w_{j}-w_{j}^{(1)}\bigr|\Bigg)^{\!4}\right]
= O(1).
\]
\end{enumerate}

\end{assumption}

Condition (i) restricts the information available to the matching procedure: the pairwise weights depend only on the subgroup indicators and observed covariates, and they are invariant to the order in which observations are supplied. The non-negativity and adding-up requirements imply that each treated unit is matched to a convex combination of untreated units. 

Condition (ii) imposes locality: the covariate discrepancy between a treated unit and its matched untreated units, averaged using the pairwise weights, vanishes as the sample grows. In other words, each treated unit is asymptotically matched to untreated units with nearly identical covariate values. 

Condition (iii) restricts how intensively an untreated unit may be reused across matched sets. The $4+\delta$ moment bound rules out matching procedures that heavily reuse only a handful of untreated units. 

Condition (iv) is a stability requirement on the matching procedure: it rules out procedures in which a single observation can exert a large influence on the overall allocation of the matching weights. This condition cannot simply be removed if we want to establish unconditional asymptotic normality, because under conditions (i)--(iii) alone, asymptotic normality of the matching estimators is guaranteed only conditionally on the subgroup indicators and the covariates. In Appendix \ref{section:example}, we provide an example of a matching procedure that satisfies conditions (i)--(iii) but not (iv), for which the resulting matched triple-differences estimator fails to be asymptotically normal unconditionally.
 
Assumption~\ref{assumption:matching} is satisfied by several commonly used matching procedures: Lemma~\ref{lemma:matchingprocedure} in the Appendix \ref{appendix:matchingprocedure} shows that it holds for nearest-neighbor matching with replacement, with either a fixed or a diverging number of neighbors and for kernel matching. The assumption does, however, exclude some procedures used in the empirical literature. Nearest-neighbor matching \emph{without} replacement violates condition (i), because its pairwise weights depend on the order in which the observations are supplied; local linear matching violates condition (i) as well, because its pairwise weights can be negative.

We also impose the following regularity conditions, which are standard in the matching literature (e.g., \citeinb{Abadie-Imbens_2006_ECMA}).
 
\begin{assumption}[Regularity conditions]\label{assumption:regularity}
For every $(g,s)\in\{0,1\}^2$, $d\in\{0,1\}$, $t=1,\dots,T$, and event time $l=0,1,\dots,T-t^*$:
\begin{enumerate}[label=(\roman*)]
\item the conditional density $f_{gs}$ of $X$ given $(G,S)=(g,s)$ is continuous, has convex and compact support, and is bounded and bounded away from zero on its support;
\item $\E[Y_{t}(d)\mid G=g,S=s,X=x]$ is Lipschitz continuous in $x$;
\item $\mathrm{Var}[\Delta Y_{l}(d)\mid G=g,S=s,X=x]$ is Lipschitz continuous in $x$ and uniformly bounded above and away from zero, that is, $0<\underline{\sigma}^2\leq \mathrm{Var}[\Delta Y_{l}(d)\mid G=g,S=s,X=x]\leq \overline{\sigma}^2$  ;
\item $\E[Y_{t}(d)^4\mid G=g,S=s,X=x]$ is uniformly bounded.
\end{enumerate}
\end{assumption}

For notational simplicity, define $\mu_l^{gs}(x)=\E[\Delta Y_l|G=g,S=s,X=x]$ and $\sigma^2_{gs,l}(x)=Var(\Delta Y_l|G=g,S=s,X=x)$.
\begin{theorem}\label{theorem:asymptotic}
    Under Assumptions \ref{assumption:noanticipation},
\ref{assumption:overlap}, \ref{assumption:parallelgap}, \ref{assumption:iid}, \ref{assumption:matching}, and \ref{assumption:regularity},
\begin{align*}
    \widehat\tau_l-\tau_l\xrightarrow{p} 0, \quad
         (V_{\tau(X)}+V_E)^{-1/2} \sqrt{n_{11}}\left(\widehat{\tau}_l-\tau_l-B \right)\xrightarrow{d}\mathcal{N}(0,1)
    \end{align*}
    where
    \begin{align*}
        V_{\tau(X)}& =\E\left[\left. \left(\mu_l^{11}(X_i)-\mu_l^{01}(X_i)-\mu_l^{10}(X_i)+\mu_l^{00}(X_i)-\tau_l \right)^2\right\vert G_i=1,S_i=1\right] \\
    V_E&= \frac{1}{n_{11}}\left(\sum_{i\in\mathcal{I}_{11}}\sigma^2_{11,l}(X_{i})+\sum_{i\in\mathcal{I}_{01}} w_{i}^2\sigma^2_{01,l}(X_{i})+\sum_{i\in\mathcal{I}_{10}} w_{i}^2\sigma^2_{10,l}(X_{i}) +\sum_{i\in\mathcal{I}_{00}} w_{i}^2\sigma^2_{00,l}(X_{i})\right)\\
   B&=\frac{1}{n_{11}}\isum\left[\1(i\in \mathcal{I}_{11})-\1(i\in \mathcal{I}_{01})w_i\right]\mu_l^{01}(X_i)+\frac{1}{n_{11}}\isum\left[\1(i\in \mathcal{I}_{11})-\1(i\in \mathcal{I}_{10})w_i\right]\mu_l^{10}(X_i)\\ 
    &-\frac{1}{n_{11}}\isum\left[\1(i\in \mathcal{I}_{11})-\1(i\in \mathcal{I}_{00})w_i\right]\mu_l^{00}(X_i) 
    \end{align*}
    
\end{theorem}
Theorem \ref{theorem:asymptotic} establishes consistency and asymptotic normality for the class of matched triple-differences estimators that fall within our framework. It extends classical asymptotic results for matching estimators, such as those of \cite{Abadie-Imbens_2006_ECMA}, both to the triple-differences setting and to a class of matching procedures. 
 
Theorem \ref{theorem:asymptotic} also shows that our matched triple-differences estimator carries a finite-sample bias term $B$. The bias term does not affect consistency, but it can affect inference: $\sqrt{n_{11}}B$ need not be asymptotically negligible, so the matched triple-differences estimator may require bias correction for valid inference. This issue parallels the one first pointed out by \cite{Abadie-Imbens_2006_ECMA} for cross-sectional matching estimators. In some cases, bias correction is not required. For example, if covariates are discrete,  we can match on them exactly, and matching bias is zero. In addition, the results of \cite{Viel_2025_wp} imply that, for nearest-neighbor matching and kernel matching with properly chosen tuning parameters, the matching bias is root-$n$ negligible when the dimension of the continuous covariate is at most $3$, so bias correction is not required\footnote{Let $r$ denote the number of continuous covariates and $(g,s)\in\{(0,1),(1,0),(0,0)\}$. The matching bias is root-$n$ negligible if (i) for nearest-neighbor matching, $M=o(n_{gs}^{2/3})$ when $r=1$, $M=o(n_{gs}^{1/2})$ when $r=2$, and $M=o(n_{gs}^{1/4})$ when $r=3$; and (ii) for kernel matching, $h=n_{gs}^{-a}$ with $1/4<a<1$ when $r=1$, $1/4<a<1/2$ when $r=2$, and $1/4<a<1/3$ when $r=3$.}.

When bias correction is required, we can correct our matched triple-differences estimator by adapting the bias-correction procedure in the cross-sectional setting \citep{Abadie-Imbens_2011_JBES,Lin_Han_2025_JOE} to each of the three terms of the bias $B$ in Theorem \ref{theorem:asymptotic}. To conduct the bias correction, we estimate the three unknown conditional mean functions $\mu_l^{gs}$, $(g,s)\in\{(0,1),(1,0),(0,0)\}$, nonparametrically by $\widehat{\mu}_l^{gs}$, using off-the-shelf methods such as series or spline estimators. The bias-corrected matched triple-differences estimator then takes the form:
\begin{align*}
    \widehat{\tau}_l^{bc}&= \widehat{\tau}_l -\widehat{B} \\
    &=\widehat{\tau}_l-\left(\frac{1}{n_{11}}\isum\left[\1(i\in \mathcal{I}_{11})-\1(i\in \mathcal{I}_{01})w_i\right]\widehat\mu_l^{01}(X_i)+\frac{1}{n_{11}}\isum\left[\1(i\in \mathcal{I}_{11})-\1(i\in \mathcal{I}_{10})w_i\right]\widehat\mu_l^{10}(X_i)\right.\\
    &\left.-\frac{1}{n_{11}}\isum\left[\1(i\in \mathcal{I}_{11})-\1(i\in \mathcal{I}_{00})w_i\right]\widehat\mu_l^{00}(X_i)\right)
\end{align*}
\begin{theorem}\label{theorem:biascorrection}
    Under Assumptions \ref{assumption:noanticipation},
\ref{assumption:overlap}, \ref{assumption:parallelgap}, \ref{assumption:iid}, \ref{assumption:matching}, \ref{assumption:regularity}, and \ref{assumption:biascorrection},
\begin{align*}
     \sqrt{n_{11}}(\widehat{B}-B)\xrightarrow{p} 0,\quad
    (V_{\tau(X)}+V_E)^{-1/2} \sqrt{n_{11}}\left(\widehat{\tau}_l-\tau_l-\widehat B \right)\xrightarrow{d}\mathcal{N}(0,1).
\end{align*}
\end{theorem}
Theorem \ref{theorem:biascorrection} shows that the bias can be estimated with an error of order smaller than $n_{11}^{-1/2}$. Consequently, the bias-corrected matched triple-differences estimator is $\sqrt{n_{11}}$-consistent and asymptotically normal, and the correction leaves the asymptotic variance unchanged, so the variance estimator in Section \ref{subsection:varianceestimation} still applies to our bias-corrected estimator.

\begin{remark}
    Although our bias-corrected estimator is motivated by removing the matching bias, equation (\ref{eq:AIPW}) shows that it has the structure of an augmented inverse probability weighting (AIPW) estimator. The first term is the outcome regression part and the remaining three terms form the weighting (IPW) part, with the inverse propensity score weights replaced by the matching weights. This parallels the novel insight of \cite{Lin-Ding-Han_2023_ECMA} and \cite{Lin_Han_2025_JOE}, who point out that the bias-corrected matching estimator in the cross-sectional setting has the structure of an AIPW estimator and provide conditions under which it enjoys the desirable properties of AIPW estimators such as double robustness and attainment of the semiparametric efficiency bound. We conjecture that, under suitable assumptions, our bias-corrected estimator $\widehat{\tau}_{l}^{bc}$ is also doubly robust and semiparametrically efficient. 
\begin{equation}
\label{eq:AIPW}
\begin{split}
        \widehat{\tau}_l^{bc}&=
\frac{1}{n_{11}}\sum_{i\in\mathcal{I}_{11}} \left[
 \Delta Y_{il}
-\widehat{\mu}_l^{01}(X_i)
-\widehat{\mu}_l^{10}(X_i)
+\widehat{\mu}_l^{00}(X_i)
\right]\\
&-\frac{1}{n_{11}}\sum_{i\in\mathcal{I}_{01}}
w_i\left[\Delta Y_{il}-\widehat{\mu}_l^{01}(X_i)\right]-\frac{1}{n_{11}}\sum_{i\in\mathcal{I}_{10}}
w_i\left[\Delta Y_{il}-\widehat{\mu}_l^{10}(X_i)\right]+\frac{1}{n_{11}}\sum_{i\in\mathcal{I}_{00}}
w_i\left[\Delta Y_{il}-\widehat{\mu}_l^{00}(X_i)\right].
\end{split}
\end{equation}
\end{remark}

\subsection{Variance Estimation}\label{subsection:varianceestimation}
Because our estimator is obtained from a weighted 3WFE regression, it is tempting to use the cluster-robust standard error from that regression. Previous studies (\citeinb{Abadie-Spiess_2022_JASA}; \citeinb{Liu_Vazquez_Bare_2026_wp}) show, however, that this standard error is inconsistent because it treats the matching weights as fixed. In this subsection, we propose a unified variance estimator that accounts for the matching step and is valid for all matching procedures within our framework. 

For each untreated subgroup $(g,s)\in \{(0,1),(1,0),(0,0)\}$, let $\widehat{\Delta Y_{il}^{gs}}\equiv\sum_{j\in \mathcal{I}_{gs}}w_{ij}\Delta Y_{jl}$ denote the matched outcome change from subgroup $(g,s)$ for treated unit $i$, and let $\widetilde w_j^2\equiv\sum_{i\in \mathcal{I}_{11}}w_{ij}^2$ denote the sum of the squared pairwise weights received by untreated unit $j$. We construct our variance estimator by estimating each of the two components of the asymptotic variance in Theorem~\ref{theorem:asymptotic},  the effect-heterogeneity component, $V_{\tau(X)}$, and the conditional-variance component, $V_E$. 
\begin{align*}
    \widehat{V}_{\tau(X)} & =  \frac{1}{n_{11}}\sum_{i\in \mathcal{I}_{11}}\left(\Delta Y_{il}-  \widehat {\Delta Y_{il}^{01}}-\widehat {\Delta Y_{il}^{10}} + \widehat {\Delta Y_{il}^{00}}-\widehat{\tau}_l \right)^2,  \\
    \widehat{V}_E & =\frac{1}{n_{11}}\sum_{j\in  \mathcal{I}_{01}} \left(  w_j^2-\widetilde w_j^2\right)\widehat{\sigma}_{01,l}^2(X_j) + \frac{1}{n_{11}}\sum_{j\in  \mathcal{I}_{10}} \left( w_j^2-\widetilde w_j^2 \right)\widehat{\sigma}_{10,l}^2(X_j) + \frac{1}{n_{11}}\sum_{j\in  \mathcal{I}_{00}} \left(  w_j^2-\widetilde w_j^2 \right)\widehat{\sigma}_{00,l}^2(X_j).
\end{align*}
Here, $\widehat{\sigma}_{gs,l}^2(X_j)$ is an estimator of the unknown conditional variance function, constructed following the strategy proposed by \cite{Abadie-Imbens_2006_ECMA}. For each untreated subgroup, we conduct an additional $K$-nearest-neighbor matching within the same subgroup and denote by $k_m(j)$ the index of the $m$-th nearest neighbor of unit $j$. The estimator $\widehat{\sigma}_{gs,l}^2(X_j)$ is then defined as:
\[
\widehat{\sigma}_{gs,l}^2(X_j)
=
\frac{K}{K+1}
\left(
   \Delta Y_{jl}-\frac{1}{K}\sum_{m=1}^K \Delta Y_{k_m(j)l}
\right)^2.
\]
Theorem~\ref{theorem:varianceestimation} shows that our variance estimator is consistent under the same assumptions as Theorem~\ref{theorem:asymptotic}. 
\begin{theorem}\label{theorem:varianceestimation}
    Under Assumptions \ref{assumption:noanticipation},
\ref{assumption:overlap}, \ref{assumption:parallelgap}, \ref{assumption:iid}, \ref{assumption:matching}, and \ref{assumption:regularity},
  \begin{align*}
      \frac{\widehat V_{\tau(X)}+\widehat V_E}{V_{\tau(X)}+V_E  } \xrightarrow{p} 1.
  \end{align*}
\end{theorem}
 Our variance estimation approach builds on that of \citet{Abadie-Imbens_2006_ECMA} for the cross-sectional setting but extends it in two directions. First, our variance estimator is tailored to the data structure and estimation procedure of the triple-differences design. This structure generates additional interaction terms between the residuals of different untreated subgroups. Our proof must handle these terms, which have no counterpart in the two-group cross-sectional setting. Second, the consistency proof for the variance estimator in \citet{Abadie-Imbens_2006_ECMA} relies on properties specific to nearest-neighbor matching with a fixed number of neighbors, so their proof cannot be applied directly here. In stead, we develop a refined proof strategy that is valid for the class of matching procedures characterized by Assumption~\ref{assumption:matching}, which includes nearest-neighbor matching with a fixed number of neighbors as a special case.

\section{Simulation}\label{section:simulation}
In this section, we provide the Monte Carlo evidence on (i) the inconsistency of the OLS estimators from the two 3WFE regression specifications studied in Section \ref{section:empiricaldiagnose} and (ii) the finite-sample performance of our proposed matched triple-differences estimators, their bias-corrected versions, and the proposed variance estimator.

We simulate a panel of $n=5,000$ units observed over four periods, $t=1,2,3,4$. Each unit belongs to one of four subgroups indexed by $(G,S)\in\{0,1\}^2$. The treatment is implemented simultaneously in period $t^*=3$ and only the exposed eligible subgroup, $(G,S)=(1,1)$, is treated. 

Following the data generating process (DGP) of \cite{Otsu_Rai_2017_JASA}, we generate each unit's $q$-dimensional covariate vector $X_i=(X_{i1},\dots, X_{iq})'$:
\begin{align*}
    X_{ij}&=\xi_i\frac{\abs{\zeta_{ij}}}{\norm{\zeta_i}},\quad \text{if } j=1,\dots,q \\ 
    \xi_i&\sim \text{Uniform}(0,1), \quad \zeta_i\sim \mathcal{N}(0,I_q)
\end{align*}
Let $\norm{X}$ denote the Euclidean norm of $X$. Under our DGP,  $\norm{X}$ follows a uniform distribution over $(0,1)$.

Each unit is assigned to one of the four subgroups according to a multinomial logit model:
\begin{align*}
    P(G_i=g,S_i=s|X_i)=\frac{\exp(\alpha_{gs}+\gamma_{gs}(\norm{X_i}-0.5))}{\sum_{(g,s)}\exp(\alpha_{gs}+\gamma_{gs}(\norm{X_i}-0.5)) }
\end{align*}
The intercepts $\alpha_{gs}$ control the relative sizes of the four subgroups. We consider two sets of intercepts. The first gives the four subgroups approximately equal sizes: $\alpha_{11}=\alpha_{10}=\alpha_{01}=\alpha_{00}=0$. The second makes the exposed ineligible subgroup $(1,0)$ smaller than the other three subgroups: $\alpha_{11}=0,  \alpha_{10}=-0.776, \alpha_{01}=0.228,\alpha_{00}=0.270 $. The targeted shares are approximately $25\%$, $11\%$, $30\%$, and $34\%$ for subgroups $(1,1)$, $(1,0)$, $(0,1)$, and $(0,0)$, respectively. This second set of intercepts is designed to mimic the subgroup sizes in our empirical application in Section \ref{section:empiricalapplication}.

The slopes $\gamma_{gs}$ determine how selection into the four subgroups depends on the covariates. We fix the slopes  at $\gamma_{11}=0,\gamma_{10}=-2 ,\gamma_{01}=-0.5, \gamma_{00}=-3$. We choose these values so that selection on the covariate is stronger along the eligibility dimension than along the exposure dimension. This mirrors many triple-differences designs, in which the the policy is typically implemented along the exposure dimension (e.g., at the county or state level). Because the policy is usually implemented quasi-randomly, selection on individual covariates is weaker long the exposure dimension; along the eligibility dimension (e.g., growing tobacco, or being pregnant or not), selection on individual covariates is usually stronger.

The potential outcomes are defined as
\begin{align*}
    Y_{it}(0)&=\alpha_i + 2 \beta_0 G_i(\norm{X_i}-0.5)t +\beta_0 S_i (\norm{X_i}-0.5)t +v_{it}(0)\\
    Y_{it}(1)&=\tau + \alpha_i + 2 \beta_1 G_i(\norm{X_i}-0.5)t +\beta_1 S_i (\norm{X_i}-0.5)t +v_{it}(1)\\
    \alpha_i&\sim \mathcal{N}(0,1), \quad v_{it}(0)\sim\mathcal{N}(0,1), \quad v_{it}(1)\sim \mathcal{N}(0,1)
\end{align*}
with $\alpha_i$, $v_{it}(0)$, $v_{it}(1)$, $X_i$ mutually independent. This specification satisfies the conditional parallel gaps assumption but violates the conditional parallel trends assumption. We set $\tau=1$ and consider two choices of $\beta_0$ and $\beta_1$. The first sets $\beta_0=\beta_1=2$, so the treatment effect is homogeneous with respect to $X$. The second sets $\beta_0=2,\beta_1=6$, so the treatment effect is heterogeneous with respect to $X$.

Based on the discussion above, we consider four main simulation designs. The first and second designs have equal subgroup sizes with a homogeneous/heterogeneous treatment effect. The third and fourth designs have a small $(1,0)$ subgroup size with a homogeneous/heterogeneous treatment effect.

We report results for $\tau_0$, the treatment effect in the implementation period $t^*=3$, based on $1,000$ repetitions with covariate dimension $q=2$. For each design, we report the OLS estimators from the two 3WFE regression specifications studied in Section \ref{section:empiricaldiagnose}, the matched triple-differences estimators under three matching methods, and their corresponding bias-corrected estimators. The three matching methods and their tuning parameters are: nearest-neighbor matching with one neighbor; nearest neighbor matching with a diverging number of neighbors, set to $n_{gs}^{1/3}$; and kernel matching with the Epanechnikov kernel and a bandwidth rule based on Stata's kmatch implementation, which sets the bandwidth to $1.5$ times the $90$th percentile of the non-zero pairwise distances from one-to-one nearest-neighbor matching with replacement. Bias correction uses a second-order series estimator. For each of the eight estimators, we report the Monte Carlo bias, the root mean squared error (RMSE), the Monte Carlo standard deviation, the average standard error computed by our proposed variance estimator, and the empirical coverage rate and length of the corresponding $95\%$ confidence interval. 

The results reported in Tables \ref{table:design1} to \ref{table:design4} confirm our theoretical findings in Sections \ref{section:empiricaldiagnose} and \ref{section:matchedtriplediffereces}. In all four designs, the OLS estimators from the two 3WFE regression specifications exhibit substantial bias, ranging from $14\%$ to $81\%$ of the true treatment effect. By contrast, our proposed estimators have negligible bias in all four designs. For nearest-neighbor matching with a diverging number of neighbors, the bias-corrected estimator has smaller bias than the uncorrected one in all four designs. Our proposed variance estimator performs well: in almost all combinations of the four designs and six matching estimators, the average standard error is close to the Monte Carlo standard deviation, and the coverage rate of the $95\%$ confidence interval is close to the nominal rate. Nearest-neighbor matching with a diverging number of neighbors and the kernel matching are more efficient than nearest-neighbor matching with one neighbor, with smaller standard errors and shorter confidence intervals.

\begin{table}[htbp]
\centering
\caption{Balanced Subgroup Size, Homogeneous Treatment Effect}
\label{table:design1}
\small
\begin{tabular}{lrrrrrr}
\toprule
& \multicolumn{1}{c}{Bias} & \multicolumn{1}{c}{RMSE} & \multicolumn{1}{c}{MCSD} & \multicolumn{1}{c}{Avg.~SE} & \multicolumn{1}{c}{Coverage} & \multicolumn{1}{c}{CI length} \\
\midrule
Cov. Trend ($X'\theta_t$) & $0.8084$ & $0.8131$ & $0.0879$ & $0.0874$ & $0.000$ & $0.3426$ \\
Cov. Level ($X'\theta$) & $0.7008$ & $0.7079$ & $0.1001$ & $0.0988$ & $0.000$ & $0.3875$ \\
\midrule
NN, $M = 1$ & $0.0060$ & $0.1203$ & $0.1202$ & $0.1225$ & $0.954$ & $0.4801$ \\
\quad bias-corrected & $0.0025$ & $0.1203$ & $0.1203$ & $0.1225$ & $0.953$ & $0.4801$ \\
\addlinespace
NN, $M = n_{gs}^{1/3}$ & $0.0238$ & $0.0964$ & $0.0935$ & $0.0946$ & $0.944$ & $0.3709$ \\
\quad bias-corrected & $0.0052$ & $0.0944$ & $0.0943$ & $0.0946$ & $0.950$ & $0.3709$ \\
\addlinespace
Kernel & $0.0086$ & $0.1029$ & $0.1026$ & $0.1040$ & $0.954$ & $0.4077$ \\
\quad bias-corrected & $0.0037$ & $0.1027$ & $0.1027$ & $0.1040$ & $0.954$ & $0.4077$ \\
\bottomrule
\end{tabular}
\par\smallskip
\begin{minipage}{\linewidth}\footnotesize
\emph{Notes:} This table reports the simulation results for design 1 with balanced subgroup size and homogeneous treatment effect. We run $1,000$ Monte Carlo replications, each with $n = 5,000$ units, $q = 2$ covariates, and $4$ periods. The treatment effect is $\tau_0 = 1$. The two OLS estimators are estimated using the two regression specifications studied in section \ref{section:empiricaldiagnose}. The three matching estimators use (i) nearest-neighbor matching with one neighbor, (ii) nearest-neighbor matching with a diverging number of neighbors, set to $n_{gs}^{1/3}$, and (iii) the kernel matching estimator with the Epanechnikov kernel with a bandwidth rule based on Stata's kmatch implementation, which is set to be $1.5$ times the $90$th percentile of non-zero pairwise distances from matching with replacement. The three bias-corrected matching estimators use the corresponding matching methods, and the bias correction is performed using a second-order series estimator.
\end{minipage}
\end{table}

\begin{table}[htbp]
\centering
\caption{Balanced Subgroup Size, Heterogeneous Treatment Effect}
\label{table:design2}
\small
\begin{tabular}{lrrrrrr}
\toprule
& \multicolumn{1}{c}{Bias} & \multicolumn{1}{c}{RMSE} & \multicolumn{1}{c}{MCSD} & \multicolumn{1}{c}{Avg.~SE} & \multicolumn{1}{c}{Coverage} & \multicolumn{1}{c}{CI length} \\
\midrule
Cov. Trend ($X'\theta_t$) & $1.1133$ & $1.1472$ & $0.2768$ & $0.2768$ & $0.016$ & $1.0851$ \\
Cov. Level ($X'\theta$) & $0.6864$ & $0.7601$ & $0.3268$ & $0.3302$ & $0.454$ & $1.2943$ \\
\midrule
NN, $M = 1$ & $-0.0085$ & $0.2900$ & $0.2901$ & $0.2993$ & $0.957$ & $1.1731$ \\
\quad bias-corrected & $-0.0120$ & $0.2901$ & $0.2900$ & $0.2992$ & $0.958$ & $1.1727$ \\
\addlinespace
NN, $M = n_{gs}^{1/3}$ & $0.0093$ & $0.2839$ & $0.2839$ & $0.2894$ & $0.949$ & $1.1343$ \\
\quad bias-corrected & $-0.0092$ & $0.2834$ & $0.2834$ & $0.2888$ & $0.951$ & $1.1320$ \\
\addlinespace
Kernel & $-0.0058$ & $0.2856$ & $0.2857$ & $0.2921$ & $0.950$ & $1.1451$ \\
\quad bias-corrected & $-0.0107$ & $0.2856$ & $0.2856$ & $0.2920$ & $0.948$ & $1.1446$ \\
\bottomrule
\end{tabular}
\par\smallskip
\begin{minipage}{\linewidth}\footnotesize
\emph{Notes:} This table reports the simulation results for design 2 with balanced subgroup size and heterogeneous treatment effect. We run $1,000$ Monte Carlo replications, each with $N = 5,000$ units, $q = 2$ covariates, and $4$ periods. The treatment effect is $\tau_0 = 4.87$. The two OLS estimators are estimated using the two regression specifications studied in section \ref{section:empiricaldiagnose}. The three matching estimators use (i) nearest-neighbor matching with one neighbor, (ii) nearest-neighbor matching with a diverging number of neighbors, set to $n_{gs}^{1/3}$, and (iii) the kernel matching estimator with the Epanechnikov kernel with a bandwidth rule based on Stata's kmatch implementation, which is set to be $1.5$ times the $90$th percentile of non-zero pairwise distances from matching with replacement. The three bias-corrected matching estimators use the corresponding matching methods, and the bias correction is performed using a second-order series estimator.
\end{minipage}
\end{table}

\begin{table}[htbp]
\centering
\caption{Small $(1,0)$ Subgroup, Homogeneous Treatment Effect}
\label{table:design3}
\small
\begin{tabular}{lrrrrrr}
\toprule
& \multicolumn{1}{c}{Bias} & \multicolumn{1}{c}{RMSE} & \multicolumn{1}{c}{MCSD} & \multicolumn{1}{c}{Avg.~SE} & \multicolumn{1}{c}{Coverage} & \multicolumn{1}{c}{CI length} \\
\midrule
Cov. Trend ($X'\theta_t$) & $0.8118$ & $0.8192$ & $0.1102$ & $0.1102$ & $0.000$ & $0.4318$ \\
Cov. Level ($X'\theta$) & $0.7087$ & $0.7196$ & $0.1253$ & $0.1242$ & $0.001$ & $0.4869$ \\
\midrule
NN, $M = 1$ & $0.0134$ & $0.1537$ & $0.1532$ & $0.1496$ & $0.951$ & $0.5866$ \\
\quad bias-corrected & $0.0011$ & $0.1541$ & $0.1542$ & $0.1496$ & $0.947$ & $0.5865$ \\
\addlinespace
NN, $M = n_{gs}^{1/3}$ & $0.0461$ & $0.1386$ & $0.1307$ & $0.1257$ & $0.922$ & $0.4927$ \\
\quad bias-corrected & $-0.0003$ & $0.1345$ & $0.1346$ & $0.1257$ & $0.942$ & $0.4927$ \\
\addlinespace
Kernel & $0.0206$ & $0.1409$ & $0.1394$ & $0.1353$ & $0.942$ & $0.5303$ \\
\quad bias-corrected & $0.0020$ & $0.1407$ & $0.1408$ & $0.1353$ & $0.941$ & $0.5303$ \\
\bottomrule
\end{tabular}
\par\smallskip
\begin{minipage}{\linewidth}\footnotesize
\emph{Notes:} This table reports the simulation results for design 3 with a small $(1,0)$ subgroup and homogeneous treatment effect. We run $1,000$ Monte Carlo replications, each with $n = 5,000$ units, $q = 2$ covariates, and $4$ periods. The treatment effect is $\tau_0 = 1$. The two OLS estimators are estimated using the two regression specifications studied in section \ref{section:empiricaldiagnose}. The three matching estimators use (i) nearest-neighbor matching with one neighbor, (ii) nearest-neighbor matching with a diverging number of neighbors, set to $n_{gs}^{1/3}$, and (iii) the kernel matching estimator with the Epanechnikov kernel with a bandwidth rule based on Stata's kmatch implementation, which is set to be $1.5$ times the $90$th percentile of non-zero pairwise distances from matching with replacement. The three bias-corrected matching estimators use the corresponding matching methods, and the bias correction is performed using a second-order series estimator.
\end{minipage}
\end{table}

\begin{table}[htbp]
\centering
\caption{Small $(1,0)$ Subgroup, Heterogeneous Treatment Effect}
\label{table:design4}
\small
\begin{tabular}{lrrrrrr}
\toprule
& \multicolumn{1}{c}{Bias} & \multicolumn{1}{c}{RMSE} & \multicolumn{1}{c}{MCSD} & \multicolumn{1}{c}{Avg.~SE} & \multicolumn{1}{c}{Coverage} & \multicolumn{1}{c}{CI length} \\
\midrule
Cov. Trend ($X'\theta_t$) & $1.2256$ & $1.2709$ & $0.3361$ & $0.3394$ & $0.055$ & $1.3304$ \\
Cov. Level ($X'\theta$) & $0.6947$ & $0.7501$ & $0.2830$ & $0.2883$ & $0.336$ & $1.1300$ \\
\midrule
NN, $M = 1$ & $-0.0005$ & $0.2714$ & $0.2715$ & $0.2708$ & $0.949$ & $1.0615$ \\
\quad bias-corrected & $-0.0128$ & $0.2722$ & $0.2720$ & $0.2705$ & $0.946$ & $1.0604$ \\
\addlinespace
NN, $M = n_{gs}^{1/3}$ & $0.0322$ & $0.2613$ & $0.2595$ & $0.2590$ & $0.954$ & $1.0152$ \\
\quad bias-corrected & $-0.0142$ & $0.2611$ & $0.2608$ & $0.2579$ & $0.954$ & $1.0111$ \\
\addlinespace
Kernel, & $0.0067$ & $0.2632$ & $0.2633$ & $0.2632$ & $0.951$ & $1.0316$ \\
\quad bias-corrected & $-0.0119$ & $0.2640$ & $0.2638$ & $0.2629$ & $0.953$ & $1.0304$ \\
\bottomrule
\end{tabular}
\par\smallskip
\begin{minipage}{\linewidth}\footnotesize
\emph{Notes:} This table reports the simulation results for design 4 with a small $(1,0)$ subgroup and heterogeneous treatment effect. We run $1,000$ Monte Carlo replications, each with $n = 5,000$ units, $q = 2$ covariates, and $4$ periods. The treatment effect is $\tau_0 = 3.13$. The two OLS estimators are estimated using the two regression specifications studied in section \ref{section:empiricaldiagnose}. The three matching estimators use (i) nearest-neighbor matching with one neighbor, (ii) nearest-neighbor matching with a diverging number of neighbors, set to $n_{gs}^{1/3}$, and (iii) the kernel matching estimator with the Epanechnikov kernel with a bandwidth rule based on Stata's kmatch implementation, which is set to be $1.5$ times the $90$th percentile of non-zero pairwise distances from matching with replacement. The three bias-corrected matching estimators use the corresponding matching methods, and the bias correction is performed using a second-order series estimator.
\end{minipage}
\end{table}

\newpage
\section{Empirical Application}\label{section:empiricalapplication}
In this section, we revisit our motivating example from \cite{Cai_2016_AEJ}, who studies how the introduction of an agricultural insurance program affects rural households' borrowing and saving behavior in $12$ tobacco-producing counties in Jiangxi Province, China. In these counties, some rural households rely on tobacco cultivation as their main source of income (hereafter, tobacco households), while others rely on other crops such as rice (hereafter, non-tobacco households). In 2003, one of the 12 counties, Guangchang County (hereafter, the exposed county), collaborated with the People's Insurance Company of China (PICC) to design and launch a tobacco production insurance program. Under the program, a household was eligible for an insurance payout if a covered weather disaster reduced its tobacco yield by at least $30\%$. Only tobacco households in the exposed county were offered the insurance; non-tobacco households in the exposed county and households in the other counties were not eligible to purchase it. Because purchase was compulsory, all tobacco households in the exposed county were covered.

\cite{Cai_2016_AEJ} uses a household-level panel of 5,746 households from the 12 counties, including more than 3,000 tobacco households (approximately $1,200$ of them in the exposed county) and about $2,200$ non-tobacco households. The dataset has two parts. The first part is an administrative dataset on household saving and borrowing information from the Rural Credit Cooperative (RCC), the main rural bank in China. The second part is RCC annual survey data on the background information of those households such as demographic information about household heads, agricultural land allocation information, and revenue from the crop production. The panel covers the years 2001 to 2008; because the insurance program was implemented in 2003, it contains two pre-treatment years (2001--2002), and six post-treatment years (2003--2008). 

\cite{Cai_2016_AEJ} uses a triple-differences design to estimate the effect of the program on household borrowing and saving. As discussed in our motivating example in Section \ref{section:econometricssetup}, a triple-differences design accounts for both the county-specific and industry-specific shocks, either of which could confound a standard difference-in-differences design. In our empirical application, we focus on the dynamic effects of insurance provision on the saving behavior of tobacco households in the exposed county. We use the same three saving outcomes as \cite{Cai_2016_AEJ}: net saving, defined as the annual increase in total savings; the saving rate, defined as the ratio of net saving to current household income; and the flexible-term saving ratio, defined as the ratio of net saving in checking accounts to total net saving.

We first analyze the three outcomes using an event-study regression similar to that in \cite{Cai_2016_AEJ}, which corresponds to the covariate-level 3WFE specification in Section \ref{section:empiricaldiagnose}:
\begin{align*}
   Y_{it}
= \lambda_{G_iS_i}
+ \lambda_{G_i t}
+ \lambda_{S_i t}
+ \sum_{l\neq -1}
    \tau_l^{CL}\, G_iS_i\1(t-t^*=l)
+ X_i'\theta
+ u_{it},
\end{align*}
where $\lambda_{G_iS_i}$, $\lambda_{G_i t}$, and $\lambda_{S_i t}$ are subgroup, exposed-county-by-time, and tobacco-by-time fixed effects, respectively, and $X_i$ is a vector of time-invariant household characteristics, including the age and education level of the household head and household size. 

We then implement our matched triple-differences estimator and construct confidence intervals using the variance estimator proposed in Section \ref{section:matchedtriplediffereces}. We match exactly on the education level of the household head and treat the head's age and household size as continuous variables because they take many values. Our matching method is bias-corrected kernel matching with an Epanechnikov kernel. The bandwidth follows the rule implemented in Stata's kmatch: $1.5$ times the $90$th percentile of the non-zero pairwise distances from one-to-one nearest-neighbor matching with replacement. The bias correction uses a second-order series estimator. For both the 3WFE regression and the matching estimator, the baseline year is $2002$, the year before the program was implemented.

Figure \ref{figure:empiricalcaikernelbc} plots the event-study estimators from the 3WFE regression and kernel matching\footnote{In appendix \ref{appendix:additionalempirical}, we also report the estimates from the kernel matching estimator without bias correction and bias-corrected nearest-neighbor matching estimator with a diverging number of neighbors. The main results are not sensitive to the choice of matching methods and whether the bias-correction is applied.}.  For net saving (Panel A), the two estimators lead to qualitatively different conclusions. Specifically, the 3WFE regression finds an sign-reversing pattern that the program increases net saving at years $0$ and $3$ and then decreases it at year $5$. In contrast, the kernel matching estimator finds no statistically significant effect on net saving in any post-treatment years, which is more plausible. Recall from Section \ref{section:empiricaldiagnose} that this 3WFE specification is consistent if and only if the parallel gaps assumption holds unconditionally. This suggests that, for the net saving, the parallel gaps assumption may be more plausible after conditioning on household characteristics. For the other two outcomes (Panels B and C), the kernel matching estimator and the 3WFE regression estimator yield similar results, although the kernel matching estimator finds larger effects on the flexible-term saving ratio in years $2$ and $3$. Overall, the matched triple-differences estimates can differ from the 3WFE estimates both in magnitude and in their qualitative conclusions.

\begin{figure}[p]
\centering
\setstretch{1}
\setlength{\parindent}{0pt}
\caption{Event-study estimates: 3WFE regression and matched triple-differences estimator}
\label{figure:empiricalcaikernelbc}
\vspace{0.04in}
\begin{tikzpicture}[every node/.style={font=\footnotesize}]
\begin{axis}[
    width=0.78\linewidth, height=0.24\textheight,
    xmin=-2.35, xmax=5.35, ymin=-0.40027801463779, ymax=0.34504899909701,
    xtick={-2,-1,0,1,2,3,4,5},
    xlabel={Event time (year $-$ 2003)}, ylabel={Treatment effect},
    axis lines=left, axis line style={black,thin},
    tick align=outside, tick style={black,thin},
    tick label style={font=\footnotesize},
    label style={font=\footnotesize},
    xlabel style={at={(axis description cs:0.5,-0.17)},anchor=north},
    ylabel style={at={(axis description cs:-0.11,0.5)},anchor=south},
    scaled y ticks=false,
    yticklabel style={/pgf/number format/fixed,/pgf/number format/precision=3},
    legend style={at={(1.02,0.5)},anchor=west,draw=none,
                  font=\footnotesize,inner sep=1pt},
    legend cell align=left, legend columns=1,
    clip=true,
]
\addplot[black,thin,no marks,forget plot] coordinates {(-2.35,0) (5.35,0)};
\addplot[black,dashed,thin,no marks,forget plot] coordinates {(-1,-0.40027801463779) (-1,0.34504899909701)};

\addplot[
    only marks,mark=*,mark size=2.1pt,color={rgb,255:red,0;green,76;blue,153},
    error bars/.cd,y dir=both,y explicit,
    error bar style={line width=0.65pt},error mark options={rotate=90,mark size=2.3pt,line width=0.65pt}
] coordinates {
(-2.065,0.03278396501074) += (0,0.033038353351182) -= (0,0.033038353351182)
(-0.065,0.07723661597553) += (0,0.051545595431263) -= (0,0.051545595431263)
(0.935,-0.043751965344329) += (0,0.040673623325812) -= (0,0.040673623325812)
(1.935,0.032237012146934) += (0,0.053347809186021) -= (0,0.053347809186021)
(2.935,0.17860639235518) += (0,0.099240990749347) -= (0,0.099240990749347)
(3.935,0.010624681265614) += (0,0.084083444079536) -= (0,0.084083444079536)
(4.935,-0.19722942837161) += (0,0.12107520493823) -= (0,0.12107520493823)
};
\addlegendentry{Regression}
\addplot[forget plot,only marks,mark=*,mark size=2.1pt,color={rgb,255:red,0;green,76;blue,153}] coordinates {(-1.065,0)};

\addplot[
    only marks,mark=triangle*,mark size=2.1pt,color={rgb,255:red,190;green,52;blue,44},
    error bars/.cd,y dir=both,y explicit,
    error bar style={line width=0.65pt},error mark options={rotate=90,mark size=2.3pt,line width=0.65pt}
] coordinates {
(-1.935,-0.0072926863016265) += (0,0.038654032950374) -= (0,0.038654032950374)
(0.065,0.0033408336817603) += (0,0.052404924595985) -= (0,0.052404924595985)
(1.065,-0.0019253958329706) += (0,0.047644232644645) -= (0,0.047644232644645)
(2.065,0.015527154522128) += (0,0.063152640502042) -= (0,0.063152640502042)
(3.065,-0.012531710609) += (0,0.090321288477436) -= (0,0.090321288477436)
(4.065,-0.056798746382658) += (0,0.078135846351844) -= (0,0.078135846351844)
(5.065,-0.075814057075813) += (0,0.11710117087035) -= (0,0.11710117087035)
};
\addlegendentry{Matched DDD}
\addplot[forget plot,only marks,mark=triangle*,mark size=2.1pt,color={rgb,255:red,190;green,52;blue,44}] coordinates {(-0.935,0)};

\end{axis}
\end{tikzpicture}\par
\vspace{-0.02in}
{\footnotesize Panel A: Net saving}\par
\vspace{0.05in}

\begin{tikzpicture}[every node/.style={font=\footnotesize}]
\begin{axis}[
    width=0.78\linewidth, height=0.24\textheight,
    xmin=-2.35, xmax=5.35, ymin=-0.13464693760165, ymax=0.066716038222654,
    xtick={-2,-1,0,1,2,3,4,5},
    xlabel={Event time (year $-$ 2003)}, ylabel={Treatment effect},
    axis lines=left, axis line style={black,thin},
    tick align=outside, tick style={black,thin},
    tick label style={font=\footnotesize},
    label style={font=\footnotesize},
    xlabel style={at={(axis description cs:0.5,-0.17)},anchor=north},
    ylabel style={at={(axis description cs:-0.11,0.5)},anchor=south},
    scaled y ticks=false,
    yticklabel style={/pgf/number format/fixed,/pgf/number format/precision=3},
    legend style={at={(1.02,0.5)},anchor=west,draw=none,
                  font=\footnotesize,inner sep=1pt},
    legend cell align=left, legend columns=1,
    clip=true,
]
\addplot[black,thin,no marks,forget plot] coordinates {(-2.35,0) (5.35,0)};
\addplot[black,dashed,thin,no marks,forget plot] coordinates {(-1,-0.13464693760165) (-1,0.066716038222654)};

\addplot[
    only marks,mark=*,mark size=2.1pt,color={rgb,255:red,0;green,76;blue,153},
    error bars/.cd,y dir=both,y explicit,
    error bar style={line width=0.65pt},error mark options={rotate=90,mark size=2.3pt,line width=0.65pt}
] coordinates {
(-2.065,-0.0083879405501005) += (0,0.023449453918368) -= (0,0.023449453918368)
(-0.065,0.023063620532447) += (0,0.022853104228507) -= (0,0.022853104228507)
(0.935,0.014396292948598) += (0,0.025536221754226) -= (0,0.025536221754226)
(1.935,-0.032938609492962) += (0,0.030063560432264) -= (0,0.030063560432264)
(2.935,-0.024354776486304) += (0,0.028987026928152) -= (0,0.028987026928152)
(3.935,-0.04001878515996) += (0,0.03117962236484) -= (0,0.03117962236484)
(4.935,-0.065892770571912) += (0,0.030702135926055) -= (0,0.030702135926055)
};
\addlegendentry{Regression}
\addplot[forget plot,only marks,mark=*,mark size=2.1pt,color={rgb,255:red,0;green,76;blue,153}] coordinates {(-1.065,0)};

\addplot[
    only marks,mark=triangle*,mark size=2.1pt,color={rgb,255:red,190;green,52;blue,44},
    error bars/.cd,y dir=both,y explicit,
    error bar style={line width=0.65pt},error mark options={rotate=90,mark size=2.3pt,line width=0.65pt}
] coordinates {
(-1.935,-0.0059915426586062) += (0,0.023468253643622) -= (0,0.023468253643622)
(0.065,0.0148217072547) += (0,0.02394563832467) -= (0,0.02394563832467)
(1.065,0.01115741914584) += (0,0.026431346625161) -= (0,0.026431346625161)
(2.065,-0.039257319286807) += (0,0.03176410827417) -= (0,0.03176410827417)
(3.065,-0.02801425477034) += (0,0.031174009477146) -= (0,0.031174009477146)
(4.065,-0.04102811553772) += (0,0.031813313816065) -= (0,0.031813313816065)
(5.065,-0.066929290403756) += (0,0.031769118306759) -= (0,0.031769118306759)
};
\addlegendentry{Matched DDD}
\addplot[forget plot,only marks,mark=triangle*,mark size=2.1pt,color={rgb,255:red,190;green,52;blue,44}] coordinates {(-0.935,0)};

\end{axis}
\end{tikzpicture}\par
\vspace{-0.02in}
{\footnotesize Panel B: Saving rate}\par
\vspace{0.05in}

\begin{tikzpicture}[every node/.style={font=\footnotesize}]
\begin{axis}[
    width=0.78\linewidth, height=0.24\textheight,
    xmin=-2.35, xmax=5.35, ymin=-0.12680145612051, ymax=0.24111493392685,
    xtick={-2,-1,0,1,2,3,4,5},
    xlabel={Event time (year $-$ 2003)}, ylabel={Treatment effect},
    axis lines=left, axis line style={black,thin},
    tick align=outside, tick style={black,thin},
    tick label style={font=\footnotesize},
    label style={font=\footnotesize},
    xlabel style={at={(axis description cs:0.5,-0.17)},anchor=north},
    ylabel style={at={(axis description cs:-0.11,0.5)},anchor=south},
    scaled y ticks=false,
    yticklabel style={/pgf/number format/fixed,/pgf/number format/precision=3},
    legend style={at={(1.02,0.5)},anchor=west,draw=none,
                  font=\footnotesize,inner sep=1pt},
    legend cell align=left, legend columns=1,
    clip=true,
]
\addplot[black,thin,no marks,forget plot] coordinates {(-2.35,0) (5.35,0)};
\addplot[black,dashed,thin,no marks,forget plot] coordinates {(-1,-0.12680145612051) (-1,0.24111493392685)};

\addplot[
    only marks,mark=*,mark size=2.1pt,color={rgb,255:red,0;green,76;blue,153},
    error bars/.cd,y dir=both,y explicit,
    error bar style={line width=0.65pt},error mark options={rotate=90,mark size=2.3pt,line width=0.65pt}
] coordinates {
(-2.065,-0.031466247576672) += (0,0.040632916047508) -= (0,0.040632916047508)
(-0.065,0.006780677822692) += (0,0.04200750436697) -= (0,0.04200750436697)
(0.935,0.030510667282213) += (0,0.04027750711002) -= (0,0.04027750711002)
(1.935,0.049415648042741) += (0,0.048055316651295) -= (0,0.048055316651295)
(2.935,0.045013756686105) += (0,0.045749248829915) -= (0,0.045749248829915)
(3.935,0.070427588164144) += (0,0.052713895580232) -= (0,0.052713895580232)
(4.935,0.14103572272288) += (0,0.053576951374135) -= (0,0.053576951374135)
};
\addlegendentry{Regression}
\addplot[forget plot,only marks,mark=*,mark size=2.1pt,color={rgb,255:red,0;green,76;blue,153}] coordinates {(-1.065,0)};

\addplot[
    only marks,mark=triangle*,mark size=2.1pt,color={rgb,255:red,190;green,52;blue,44},
    error bars/.cd,y dir=both,y explicit,
    error bar style={line width=0.65pt},error mark options={rotate=90,mark size=2.3pt,line width=0.65pt}
] coordinates {
(-1.935,-0.026180834651252) += (0,0.047448326311303) -= (0,0.047448326311303)
(0.065,0.0055140263956653) += (0,0.05200625198073) -= (0,0.05200625198073)
(1.065,0.034105963281522) += (0,0.043579512138134) -= (0,0.043579512138134)
(2.065,0.082572994502872) += (0,0.050482823855256) -= (0,0.050482823855256)
(3.065,0.073966511152735) += (0,0.052421800101064) -= (0,0.052421800101064)
(4.065,0.081444440576871) += (0,0.055482952567311) -= (0,0.055482952567311)
(5.065,0.13609657334295) += (0,0.056392364553265) -= (0,0.056392364553265)
};
\addlegendentry{Matched DDD}
\addplot[forget plot,only marks,mark=triangle*,mark size=2.1pt,color={rgb,255:red,190;green,52;blue,44}] coordinates {(-0.935,0)};

\end{axis}
\end{tikzpicture}\par
\vspace{-0.02in}
{\footnotesize Panel C: Flexible-term saving ratio}\par
\vspace{0.05in}

\vspace{0.02in}
\begin{minipage}{\linewidth}
\scriptsize
\textit{Notes:} This figure plots event-study estimates from the covariate-level 3WFE regression and from the matched triple-differences estimator with bias-corrected kernel matching. Event time $-1$ (2002) is the reference period. Bars show 95\% confidence intervals. The covariates included in the analysis are: the age and education level of the household head and household size. Cross-sectional sample sizes for panels A, B, and C are 3,736 (29,888 observations), 3,332 households (26,656 observations), and 3,545 households (28,360 observations), respectively. Matching uses an Epanechnikov kernel and bias correction uses the second-order series estimator.
\par
\vspace{0.25\baselineskip}
\end{minipage}
\end{figure}
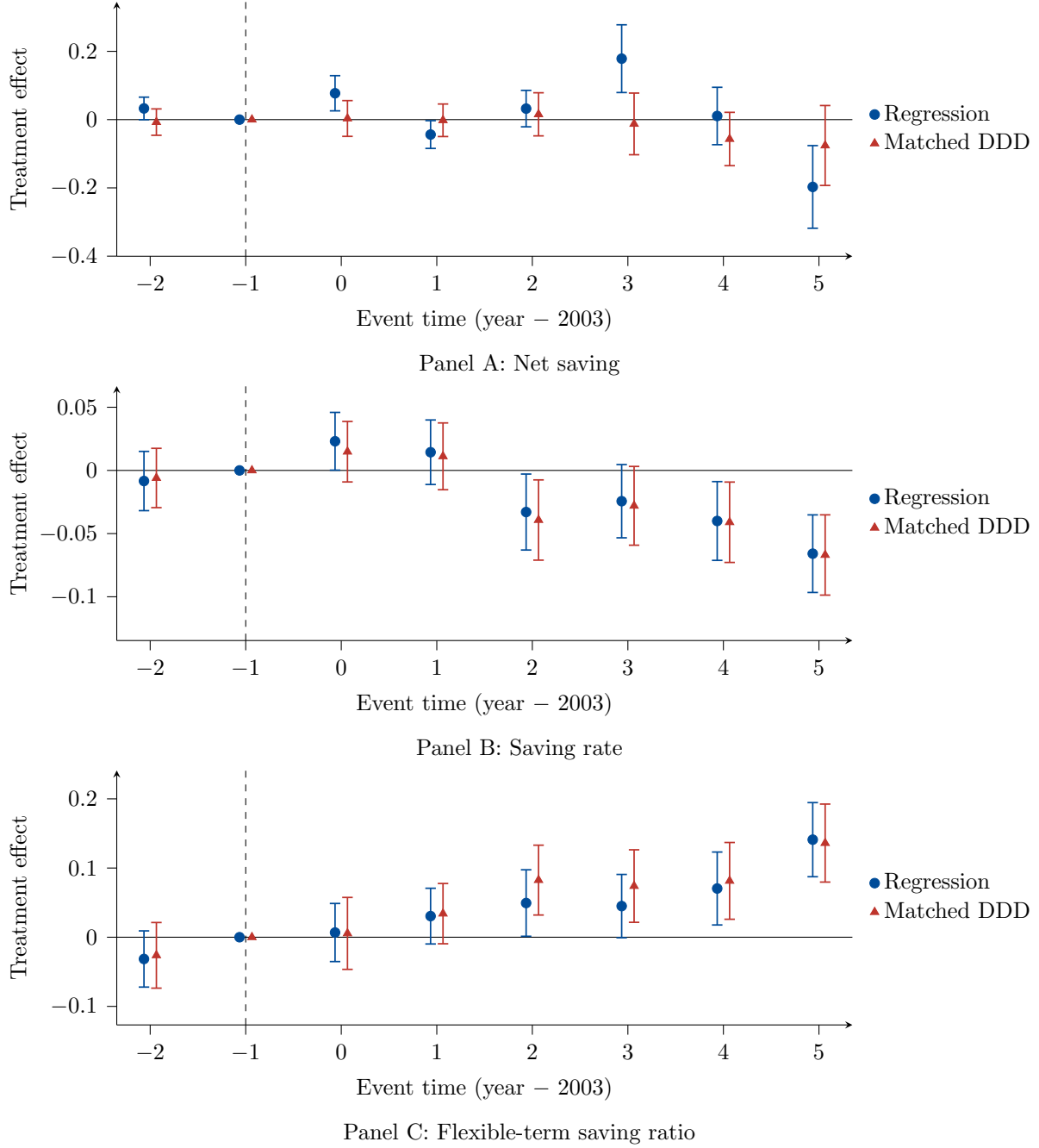

\newpage

\section{Conclusion}\label{section:conclusion}
We formally analyze two commonly used empirical strategies in triple-differences designs: augmenting the 3WFE regression specification with either covariate trends or covariate levels. These strategies are often used when the parallel gaps assumption is believed to hold conditionally on covariates. Our analysis shows that both generally lead to inconsistent estimators of the ATT. First, we show that the weights of the regression estimands for both specifications fail to align the covariate distribution of each untreated subgroup with that of the treated subgroup so neither satisfies the covariate balancing condition under which the estimands would equal the ATT. Second, we characterize the additional assumptions under which each estimand equals the ATT. For the covariate-trend specification, the sufficient conditions require the treatment effect to be homogeneous in the covariates, together with a strong parametric assumption on either the conditional mean function or the propensity score. For the covariate-level specification, the necessary and sufficient condition is that the parallel gaps assumption holds unconditionally, in which case there is no need to add covariates.

To address these limitations, we propose a matched triple-differences framework for valid estimation and inference. Our two-step estimation procedure first conducts three separate pairwise matchings to satisfy the covariate balancing condition and then runs a weighted 3WFE regression, which yields the matched triple-differences estimator. We then develop a novel asymptotic framework to establish the consistency and asymptotic normality of our estimator for a class of matching procedures, including kernel matching and nearest neighbor matching with a fixed or diverging number of neighbors. In addition, we propose a bias-corrected matched triple-differences estimator and a valid variance estimator that accounts for the matching step. Monte Carlo simulations confirm our theoretical results and show that the proposed estimators and variance estimator perform well in finite samples. Revisiting \cite{Cai_2016_AEJ}, we find that, compared with the regression specification used in that paper, our matched triple-differences estimators can yield treatment effect estimates of a different magnitude and lead to different qualitative conclusions.

For future work, we envision extending our proposed matching estimators to settings with staggered adoption and to more flexible sampling schemes, such as repeated cross sections. Another interesting direction is to adapt recent developments in the matching literature, such as \cite{Lin-Ding-Han_2023_ECMA} and \cite{Lin_Han_2025_JOE}, to establish the double robustness of our proposed bias-corrected matching estimator and to characterize when it attains the semiparametric efficiency bound.

\bibliography{matching_DDD_references}

@article{Abadie-Imbens_2006_ECMA,
  title={Large sample properties of matching estimators for average treatment effects},
  author={Abadie, Alberto and Imbens, Guido W.},
  journal={Econometrica},
  volume={74},
  number={1},
  pages={235--267},
  year={2006},
  publisher={Wiley}
}

@article{Abadie-Imbens_2011_JBES,
  title={Bias-corrected matching estimators for average treatment effects},
  author={Abadie, Alberto and Imbens, Guido W.},
  journal={Journal of Business \& Economic Statistics},
  volume={29},
  number={1},
  pages={1--11},
  year={2011},
  publisher={Taylor \& Francis}
}

@article{Abadie-Spiess_2022_JASA,
  title={Robust post-matching inference},
  author={Abadie, Alberto and Spiess, Jann},
  journal={Journal of the American Statistical Association},
  volume={117},
  number={538},
  pages={983--995},
  year={2022},
  publisher={Taylor & Francis}
}

@article{Blandol_et_al_2026_Restud,
    author  = {Christine Blandhol and John Bonney and Magne Mogstad and Alexander Torgovitsky},
  title   = {When is {TSLS} Actually {LATE}?},
  journal = {The Review of Economic Studies},
  volume = {},
  number = {},
  pages   = {},
  year    = {2026},
  publisher={Oxford Academic}
}

@article{Caetano-Callaway_2026_wp,
  title={Difference-in-Differences when Parallel Trends Holds Conditional on Covariates},
  author={Caetano, Carolina and Callaway, Brantly},
  journal={Working Paper},
  volume={},
  number={},
  pages={},
  year={2026},
  publisher={Arxiv}
}

@article{Cai_2016_AEJ,
  title={Insurance impact on production and financial decisions},
  author={Jing Cai},
  journal={American Economic Journal: Economic Policy},
  volume={8},
  number={2},
  pages={44-88},
  year={2016},
  publisher={American Economic Association}
}

@article{Chaise-Dhaut_2022_EJ,
    author = {Chaisemartin, Cl{\'e}ment de and D’Haultf{\oe}uille, Xavier},
    title = {Two-way fixed effects and differences-in-differences with heterogeneous treatment effects: a survey},
    journal = {The Econometrics Journal}, year = {2022}, month = {06}
}

@article{Chatt_Zubi_2023_biometrika,
  title={On the implied weights of linear regression for causal inference},
  author={Chattopadhyay, Ambarish and Zubizarreta, Jose R},
  journal={Biometrika},
  volume={110},
  number={3},
  pages={615-629},
  year={2023},
  publisher={Oxford Academic}
}

@article{Gruber_1994_AER,
  title={The Incidence of Mandated Maternity Benefits},
  author={Gruber, Jonathan},
  journal={American Economic Review},
  volume={84},
  number={3},
  pages={622-641},
  year={1994},
  publisher={American Economic Association}
}

@article{Heckman-Ichimura-Todd_1997_ReStud,
    author = {Heckman, James J. and Ichimura, Hidehiko and Todd, Petra E.},
    title = {Matching As An Econometric Evaluation Estimator: Evidence from Evaluating a Job Training Programme},
    journal = {The Review of Economic Studies}, volume = {64}, number = {4}, pages = {605-654}, year = {1997}, month = {10}
}

@article{Heckman-etal_1998_ECMA,
  title={Characterizing selection bias using experimental data},
  author={Heckman, James J. and Ichimura, Hidehiko and Smith, Jeffrey, A. and Todd Petra},
  journal={Econometrica},
  volume={66},
  number={5},
  pages={1017-1098},
  year={1998},
  publisher={Wiley}
}

@article{Leventer_2025_wp,
  title={Conditional Triple Difference-in-Differences},
  author={Leventer, Dor},
  journal={Working Paper},
  volume={},
  number={},
  pages={},
  year={2025},
  publisher={arXiv}
}

@article{Lin-Ding-Han_2023_ECMA,
  author = {Lin, Zhexiao and Ding, Peng and Han, Fang},
  title = {Estimation Based on Nearest Neighbor Matching: From Density Ratio to Average Treatment Effect},
  journal = {Econometrica},
  volume = {91},
  number = {6},
  pages = {2187-2217},
  year = {2023}
}

@article{Lin_Han_2025_JOE,
  author = {Lin, Zhexiao and Han, Fang},
  title = {On regression-adjusted imputation estimators of average treatment effects},
  journal = {Journal of Econometrics},
  volume = {251},
  number = {},
  pages = {106080},
  year = {2025}
}

@article{Liu_Vazquez_Bare_2026_wp,
  title={Post-Matching Two-Way Fixed Effects Estimation},
  author={Liu, Yihong and Vazquez-Bare, Gonzalo},
  journal={Working Paper},
  volume={},
  number={},
  pages={},
  year={2026},
  publisher={arXiv}
}

@article{Meng_et_al_2026_wp,
  title={Simple and Reliable Inference for Matching Estimators: A Single-Matching Variance Approach},
  author={Meng, Xiang and Smith, Aaron and Miratrix, Luke},
  journal={Working Paper},
  volume={},
  number={},
  pages={},
  year={2026},
  publisher={}
}

@article{Olden_Moen_2022_EJ,
  author = {Olden, Andreas and M{\o}en, Jarle},
  title = {The Triple Difference Estimator},
  journal = {The Econometrics Journal},
  volume = {25},
  number = {3},
  pages = {531-553},
  year = {2022}
}

@article{Ortiz-Villavicencio_SantAnna_2025_working,
  author = {Ortiz-Villavicencio, Marcelo and Sant'Anna, Pedro H.C.},
  title = {Better Understanding Triple Differences Estimators},
  journal = {arXiv},
  volume = {},
  number = {},
  pages = {1-60},
  year = {2025}
}

@article{Otsu_Rai_2017_JASA,
  author = {Otsu, Taisuke and Rai, Yoshiyasu},
  title = {Bootstrap Inference of Matching Estimators for Average Treatment Effects},
  journal = {Journal of the American Statistical Association},
  volume = {112},
  number = {520},
  pages = {1720-1732},
  year = {2017}
}

@article{Roth-etal_2023_JoE,
  title={What’s trending in difference-in-differences? A synthesis of the recent econometrics literature},
  author={Roth, Johnathan and Sant'Anna, Pedro HC and Bilinski, Alyssa and Poe, John},
  journal={Journal of Econometrics},
  volume={235},
  number={2},
  pages={2218-2244},
  year={2023},
  publisher={Elsevier}
}

@article{Sloczynski_2022_Restat,
  title = {Interpreting OLS estimands when treatment effects are heterogeneous: smaller groups get larger weights},
  author = {Tymon, S{\l}oczy\'nski},
  journal = {Review of Economics and Statistics},
  volume = {104},
  number = {3},
  pages = {501-509},
  year = {2022},
  publisher={MIT Press}
}

@article{Sloczynski_2026_Restud,
  title = {When Should We (Not) Interpret Linear IV Estimands as LATE?},
  author = {Tymon, S{\l}oczy\'nski},
  journal = {Working Paper},
  volume = {},
  number = {},
  pages = {},
  year = {2026},
  publisher={arXiv},
}

@article{Strezhnev_2023_wp,
  title={Decomposing triple-differences regression under staggered adoption},
  author={Strezhnev, Anton},
  journal={Working Paper},
  volume={},
  number={},
  pages={},
  year={2023},
  publisher={arXiv}
}

@article{Viel_2025_wp,
  title={Convergence rate for Nearest Neighbour matching: geometry of the domain and higher-order regularity},
  author={Viel, Simon and Truquet, Lionel and Yamane, Ikko},
  journal={Working Paper},
  volume={},
  number={},
  pages={},
  year={2025},
  publisher={arXiv}
}

\newpage

\section*{Appendix}
\appendix

\numberwithin{table}{section}
\renewcommand{\thetable}{\thesection\arabic{table}}
\renewcommand{\theHtable}{appendix.\thesection.\arabic{table}}

\section{Survey of Empirical Literature}

Table \ref{table:ddd_literature} lists 18 papers published between 2010 and 2026 in top economics journals using the two types of 3WFE regression specifications analyzed in this paper. The journals we consider are, in alphabetical order: American Economic Journal: Applied Economics (AEJAE), American Economic Journal: Economic Policy (AEJEP), American Economic Review (AER), Economic Journal (EJ), Journal of Public Economics (JPubE), Quarterly Journal of Economics (QJE), Review of Economics and Statistics (REStat), Review of Economic Studies (REStud). The column $X_i'\theta$ indicates that the paper add the covariate levels to the 3WFE regression and column $\sum_{r=1}^T X_i'\theta_r\1(t=r) $ indicates that the paper adds covariate trends to the 3WFE regression.

\begin{sidewaystable}[htbp]
\centering
\caption{Survey of Empirical Literature}
\label{table:ddd_literature}

\begin{tabular}{p{0.4cm} p{4.2cm} p{9.5cm} p{1.3cm} p{2.0cm} p{1.0cm} p{4.0cm}}
\hline
& Author(s) & Title & Year & Journal & $X_i'\theta$ &  $\sum_{r=1}^T X_i'\theta_r\1(t=r) $ \\
\hline
1 & Akosa Antwi, Moriya, Simon
& Effects of Federal Policy to Insure Young Adults: Evidence from the 2010 Affordable Care Act's Dependent-Coverage Mandate
& 2013 & AEJEP & \checkmark & \\[0.6em]

2 & Bahaj and Reis
& Central Bank Swap Lines: Evidence on the Effects of the Lender of Last Resort
& 2022 & REStud & & \checkmark \\[0.6em]

3 & Cabral, Carey, Miller
& The Impact of Provider Payments on Health Care Utilization of Low-Income Individuals: Evidence from Medicare and Medicaid
& 2025 & AEJEP & \checkmark & \\[0.6em]

4 & Cai
& The Impact of Insurance Provision on Household Production and Financial Decisions
& 2016 & AEJEP &\checkmark & \\[0.6em]

5 & Clemens and Wither
& The Minimum Wage and the Great Recession: Evidence of Effects on the Employment and Income Trajectories of Low-Skilled Workers
& 2019 & JPubE & & \checkmark \\[0.6em]

6 & Ellen, Hartley, Lin, You
& The Bronx Is Burning: Urban Disinvestment Effects of Fair Access to Insurance Requirements Plans
& 2026 & REStat (forthcoming) & &\checkmark \\[0.6em]

7 & Garthwaite, Gross, Notowidigdo
& Public Health Insurance, Labor Supply, and Employment Lock
& 2014 & QJE &\checkmark & \\[0.6em]

8 & Ginja, Karimi, Xiao
& Employer Responses to Family Leave Programs
& 2023 & AEJAE &\checkmark &\\[0.6em]

9 & Harrison and Scorse
& Multinationals and Anti-Sweatshop Activism
& 2010 & AER & &\checkmark\\
\hline
\end{tabular}

\end{sidewaystable}

\begin{sidewaystable}[!htbp]
\centering
\addtocounter{table}{-1}
\caption{Survey of Empirical Literature (continued)}

\begin{tabular}{p{0.4cm} p{4.2cm} p{9.5cm} p{1.3cm} p{2.0cm} p{1.0cm} p{4.0cm}}
\hline
& Author(s) & Title & Year & Journal & $X_i'\theta$ &  $\sum_{r=1}^T X_i'\theta_r\1(t=r) $ \\
\hline
10 & Kantor and Whalley
& Moonshot: Public R\&D and Growth
& 2025 & AER & & \checkmark \\[0.6em]

11 & Kaplan, Piedra, Seira
& Entry Regulation and Business Start-ups: Evidence from Mexico
& 2011 & JPubE & &\checkmark \\[0.6em]

12 & Kuka and Shenhav
& Long-Run Effects of Incentivizing Work after Childbirth
& 2024 & AER &\checkmark &\checkmark \\[0.6em]

13 & Mastrorocco and Ornaghi
& Who Watches the Watchmen? Local News and Police Behavior in the United States
& 2025 & AEJEP & &\checkmark \\[0.6em]

14 & Milner
& The Impact of State-Provided Education: Evidence from the 1870 Education Act
& 2025 & EJ & &\checkmark \\[0.6em]

15 & Mitrunen
& War Reparations, Structural Change, and Intergenerational Mobility
& 2025 & QJE & &\checkmark \\[0.6em]

16 & Rabat\'{e}
& Can I Stay or Should I Go? Mandatory Retirement and the Labor-Force Participation of Older Workers
& 2019 & JPubE &\checkmark & \\[0.6em]

17 & Walker
& The Transitional Costs of Sectoral Reallocation: Evidence from the Clean Air Act and the Workforce
& 2013 & QJE & &\checkmark \\[0.6em]

18 & Wettstein
& Retirement Lock and Prescription Drug Insurance: Evidence from Medicare Part D
& 2020 & AEJEP & &\checkmark \\
\hline
\end{tabular}

\end{sidewaystable}

\section{Examples of Matching Procedures}\label{appendix:matchingprocedure}
\setcounter{assumption}{0}
\renewcommand{\theassumption}{B.\arabic{assumption}}
\renewcommand{\theHassumption}{B.\arabic{assumption}}

\subsection{Nearest Neighbor Matching with Fixed Neighbors}
In nearest-neighbor matching with fixed neighbors (\citeinb{Abadie-Imbens_2006_ECMA}), each treated unit is matched to its nearest $M$ (fixed)  neighbors in the untreated subgroup. For each $i\in\mathcal I_{11}$, the matched set $\mathcal M_{gs}(i)$ contains the $M$ comparison units whose covariate values are closest to $X_i$. Writing it formally,
$\mathcal{M}_{gs}(i)=\{j_1(i,g,s),\dots, j_M(i,g,s) \}$,
where $j_m(i,g,s)$ is the index of the $m$-th nearest neighbor of unit $i$ among units in the untreated group $(g,s)$
\[
\sum_{l:(G_l,S_l)=(g,s)}\mathbbm{1}\{\norm{X_l-X_i}\leq \norm{X_j-X_i}\}=m.
\]
The corresponding pairwise weight $w_{ij}^{fix}$ equals $\frac{1}{M}$ if $j\in\mathcal M_{gs}(i)$ and equals $0$ if $j\notin\mathcal M_{gs}(i)$:
\[
w_{ij}^{fix} = \frac{1}{M}\, \1\!\left\{j\in\mathcal M_{gs}(i)\right\},
\qquad j\in\mathcal I_{gs}.
\]
The matching weight of untreated unit $j$ is therefore
\[
w_{j}^{fix} = \frac{1}{M} \sum_{i\in\mathcal I_{11}} \1\!\left\{j\in\mathcal M_{gs}(i)\right\},
\]
that is, the number of times unit $j$ is used as a match, divided by $M$.

\subsection{Nearest Neighbor Matching with Diverging Neighbors}
Nearest neighbor matching with diverging neighbors (\cite{Lin-Ding-Han_2023_ECMA}) has the same matching algorithm and the corresponding pairwise weight $w_{ij}^{diverge}$ and matching weight $w_j^{diverge}$ as the fixed neighbors. The only differences is that the number of the neighbors $M$ is diverging with the untreated group sample size $n_{gs}$. To be specific $M$ satisfies the the following assumption.
\begin{assumption}\label{assumption:NNMdiverge}
    As $n\rightarrow \infty$,
    \[\frac{Mlog(n_{gs})}{n_{gs}}\rightarrow 0, \quad M\rightarrow \infty \]
\end{assumption}

\subsection{Kernel Matching}
Kernel matching (\cite{Heckman-Ichimura-Todd_1997_ReStud}, \cite{Heckman-etal_1998_ECMA}) each reference unit is matched to a kernel-weighted combination of comparison units, with the weights determined by the proximity of their covariate values. Writing formally, let $K:\mathbb R^q\rightarrow[0,\infty)$ be a multivariate kernel function and let $h>0$ be the bandwidth that depends on the the sample size of the untreated group $n_{gs}$. The pairwise weight correspond to kernel matching is:
\[
w_{ij}^{kernel}
= \frac{K\!\left(\frac{X_j - X_i}{h}\right)}
       {\sum_{l\in\mathcal I_{gs}} K\!\left(\frac{X_l - X_i}{h}\right)},
\qquad j\in\mathcal I_{gs}.
\]
For each $i\in\mathcal I_{11}$, the matched set consists of the comparison units
that receive positive kernel weight:
\[
\mathcal M_{gs}(i)
=
\left\{
j\in\mathcal I_{gs}:
K\!\left(\frac{X_j - X_i}{h}\right)>0
\right\}.
\]
The matching weight of comparison unit $j$ is therefore
\[
w_{j}^{kernel} = \sum_{i\in\mathcal I_{11}}
\frac{K\!\left(\frac{X_j - X_i}{h}\right)}
     {\sum_{l\in\mathcal I_{gs}} K\!\left(\frac{X_l - X_i}{h}\right)},
\]
We impose the following assumption on the Kernel function and the bandwidth, which are common in kernel matching literature:
\begin{assumption}\label{assumption:kernel}
    (i) The kernel function has a compact support. There exists $0<r<R<\infty$ such that the kernel function $K$ satisfies: $0 \leq K(u) \leq \overline{K},  $ $\underline{K} \leq K(u) \leq \overline{K}$, if $\norm{u}\leq r$, and $K(u)=0$ if $\norm{u}>R$, where $\underline{K}$ and $\overline{K}$ are positive real numbers.\\
    (ii) As $n\rightarrow \infty$,
    \[h\rightarrow 0, \quad n_{gs}h^q\rightarrow \infty \]
\end{assumption}

\begin{lemma}\label{lemma:matchingprocedure}
    Under assumption \ref{assumption:regularity}(i), we have the following results:\\
    (a) $w_{ij}^{fix}$ and $w_j^{fix}$ satisfy assumption \ref{assumption:matching}.\\
    (b) If assumption \ref{assumption:NNMdiverge} holds, then $w_{ij}^{diverge}$ and $w_j^{diverge}$ satisfies assumption \ref{assumption:matching}. \\
    (c) If assumption \ref{assumption:kernel} holds, then $w_{ij}^{kernel}$ and $w_j^{kernel}$ satisfy assumption \ref{assumption:matching}.
\end{lemma}

\section{Necessity of Condition (iv) in Assumption \ref{assumption:matching}}\label{section:example}
In this section, we construct a matching procedure that satisfies conditions (i)--(iii) of Assumption \ref{assumption:matching} but violates condition (iv), and we provide simulation evidence that the resulting matched triple-differences estimator is not asymptotically normal unconditionally.
 
Consider the following \emph{switch matching rule}. Let $T_n=\frac{\overline{X}_{11}}{\sqrt{\widehat{V}(\overline{X}_{11})/n_{11}}}$ denote the standardized sample mean of $X$ in the treated cell. If $\abs{T_n}<z_{0.75}$, we construct the matched triple-differences estimator using one-nearest-neighbor matching, following the procedure in Section \ref{section:matchedtriplediffereces}; if $\abs{T_n}\geq z_{0.75}$, we use fifty-nearest-neighbor matching instead. Here $z_{0.75}$ denotes the $75$th percentile of the standard normal distribution. Because $T_n$ is asymptotically standard normal, each of the two matching procedures is used in roughly half of all repeated samples.
 
By construction, the switch rule satisfies condition (i). The rule switches between two nearest-neighbor procedures, each of which satisfies conditions (ii) and (iii) by Lemma \ref{lemma:matchingprocedure}; the law of iterated expectations then implies that the switch rule satisfies conditions (ii) and (iii) as well. Condition (iv), however, is violated: one can show that, for some constant $C>0$,
\begin{align*}
    \E\!\left[\Bigg(\sum_{j\in\mathcal I_{gs}\setminus\{1\}}
\bigl|\,w_{j}-w_{j}^{(1)}\bigr|\Bigg)^{\!4}\right]\geq Cn^{7/2}.
\end{align*}
 
Next, we provide simulation evidence that the centered and normalized matched triple-differences estimator $(V_{\tau(X)}+V_{E})^{-1/2}\sqrt{n_{11}}\left(\widehat{\tau}_l-\tau_l-B\right)$ constructed under the switch rule is not asymptotically normal. The data-generating process follows Design 2 in Section \ref{section:simulation}, with two modifications. First, the logit models for $G$ and $S$ both have an intercept of $-1.5$ and slopes $\gamma_G=-\gamma_S=1$, so that the treated cell contains a smaller share of observations:
\begin{align*}
    P(G_i=1 \mid X_i)=\frac{\exp(-1.5+\gamma_G X_i)}{1+\exp(-1.5+\gamma_G X_i)},
\qquad
P(S_i=1 \mid X_i)=\frac{\exp(-2+\gamma_S X_i)}{1+\exp(-2+\gamma_S X_i)}.
\end{align*}
Second, the coefficient on the covariate-specific trend is reduced from $10$ to $2$, so that effect heterogeneity is smaller:
\begin{align*}
    Y_{it}(0)&=\alpha_i -2X_i t\,\1(G_i=1)-2X_i t\,\1(S_i=1)+u_{it}, \\
    Y_{it}(1)&=\alpha_i +2X_i t\,\1(G_i=1)+2X_i t\,\1(S_i=1)+v_{it}.
\end{align*}
The simulation uses a sample size of $10{,}000$ and $10{,}000$ Monte Carlo replications. We plot the empirical density of the centered and normalized estimator against the standard normal density and report the Kolmogorov--Smirnov test for normality.

Figure \ref{figure:exampledensity} shows a sizable discrepancy between the empirical density of the centered and normalized matched triple-differences estimator and the standard normal density. The Kolmogorov--Smirnov statistic is $0.0424$, with a $p$-value below $0.001$. Together, these results provide strong evidence that, under the switch matching rule, the matched triple-differences estimator is not asymptotically normal unconditionally.

\begin{figure}
    \centering
    \includegraphics[scale=0.2]{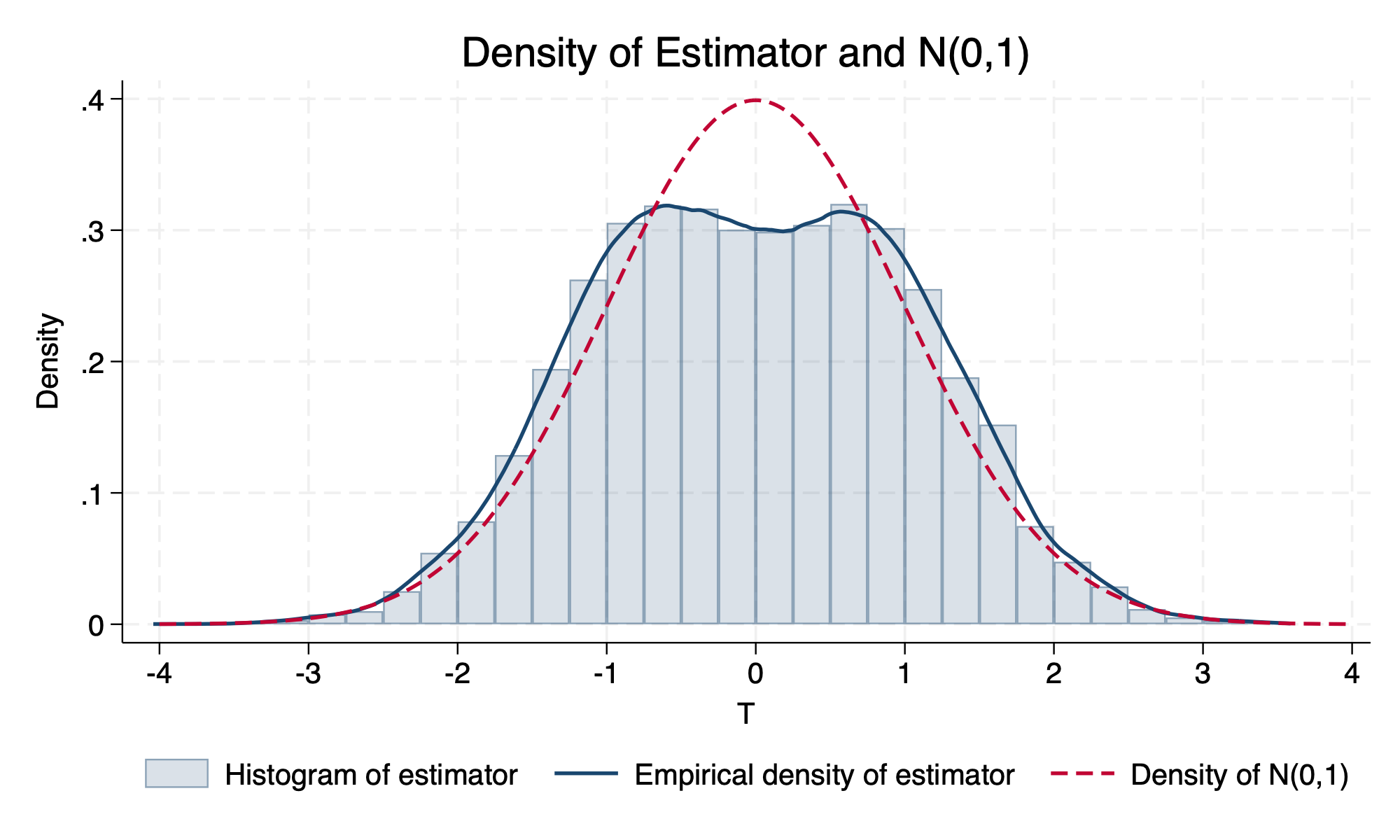}
    \caption{Density of the centered and normalized estimator under the switch rule versus $N(0,1)$}
    \label{figure:exampledensity}
\end{figure}

\section{Additional Assumptions for Bias Correction}\label{section:biascorrection}
\setcounter{assumption}{0}
\renewcommand{\theassumption}{D.\arabic{assumption}}
\renewcommand{\theHassumption}{D.\arabic{assumption}}
In this section, we present additional assumption for establishing the asymptotic normality for the bias-corrected matched triple-differences estimator.
 
\begin{assumption}[Bias correction]\label{assumption:biascorrection}
\begin{enumerate}[label=(\roman*)]
\item The support of $X$ is a Cartesian product of compact intervals.
\item For every $(g,s)\in \{(0,1),(1,0),(0,0)\}$ and every event time $l$, $\mu_l^{gs}$ is $p$ times continuously differentiable.
\item For every $(g,s)\in \{(0,1),(1,0),(0,0)\}$ and every event time $l$, the estimator $\widehat{\mu}_l^{gs}$ satisfies
    \begin{align*}
        \max_{m\in \Lambda_p }\norm{\partial^{m}\widehat{\mu}_l^{gs}}_{\infty}&=O_p(1), \\
        \max_{m\in \Lambda_r}\norm{\partial^{m}\widehat{\mu}_l^{gs}-\partial^{m}\mu_l^{gs}}_{\infty}&=O_p(n^{-\gamma_r})
        \quad\text{for every } r\in \{1,\dots,p-1 \},
    \end{align*}
    where, for $r\in \{1,\dots,p\}$, $\Lambda_r$ denotes the set of all $q$-dimensional vectors of non-negative integers $m=(m_1,\dots,m_q)$ with $\sum_{i=1}^q m_i=r$, and $\gamma_1,\dots,\gamma_{p-1}$ are positive constants.
\item For every $(g,s)\in \{(0,1),(1,0),(0,0)\}$, the pairwise weights in addition satisfy
    \begin{align*}
        \E\!\left[\sum_{j\in\mathcal I_{gs}} w_{ij}\left\lVert X_j-X_i\right\rVert^p\,\middle|\, i\in\mathcal I_{11}\right]& = o(n^{-1/2}), \\
        \E\!\left[\sum_{j\in\mathcal I_{gs}} w_{ij}\left\lVert X_j-X_i\right\rVert^r\,\middle|\, i\in\mathcal I_{11}\right] &= o(n^{-1/2+\gamma_r})\quad \text{for every } r\in \{1,\dots,p-1 \}.
    \end{align*}
\end{enumerate}
\end{assumption}
 
Assumption \ref{assumption:biascorrection} is similar to Assumption 3.5$'$ of \cite{Lin_Han_2025_JOE}. Condition (iii) can be verified for standard nonparametric regression methods, such as series and spline estimators, provided the conditional means are smooth as required by condition (ii). Condition (iv) is a locality requirement, similar to condition (ii) of Assumption \ref{assumption:matching}, but it requires the weighted moments of the covariate discrepancy to shrink fast enough for the bias-corrected estimator to attain $\sqrt{n}$-consistency.
 
\section{Additional Empirical Application Results}\label{appendix:additionalempirical}
In this section, we present the additional results for our empirical application. Figure \ref{figure:empiricalcaikernel} shows the event study plot using the kernel matching estimator with the same bandwidth choice as section \ref{section:empiricalapplication} without bias correction. Figure \ref{figure:empiricalcainndivbc} shows the event study estimator using the bias-corrected nearest neighbor matching estimator with diverging neighbors ($M=n_{gs}^{1/3}$). The change of the matching methods does not change most of our main results.
\begin{figure}[p]
\centering
\setstretch{1}
\setlength{\parindent}{0pt}
\caption{Matched DDD and Regression event studies: Kernel matching.}
\label{figure:empiricalcaikernel}
\vspace{0.04in}
\begin{tikzpicture}[every node/.style={font=\footnotesize}]
\begin{axis}[
    width=0.78\linewidth, height=0.24\textheight,
    xmin=-2.35, xmax=5.35, ymin=-0.40027801463779, ymax=0.34504899909701,
    xtick={-2,-1,0,1,2,3,4,5},
    xlabel={Event time (year $-$ 2003)}, ylabel={Treatment effect},
    axis lines=left, axis line style={black,thin},
    tick align=outside, tick style={black,thin},
    tick label style={font=\footnotesize},
    label style={font=\footnotesize},
    xlabel style={at={(axis description cs:0.5,-0.17)},anchor=north},
    ylabel style={at={(axis description cs:-0.11,0.5)},anchor=south},
    scaled y ticks=false,
    yticklabel style={/pgf/number format/fixed,/pgf/number format/precision=3},
    legend style={at={(1.02,0.5)},anchor=west,draw=none,
                  font=\footnotesize,inner sep=1pt},
    legend cell align=left, legend columns=1,
    clip=true,
]
\addplot[black,thin,no marks,forget plot] coordinates {(-2.35,0) (5.35,0)};
\addplot[black,dashed,thin,no marks,forget plot] coordinates {(-1,-0.40027801463779) (-1,0.34504899909701)};

\addplot[
    only marks,mark=*,mark size=2.1pt,color={rgb,255:red,0;green,76;blue,153},
    error bars/.cd,y dir=both,y explicit,
    error bar style={line width=0.65pt},error mark options={rotate=90,mark size=2.3pt,line width=0.65pt}
] coordinates {
(-2.065,0.03278396501074) += (0,0.033038353351182) -= (0,0.033038353351182)
(-0.065,0.07723661597553) += (0,0.051545595431263) -= (0,0.051545595431263)
(0.935,-0.043751965344329) += (0,0.040673623325812) -= (0,0.040673623325812)
(1.935,0.032237012146934) += (0,0.053347809186021) -= (0,0.053347809186021)
(2.935,0.17860639235518) += (0,0.099240990749347) -= (0,0.099240990749347)
(3.935,0.010624681265614) += (0,0.084083444079536) -= (0,0.084083444079536)
(4.935,-0.19722942837161) += (0,0.12107520493823) -= (0,0.12107520493823)
};
\addlegendentry{Regression}
\addplot[forget plot,only marks,mark=*,mark size=2.1pt,color={rgb,255:red,0;green,76;blue,153}] coordinates {(-1.065,0)};

\addplot[
    only marks,mark=triangle*,mark size=2.1pt,color={rgb,255:red,190;green,52;blue,44},
    error bars/.cd,y dir=both,y explicit,
    error bar style={line width=0.65pt},error mark options={rotate=90,mark size=2.3pt,line width=0.65pt}
] coordinates {
(-1.935,-0.0045184127945945) += (0,0.038579343534654) -= (0,0.038579343534654)
(0.065,0.010486062913504) += (0,0.052371191362512) -= (0,0.052371191362512)
(1.065,-0.0057368667783195) += (0,0.047502994625915) -= (0,0.047502994625915)
(2.065,0.0088612619268516) += (0,0.063082628118996) -= (0,0.063082628118996)
(3.065,-0.024498987729501) += (0,0.090248125941978) -= (0,0.090248125941978)
(4.065,-0.078637613461037) += (0,0.077979145942919) -= (0,0.077979145942919)
(5.065,-0.11730109477229) += (0,0.11690348152878) -= (0,0.11690348152878)
};
\addlegendentry{Matched DDD}
\addplot[forget plot,only marks,mark=triangle*,mark size=2.1pt,color={rgb,255:red,190;green,52;blue,44}] coordinates {(-0.935,0)};

\end{axis}
\end{tikzpicture}\par
\vspace{-0.02in}
{\footnotesize Panel A: Net saving}\par
\vspace{0.05in}

\begin{tikzpicture}[every node/.style={font=\footnotesize}]
\begin{axis}[
    width=0.78\linewidth, height=0.24\textheight,
    xmin=-2.35, xmax=5.35, ymin=-0.13464693760165, ymax=0.066716038222654,
    xtick={-2,-1,0,1,2,3,4,5},
    xlabel={Event time (year $-$ 2003)}, ylabel={Treatment effect},
    axis lines=left, axis line style={black,thin},
    tick align=outside, tick style={black,thin},
    tick label style={font=\footnotesize},
    label style={font=\footnotesize},
    xlabel style={at={(axis description cs:0.5,-0.17)},anchor=north},
    ylabel style={at={(axis description cs:-0.11,0.5)},anchor=south},
    scaled y ticks=false,
    yticklabel style={/pgf/number format/fixed,/pgf/number format/precision=3},
    legend style={at={(1.02,0.5)},anchor=west,draw=none,
                  font=\footnotesize,inner sep=1pt},
    legend cell align=left, legend columns=1,
    clip=true,
]
\addplot[black,thin,no marks,forget plot] coordinates {(-2.35,0) (5.35,0)};
\addplot[black,dashed,thin,no marks,forget plot] coordinates {(-1,-0.13464693760165) (-1,0.066716038222654)};

\addplot[
    only marks,mark=*,mark size=2.1pt,color={rgb,255:red,0;green,76;blue,153},
    error bars/.cd,y dir=both,y explicit,
    error bar style={line width=0.65pt},error mark options={rotate=90,mark size=2.3pt,line width=0.65pt}
] coordinates {
(-2.065,-0.0083879405501005) += (0,0.023449453918368) -= (0,0.023449453918368)
(-0.065,0.023063620532447) += (0,0.022853104228507) -= (0,0.022853104228507)
(0.935,0.014396292948598) += (0,0.025536221754226) -= (0,0.025536221754226)
(1.935,-0.032938609492962) += (0,0.030063560432264) -= (0,0.030063560432264)
(2.935,-0.024354776486304) += (0,0.028987026928152) -= (0,0.028987026928152)
(3.935,-0.04001878515996) += (0,0.03117962236484) -= (0,0.03117962236484)
(4.935,-0.065892770571912) += (0,0.030702135926055) -= (0,0.030702135926055)
};
\addlegendentry{Regression}
\addplot[forget plot,only marks,mark=*,mark size=2.1pt,color={rgb,255:red,0;green,76;blue,153}] coordinates {(-1.065,0)};

\addplot[
    only marks,mark=triangle*,mark size=2.1pt,color={rgb,255:red,190;green,52;blue,44},
    error bars/.cd,y dir=both,y explicit,
    error bar style={line width=0.65pt},error mark options={rotate=90,mark size=2.3pt,line width=0.65pt}
] coordinates {
(-1.935,-0.0081265253494569) += (0,0.023426805222027) -= (0,0.023426805222027)
(0.065,0.018524730169452) += (0,0.02387724512331) -= (0,0.02387724512331)
(1.065,0.012846374035813) += (0,0.02625886988281) -= (0,0.02625886988281)
(2.065,-0.041110218730918) += (0,0.031551723372956) -= (0,0.031551723372956)
(3.065,-0.030719869786783) += (0,0.031091420698654) -= (0,0.031091420698654)
(4.065,-0.048407812094209) += (0,0.031739619493544) -= (0,0.031739619493544)
(5.065,-0.071365500653015) += (0,0.03162119865168) -= (0,0.03162119865168)
};
\addlegendentry{Matched DDD}
\addplot[forget plot,only marks,mark=triangle*,mark size=2.1pt,color={rgb,255:red,190;green,52;blue,44}] coordinates {(-0.935,0)};

\end{axis}
\end{tikzpicture}\par
\vspace{-0.02in}
{\footnotesize Panel B: Saving rate}\par
\vspace{0.05in}

\begin{tikzpicture}[every node/.style={font=\footnotesize}]
\begin{axis}[
    width=0.78\linewidth, height=0.24\textheight,
    xmin=-2.35, xmax=5.35, ymin=-0.12680145612051, ymax=0.24111493392685,
    xtick={-2,-1,0,1,2,3,4,5},
    xlabel={Event time (year $-$ 2003)}, ylabel={Treatment effect},
    axis lines=left, axis line style={black,thin},
    tick align=outside, tick style={black,thin},
    tick label style={font=\footnotesize},
    label style={font=\footnotesize},
    xlabel style={at={(axis description cs:0.5,-0.17)},anchor=north},
    ylabel style={at={(axis description cs:-0.11,0.5)},anchor=south},
    scaled y ticks=false,
    yticklabel style={/pgf/number format/fixed,/pgf/number format/precision=3},
    legend style={at={(1.02,0.5)},anchor=west,draw=none,
                  font=\footnotesize,inner sep=1pt},
    legend cell align=left, legend columns=1,
    clip=true,
]
\addplot[black,thin,no marks,forget plot] coordinates {(-2.35,0) (5.35,0)};
\addplot[black,dashed,thin,no marks,forget plot] coordinates {(-1,-0.12680145612051) (-1,0.24111493392685)};

\addplot[
    only marks,mark=*,mark size=2.1pt,color={rgb,255:red,0;green,76;blue,153},
    error bars/.cd,y dir=both,y explicit,
    error bar style={line width=0.65pt},error mark options={rotate=90,mark size=2.3pt,line width=0.65pt}
] coordinates {
(-2.065,-0.031466247576672) += (0,0.040632916047508) -= (0,0.040632916047508)
(-0.065,0.006780677822692) += (0,0.04200750436697) -= (0,0.04200750436697)
(0.935,0.030510667282213) += (0,0.04027750711002) -= (0,0.04027750711002)
(1.935,0.049415648042741) += (0,0.048055316651295) -= (0,0.048055316651295)
(2.935,0.045013756686105) += (0,0.045749248829915) -= (0,0.045749248829915)
(3.935,0.070427588164144) += (0,0.052713895580232) -= (0,0.052713895580232)
(4.935,0.14103572272288) += (0,0.053576951374135) -= (0,0.053576951374135)
};
\addlegendentry{Regression}
\addplot[forget plot,only marks,mark=*,mark size=2.1pt,color={rgb,255:red,0;green,76;blue,153}] coordinates {(-1.065,0)};

\addplot[
    only marks,mark=triangle*,mark size=2.1pt,color={rgb,255:red,190;green,52;blue,44},
    error bars/.cd,y dir=both,y explicit,
    error bar style={line width=0.65pt},error mark options={rotate=90,mark size=2.3pt,line width=0.65pt}
] coordinates {
(-1.935,-0.027154497723546) += (0,0.047446848098552) -= (0,0.047446848098552)
(0.065,0.0084671514860413) += (0,0.051969209321479) -= (0,0.051969209321479)
(1.065,0.036294101775876) += (0,0.04348937529638) -= (0,0.04348937529638)
(2.065,0.06581290751656) += (0,0.0501066383043) -= (0,0.0501066383043)
(3.065,0.06689455987042) += (0,0.052231537503117) -= (0,0.052231537503117)
(4.065,0.080887034468495) += (0,0.055444317893088) -= (0,0.055444317893088)
(5.065,0.14138397016649) += (0,0.056169194055113) -= (0,0.056169194055113)
};
\addlegendentry{Matched DDD}
\addplot[forget plot,only marks,mark=triangle*,mark size=2.1pt,color={rgb,255:red,190;green,52;blue,44}] coordinates {(-0.935,0)};

\end{axis}
\end{tikzpicture}\par
\vspace{-0.02in}
{\footnotesize Panel C: Flexible-term saving ratio}\par
\vspace{0.05in}

\vspace{0.02in}
\begin{minipage}{\linewidth}
\scriptsize
\textit{Notes:} This figure plots the event study estimators from the regression and matching. The covariates included in the analysis are:  age and education level of household head, and household size. Cross sectional sample sizes for panel A, B, and C are 3,736 (29,888 observations), 3,332 households (26,656 observations), and 3,545 households (28,360 observations). The matching method used is the Epanechnikov kernel matching without bias correction.
\par
\vspace{0.25\baselineskip}
\noindent\begin{minipage}[t][6\baselineskip][t]{\linewidth}
\strut
\end{minipage}
\end{minipage}
\end{figure}
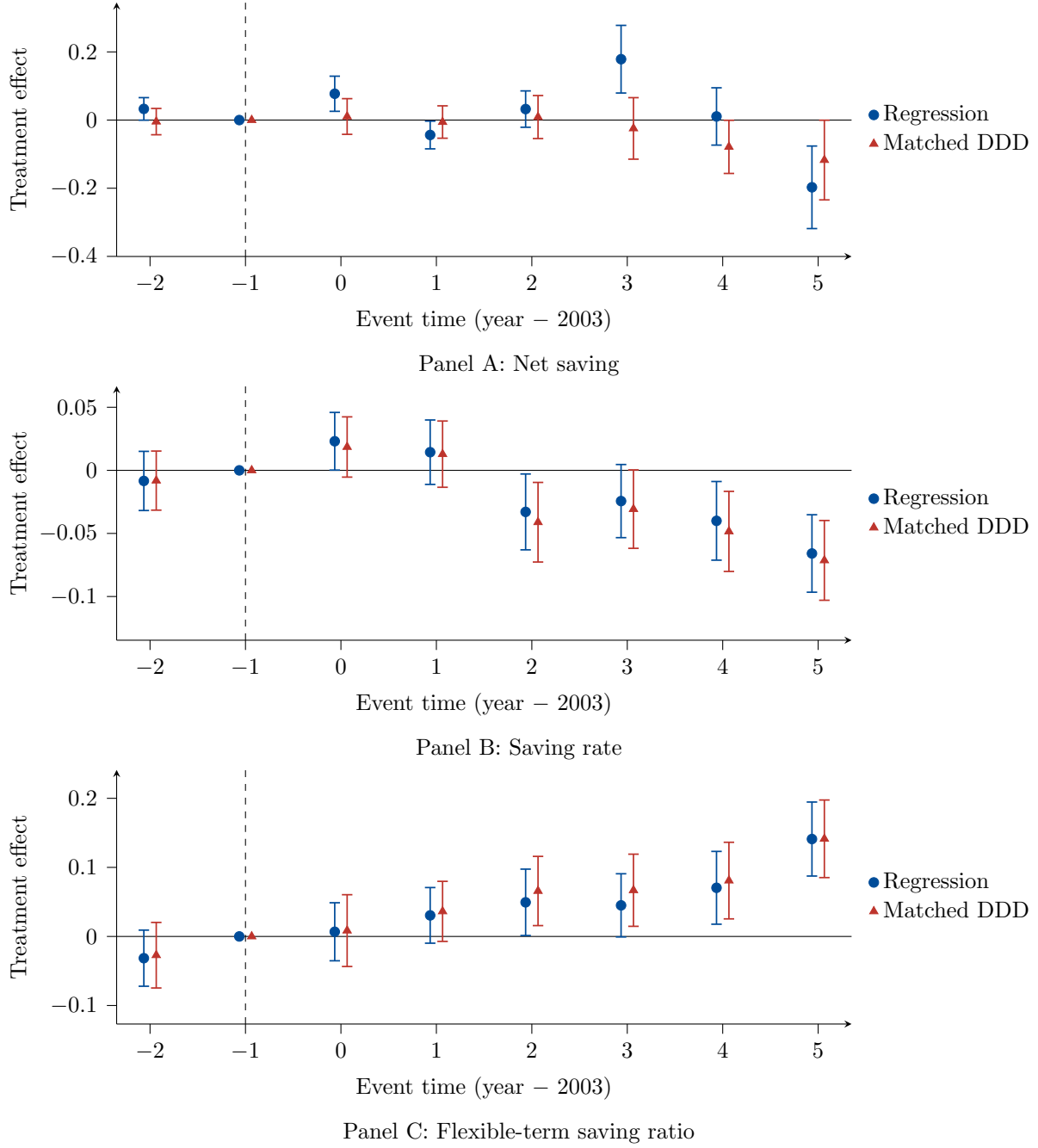

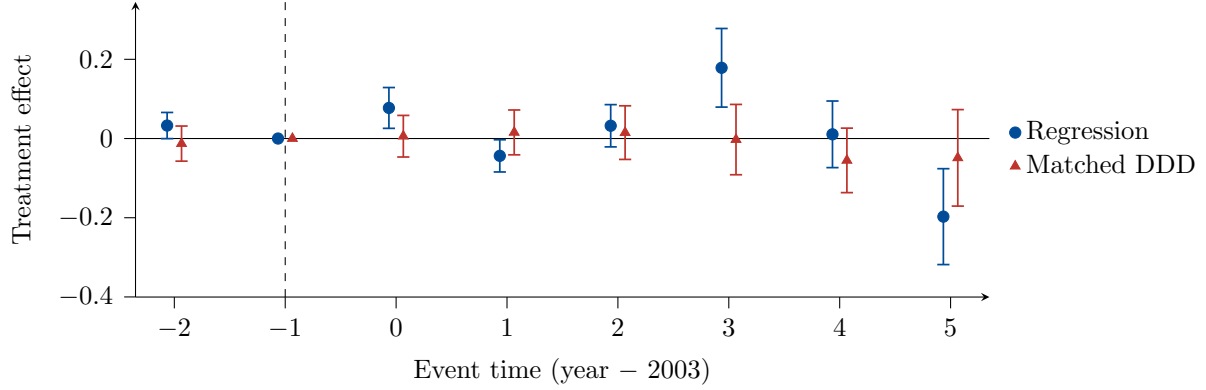
\begin{figure}[p]
\centering
\setstretch{1}
\setlength{\parindent}{0pt}
\caption{Matched DDD and Regression event studies: Bias-corrected Nearest-neighbor matching.}
\label{figure:empiricalcainndivbc}
\vspace{0.04in}
\begin{tikzpicture}[every node/.style={font=\footnotesize}]
\begin{axis}[
    width=0.78\linewidth, height=0.24\textheight,
    xmin=-2.35, xmax=5.35, ymin=-0.40027801463779, ymax=0.34504899909701,
    xtick={-2,-1,0,1,2,3,4,5},
    xlabel={Event time (year $-$ 2003)}, ylabel={Treatment effect},
    axis lines=left, axis line style={black,thin},
    tick align=outside, tick style={black,thin},
    tick label style={font=\footnotesize},
    label style={font=\footnotesize},
    xlabel style={at={(axis description cs:0.5,-0.17)},anchor=north},
    ylabel style={at={(axis description cs:-0.11,0.5)},anchor=south},
    scaled y ticks=false,
    yticklabel style={/pgf/number format/fixed,/pgf/number format/precision=3},
    legend style={at={(1.02,0.5)},anchor=west,draw=none,
                  font=\footnotesize,inner sep=1pt},
    legend cell align=left, legend columns=1,
    clip=true,
]
\addplot[black,thin,no marks,forget plot] coordinates {(-2.35,0) (5.35,0)};
\addplot[black,dashed,thin,no marks,forget plot] coordinates {(-1,-0.40027801463779) (-1,0.34504899909701)};

\addplot[
    only marks,mark=*,mark size=2.1pt,color={rgb,255:red,0;green,76;blue,153},
    error bars/.cd,y dir=both,y explicit,
    error bar style={line width=0.65pt},error mark options={rotate=90,mark size=2.3pt,line width=0.65pt}
] coordinates {
(-2.065,0.03278396501074) += (0,0.033038353351182) -= (0,0.033038353351182)
(-0.065,0.07723661597553) += (0,0.051545595431263) -= (0,0.051545595431263)
(0.935,-0.043751965344329) += (0,0.040673623325812) -= (0,0.040673623325812)
(1.935,0.032237012146934) += (0,0.053347809186021) -= (0,0.053347809186021)
(2.935,0.17860639235518) += (0,0.099240990749347) -= (0,0.099240990749347)
(3.935,0.010624681265614) += (0,0.084083444079536) -= (0,0.084083444079536)
(4.935,-0.19722942837161) += (0,0.12107520493823) -= (0,0.12107520493823)
};
\addlegendentry{Regression}
\addplot[forget plot,only marks,mark=*,mark size=2.1pt,color={rgb,255:red,0;green,76;blue,153}] coordinates {(-1.065,0)};

\addplot[
    only marks,mark=triangle*,mark size=2.1pt,color={rgb,255:red,190;green,52;blue,44},
    error bars/.cd,y dir=both,y explicit,
    error bar style={line width=0.65pt},error mark options={rotate=90,mark size=2.3pt,line width=0.65pt}
] coordinates {
(-1.935,-0.012801587014608) += (0,0.044354735953006) -= (0,0.044354735953006)
(0.065,0.0058417971755921) += (0,0.052506760836412) -= (0,0.052506760836412)
(1.065,0.015467949917381) += (0,0.056613844350808) -= (0,0.056613844350808)
(2.065,0.014919957390283) += (0,0.067698675313837) -= (0,0.067698675313837)
(3.065,-0.0026846506111998) += (0,0.088771480341196) -= (0,0.088771480341196)
(4.065,-0.05521046061814) += (0,0.081414310516548) -= (0,0.081414310516548)
(5.065,-0.048797220957064) += (0,0.12183620995834) -= (0,0.12183620995834)
};
\addlegendentry{Matched DDD}
\addplot[forget plot,only marks,mark=triangle*,mark size=2.1pt,color={rgb,255:red,190;green,52;blue,44}] coordinates {(-0.935,0)};

\end{axis}
\end{tikzpicture}\par
\vspace{-0.02in}
{\footnotesize Panel A: Net saving}\par
\vspace{0.05in}

\begin{tikzpicture}[every node/.style={font=\footnotesize}]
\begin{axis}[
    width=0.78\linewidth, height=0.24\textheight,
    xmin=-2.35, xmax=5.35, ymin=-0.13464693760165, ymax=0.066716038222654,
    xtick={-2,-1,0,1,2,3,4,5},
    xlabel={Event time (year $-$ 2003)}, ylabel={Treatment effect},
    axis lines=left, axis line style={black,thin},
    tick align=outside, tick style={black,thin},
    tick label style={font=\footnotesize},
    label style={font=\footnotesize},
    xlabel style={at={(axis description cs:0.5,-0.17)},anchor=north},
    ylabel style={at={(axis description cs:-0.11,0.5)},anchor=south},
    scaled y ticks=false,
    yticklabel style={/pgf/number format/fixed,/pgf/number format/precision=3},
    legend style={at={(1.02,0.5)},anchor=west,draw=none,
                  font=\footnotesize,inner sep=1pt},
    legend cell align=left, legend columns=1,
    clip=true,
]
\addplot[black,thin,no marks,forget plot] coordinates {(-2.35,0) (5.35,0)};
\addplot[black,dashed,thin,no marks,forget plot] coordinates {(-1,-0.13464693760165) (-1,0.066716038222654)};

\addplot[
    only marks,mark=*,mark size=2.1pt,color={rgb,255:red,0;green,76;blue,153},
    error bars/.cd,y dir=both,y explicit,
    error bar style={line width=0.65pt},error mark options={rotate=90,mark size=2.3pt,line width=0.65pt}
] coordinates {
(-2.065,-0.0083879405501005) += (0,0.023449453918368) -= (0,0.023449453918368)
(-0.065,0.023063620532447) += (0,0.022853104228507) -= (0,0.022853104228507)
(0.935,0.014396292948598) += (0,0.025536221754226) -= (0,0.025536221754226)
(1.935,-0.032938609492962) += (0,0.030063560432264) -= (0,0.030063560432264)
(2.935,-0.024354776486304) += (0,0.028987026928152) -= (0,0.028987026928152)
(3.935,-0.04001878515996) += (0,0.03117962236484) -= (0,0.03117962236484)
(4.935,-0.065892770571912) += (0,0.030702135926055) -= (0,0.030702135926055)
};
\addlegendentry{Regression}
\addplot[forget plot,only marks,mark=*,mark size=2.1pt,color={rgb,255:red,0;green,76;blue,153}] coordinates {(-1.065,0)};

\addplot[
    only marks,mark=triangle*,mark size=2.1pt,color={rgb,255:red,190;green,52;blue,44},
    error bars/.cd,y dir=both,y explicit,
    error bar style={line width=0.65pt},error mark options={rotate=90,mark size=2.3pt,line width=0.65pt}
] coordinates {
(-1.935,-0.004967074575513) += (0,0.023275989710374) -= (0,0.023275989710374)
(0.065,0.012726526604046) += (0,0.024825967964392) -= (0,0.024825967964392)
(1.065,0.012783708673778) += (0,0.02669912742715) -= (0,0.02669912742715)
(2.065,-0.035565952317581) += (0,0.033331138965196) -= (0,0.033331138965196)
(3.065,-0.026311100976547) += (0,0.031582659227105) -= (0,0.031582659227105)
(4.065,-0.038604147184364) += (0,0.032803121596613) -= (0,0.032803121596613)
(5.065,-0.069614514669872) += (0,0.034144072396003) -= (0,0.034144072396003)
};
\addlegendentry{Matched DDD}
\addplot[forget plot,only marks,mark=triangle*,mark size=2.1pt,color={rgb,255:red,190;green,52;blue,44}] coordinates {(-0.935,0)};

\end{axis}
\end{tikzpicture}\par
\vspace{-0.02in}
{\footnotesize Panel B: Saving rate}\par
\vspace{0.05in}

\begin{tikzpicture}[every node/.style={font=\footnotesize}]
\begin{axis}[
    width=0.78\linewidth, height=0.24\textheight,
    xmin=-2.35, xmax=5.35, ymin=-0.12680145612051, ymax=0.24111493392685,
    xtick={-2,-1,0,1,2,3,4,5},
    xlabel={Event time (year $-$ 2003)}, ylabel={Treatment effect},
    axis lines=left, axis line style={black,thin},
    tick align=outside, tick style={black,thin},
    tick label style={font=\footnotesize},
    label style={font=\footnotesize},
    xlabel style={at={(axis description cs:0.5,-0.17)},anchor=north},
    ylabel style={at={(axis description cs:-0.11,0.5)},anchor=south},
    scaled y ticks=false,
    yticklabel style={/pgf/number format/fixed,/pgf/number format/precision=3},
    legend style={at={(1.02,0.5)},anchor=west,draw=none,
                  font=\footnotesize,inner sep=1pt},
    legend cell align=left, legend columns=1,
    clip=true,
]
\addplot[black,thin,no marks,forget plot] coordinates {(-2.35,0) (5.35,0)};
\addplot[black,dashed,thin,no marks,forget plot] coordinates {(-1,-0.12680145612051) (-1,0.24111493392685)};

\addplot[
    only marks,mark=*,mark size=2.1pt,color={rgb,255:red,0;green,76;blue,153},
    error bars/.cd,y dir=both,y explicit,
    error bar style={line width=0.65pt},error mark options={rotate=90,mark size=2.3pt,line width=0.65pt}
] coordinates {
(-2.065,-0.031466247576672) += (0,0.040632916047508) -= (0,0.040632916047508)
(-0.065,0.006780677822692) += (0,0.04200750436697) -= (0,0.04200750436697)
(0.935,0.030510667282213) += (0,0.04027750711002) -= (0,0.04027750711002)
(1.935,0.049415648042741) += (0,0.048055316651295) -= (0,0.048055316651295)
(2.935,0.045013756686105) += (0,0.045749248829915) -= (0,0.045749248829915)
(3.935,0.070427588164144) += (0,0.052713895580232) -= (0,0.052713895580232)
(4.935,0.14103572272288) += (0,0.053576951374135) -= (0,0.053576951374135)
};
\addlegendentry{Regression}
\addplot[forget plot,only marks,mark=*,mark size=2.1pt,color={rgb,255:red,0;green,76;blue,153}] coordinates {(-1.065,0)};

\addplot[
    only marks,mark=triangle*,mark size=2.1pt,color={rgb,255:red,190;green,52;blue,44},
    error bars/.cd,y dir=both,y explicit,
    error bar style={line width=0.65pt},error mark options={rotate=90,mark size=2.3pt,line width=0.65pt}
] coordinates {
(-1.935,-0.035907612692554) += (0,0.056654786131993) -= (0,0.056654786131993)
(0.065,0.0027644551583313) += (0,0.063656042844038) -= (0,0.063656042844038)
(1.065,0.032666455185212) += (0,0.050281596125639) -= (0,0.050281596125639)
(2.065,0.074529172681313) += (0,0.051547267760758) -= (0,0.051547267760758)
(3.065,0.062265555909227) += (0,0.066354693944396) -= (0,0.066354693944396)
(4.065,0.07893610113998) += (0,0.061250524302551) -= (0,0.061250524302551)
(5.065,0.13162352714768) += (0,0.061365148667294) -= (0,0.061365148667294)
};
\addlegendentry{Matched DDD}
\addplot[forget plot,only marks,mark=triangle*,mark size=2.1pt,color={rgb,255:red,190;green,52;blue,44}] coordinates {(-0.935,0)};

\end{axis}
\end{tikzpicture}\par
\vspace{-0.02in}
{\footnotesize Panel C: Flexible-term saving ratio}\par
\vspace{0.05in}

\vspace{0.02in}
\begin{minipage}{\linewidth}
\scriptsize
\textit{Notes:} This figure plots the event study estimators from the regression and matching. The covariates included in the analysis are:  age and education level of household head, and household size. Cross sectional sample sizes for panel A, B, and C are 3,736 (29,888 observations), 3,332 households (26,656 observations), and 3,545 households (28,360 observations). The matching method used is the nearest neighbor matching estimators ($M=n_{gs}^{1/3}$) with quadratic bias correction using the series estimator.
\par
\vspace{0.25\baselineskip}
\noindent\begin{minipage}[t][6\baselineskip][t]{\linewidth}
\strut
\end{minipage}
\end{minipage}
\end{figure}

\section{Auxiliary Lemma}
\begin{lemma}\label{lemma:regressionweight}
 The induced weights by regression specification $(2)$ satisfy,  
 \begin{align*}
    &\E\left[\left.\frac{(1-p_{11})L(D|G,S,X)}{\E[(D-L(D|G,S,X))^2]} X \right| D=0\right] =\E\left[\left.\frac{p_{01}L(D|G,S,X)}{\E[(D-L(D|G,S,X))^2]} X \right| G=0,S=1\right]\\
    &+ \E\left[\left.\frac{p_{10}L(D|G,S,X)}{\E[(D-L(D|G,S,X))^2]} X \right| G=1,S=0\right] 
    +\E\left[\left.\frac{p_{00}L(D|G,S,X)}{\E[(D-L(D|G,S,X))^2]} X \right| G=0,S=0\right]\\
    &\E\left[\left.\frac{p_{11}\left(1-L(D|G,S,X)\right)}{\E[(D-L(D|G,S,X))^2]} X \right| G=1,S=1\right]=\E\left[\left.\frac{(1-p_{11})L(D|G,S,X)}{\E[(D-L(D|G,S,X))^2]} X \right| D=0\right] \\
 \end{align*}
\end{lemma}

\begin{lemma}\label{lemma:convergence}
Consider a subset of the data with two subgroups, the treated subgroup $(1,1)$ and one of the comparison group $(g,s)$. And we conduct the pairwise matching procedure described in section \ref{section:matchedtriplediffereces} satisfying assumption \ref{assumption:matching}. Let $(Z_i)_{i=1}^n$ be a sequence of i.i.d random variables with $\E[Z_i|G_i=g,S_i=s]=\mu_Z(g,s)$, $\E[Z_i|G_i=g,S_i=s,X_i=x]=\mu_Z(g,s,x)$ and $\V[Z_i|G_i=g,S_i=s,X_i=x]=\sigma^2_Z(g,s,x)$ with both functions being Lipschitz-continuous and with $\E[Z_i^4|G_i=g,S_i=s,X_i=x]$ uniformly bounded over $x$. Then

\[ T_n= \frac{1}{N_{11}}\sum_{i\in \mathcal{I}_{11}}Z_i - \frac{1}{n_{11}}\sum_{i\in \mathcal{I}_{gs}}w_i  Z_i\plim  \mu_{Z}(1,1)-\E[\mu_Z(g,s,X)|G=1,S=1]\]

\end{lemma}

\begin{lemma}\label{lemma:varianceconvergence}
Under assumption \ref{assumption:iid}, \ref{assumption:matching}, and \ref{assumption:regularity}
\begin{align*}
     \frac{1}{n}\sum_{i\in\mathcal{I}_{11}}\sigma^2_{11,l}(X_{i})- \E\left[\frac{1}{n}\sum_{i\in\mathcal{I}_{11}}\sigma^2_{11,l}(X_{i}) \right] &=o_p(1) \\
     \frac{1}{n_{11}}\sum_{i\in\mathcal{I}_{11}}\sigma^2_{11,l}(X_{i})- \E\left[\frac{1}{n_{11}}\sum_{i\in\mathcal{I}_{11}}\sigma^2_{11,l}(X_{i}) \right] &=o_p(1) \\
\end{align*}
and for $(g,s)\in\{(0,1),(1,0),(0,0) \}$
\begin{align*}
    \frac{1}{n}\sum_{i\in\mathcal{I}_{gs}} w_{i}^2\sigma^2_{gs,l}(X_{i})-\E\left[ \frac{1}{n}\sum_{i\in\mathcal{I}_{gs}} w_{i}^2\sigma^2_{gs,l}(X_{i})\right] & =o_p(1) \\
     \frac{1}{n_{11}}\sum_{i\in\mathcal{I}_{gs}} w_{i}^2\sigma^2_{gs,l}(X_{i})-\E\left[ \frac{1}{n_{11}}\sum_{i\in\mathcal{I}_{gs}} w_{i}^2\sigma^2_{gs,l}(X_{i})\right] & =o_p(1) \\
\end{align*}
\end{lemma}

\begin{lemma}\label{lemma:weightproperty}
    Under assumption \ref{assumption:matching}, for any $(g,s)\in\{(0,1),(1,0),(0,0) \} $ the pairwise weight $w_{ij}$, $w_{ij'}$ and the the matching weight $w_j$, $w_{j'}$ induced by the pairwise matching between the treated subgroup $(1,1)$, and the untreated subgroup $(g,s)$ satisfies:

    (i) \begin{align*}
        &\E\left[\frac{1}{n^2}\sum_{j\in\mathcal{I}_{gs}}w_{j}^2\right] = O(n^{-1}), \quad  \frac{1}{n^2}\sum_{j\in\mathcal{I}_{gs}}w_{j}^2 = O_p(n^{-1}) \\
        &\sum_{i\in \mathcal{I}_{11}}w_{ij}w_{ij'} \leq \max_{j\in \mathcal{I}_{gs}} w_j,\ \quad \max_{j\in \mathcal{I}_{gs}} w_j \leq \left(\sum_{j\in\mathcal{I}_{gs}}w_{j}^2 \right)=O_p(n^{1/2})
    \end{align*}
    
    (ii) \[ \E\left[\left(\sum_{j\in \mathcal{I}_{gs\setminus\{1\}}}\abs{w_j^2-\left(w_j^{(1)}\right)^2}\right)^2 \right]  =O (n^{1/2})  \]
    (iii) For $r< 4+\delta$, we have: 
    \[ \E[w_j^{r}|j\in \mathcal{I}_{gs}]=O(1), \quad  \frac{1}{n_{gs}}\sum_{i\in \mathcal{I}_{gs}}w_j^r=O(1) \]
    (iv) \begin{align*}
        \frac{1}{n_{11}}\sum_{i\in \mathcal{I}_{11}}\sum_{j\in \mathcal{I}_{gs}}w_{ij}\norm{X_i-X_j} =o_p(1), \quad  \frac{1}{n_{11}}\sum_{i\in \mathcal{I}_{11}}\left(\sum_{j\in \mathcal{I}_{gs}}w_{ij}\norm{X_i-X_j}\right)^2 =o_p(1)
    \end{align*}
\end{lemma}

\begin{lemma}\label{lemma:subgroupproportion}
    For any $(g,s)\in \{0,1\}^2$, under assumption \ref{assumption:iid}, we have:
    \begin{align*}
        \E\left[\left(\frac{n}{n_{gs}}-\frac{1}{p_{gs}} \right)^2 \right]=O(n^{-1})
    \end{align*}
\end{lemma}

 \begin{lemma}\label{lemma:conditionalchebyshev}
     Let $T_n$ be a sequence of random variables such that $\E[T_n|\G,\bs,\X]=0$ and $Var(T_n|\G,\bs,\X)\xrightarrow{p}0$, then $T_n\xrightarrow{p}0$.
 \end{lemma}

 \begin{lemma} \label{lemma:reallocationbound}
 For the nearest neighbor matching with fixed neighbor, nearest neighbor matching with diverging neighbor, and the kernel matching introduced in appendix A. The matching weight $w_j$ generated by those three procedures all satisfies:
 \begin{align*}
     \sum_{j\in \mathcal{I}_{gs}\setminus\{1\}}\abs{w_j-w_j^{(1)}}\leq w_1+w_1^{(1)}+\1(1\in \mathcal{I}_{11}) +\1(1'\in \mathcal{I}_{11})
 \end{align*} 
 This implies that the nearest neighbor matching with fixed neighbor, nearest neighbor matching with diverging neighbor, and the kernel matching all satisfies assumption \ref{assumption:matching}, which is:
 \begin{align*}
     \E\left[\left(\sum_{j\in \mathcal{I}_{gs}\setminus\{1\}}\abs{w_j-w_j^{(1)}}\right)^4 \right] =O(1)
 \end{align*}
 \end{lemma}

\begin{lemma}\label{lemma:maximumdistancelimit}
If we use nearest-neighbor matching to match each unit in subgroup $(g,s)$ to units within the same subgroup, then for each fixed $m$, the sample maximum of the norm of the matching discrepancy for the $m$th nearest neighbor is $o_p(1)$. Specifically, let $k_m(i)$ denote the index of the $m$th closest match within the same subgroup for unit $i$ in subgroup $(g,s)$. Then
\[
\max_{i\in \mathcal{I}_{gs}}\left\|X_i - X_{k_m(i)}\right\| = o_p(1).
\]
\end{lemma}

\begin{lemma}\label{lemma:residualbound}
    Define $\epsilon_{il}^{gs}=\Delta Y_{il}- \E[\Delta Y_{il}|G_i=g,S_i=s,X_i]=\Delta Y_{il} -\mu_l^{gs}(X_i)$, then for any $(g,s)\in \{0,1\}^2$, under assumption \ref{assumption:regularity}, there exists some constant $C>0$, such that
    \begin{align*}
        \E\left[\left.\left(\epsilon_{il}^{gs}\right)^4\right\lvert \G,\bs,\X,G_i=g,S_i=s \right]<C
    \end{align*}
\end{lemma}

\begin{lemma}\label{lemma:equivalentestimator}
    The estimator $\widehat{\tau}_l$ can be written as the following two equivalent ways:
    \begin{align*}
         \widehat \tau_l&= \frac{1}{n_{11}}\sum_{i\in\mathcal{I}_{11}}\Delta Y_{il} -\frac{1}{n_{11}}\sum_{i\in\mathcal{I}_{01}}w_i\Delta Y_{il}-\left(\frac{1}{n_{11}}\sum_{i\in\mathcal{I}_{10}}w_i\Delta Y_{il} -\frac{1}{n_{11}}\sum_{i\in\mathcal{I}_{00}}w_i\Delta Y_{il}\right) \\
         &=\frac{1}{n_{11}}\sum_{i\in \mathcal{I}_{11}}\left(\Delta Y_{il}-  \widehat {\Delta Y_{il}^{01}}-\widehat {\Delta Y_{il}^{10}} + \widehat {\Delta Y_{il}^{00}} \right),
    \end{align*}
    where, $\widehat {\Delta Y_{il}^{gs}} =\sum_{j\in \mathcal{I}_{gs}}w_{ij}\Delta Y_{jl} $
\end{lemma}

\section{Proof of Theorems}\label{appendix:theoremproofs}
\begin{proof} Proof of Theorem \ref{theorem:identification}\\

Notice that: $Y_t(0)-Y_{t'}(0) =\sum_{g=0}^{t+1-t'} Y_{t-g}(0)-Y_{t-g-1}(0)$, for any $1\leq t'<t\leq T$, so
\begin{align*}
    &\E[ Y_t(0)-Y_{t'}(0)|G=1,S=1,X]-\E[ Y_t(0)-Y_{t'}(0)|G=0,S=1,X]\\
        &-[\E[Y_t(0)-Y_{t'}(0)|G=1,S=0,X]-\E[ Y_t(0)-Y_{t'}(0)|G=0,S=0,X] \\
        &=\sum_{g=0}^{t+1-t'} \left(\E[ Y_{t-g}(0)-Y_{t-g-1}(0)|G=1,S=1,X]-\E[ Y_{t-g}(0)-Y_{t-g-1}(0)|G=0,S=1,X] \right. \\
       &\left. -\E[Y_{t-g}(0)-Y_{t-g-1}(0)|G=1,S=0,X]-\E[ Y_{t-g}(0)-Y_{t-g-1}(0)|G=0,S=0,X]   \right) \\
       & =0
\end{align*}
The last line follows from assumption \ref{assumption:parallelgap}.
\label{proof:identification}
    \begin{equation*}
    \begin{split}
         \tau_l&=\E[Y_{t^*+l}(1)-Y_{t^*+l}(0)|G=1,S=1]\\
        &=\E[ \E[Y_{t^*+l}(1)-Y_{t^*+l}(0)|G_i=1,S=1,X]|G=1,S=1]\\
        &=\E[ \E[Y_{t^*+l}(1)-Y_{t^*+l}(0)|G=1,S=1,X]|G=1,S=1]\\
        &+\E[\E[ Y_{t^*+l}(0)-Y_{t^*-1}(0)|G=1,S=1,X]-\E[ Y_{t^*+l}(0)-Y_{t^*-1}(0)|G=0,S=1,X]|G=1,S=1]\\
        &-\E[\E[Y_{t^*+l}(0)-Y_{t^*-1}(0)|G=1,S=0,X]-\E[Y_{t^*+l}(0)-Y_{t^*-1}(0)|G=0,S=0,X] |G=1,S=1]\\
        &=\E[\E[Y_{t^*+l}(1)-Y_{t^*-1}(0)|G=1,S=1,X]-\E[ Y_{t^*+l}(0)-Y_{t^*-1}(0)|G=0,S=1,X]|G=1,S=1]\\
        &-\E[\E[Y_{t^*+l}(0)-Y_{t^*-1}(0)|G=1,S=0,X]-\E[Y_{t^*+l}(0)-Y_{t^*-1}(0)|G=0,S=0,X] |G=1,S=1]\\
        &=\E[\E[Y_{t^*+l}-Y_{t^*-1}|G=1,S=1,X]-\E[Y_{t^*+l}-Y_{t^*-1}|G=0,S=1,X]|G=1,S=1]\\
        &-\E[\E[Y_{t^*+l}-Y_{t^*-1}|G=1,S=0,X]-\E[Y_{t^*+l}-Y_{t^*-1}|G=0,S=0,X] |G=1,S=1]\\
        &=   \E[Y_{t^*+l}-Y_{t^*-1}|G=1,S=1]-\E[\E[Y_{t^*+l}-Y_{t^*-1}|G=0,S=1,X]|G=1,S=1]\\
            &-\E[Y_{t^*+l}-Y_{t^*-1}|G=1,S=0,X]|G=1,S=1] +\E[\E[Y_{t^*+l}-Y_{t^*-1}|G=0,S=0,X_i]|G=1,S=1] 
    \end{split} 
    \end{equation*}
    The first equality follows from law of iterated expectation. The second equality follows from the equation above and law of iterated expectation. The third equality follows by canceling out $\E[Y_{t}(0)|G=1,S=1]$. The fourth inequality follows from the assumption \ref{assumption:overlap}. The last equality follows from applying law of iterated expectation to the first term.
    
   On the other hand, for each $(g,s)\in \{(0,1),(1,0),(0,0) \}$, we have:
    \begin{align*}
     &\E[\E[Y_{t^*+l}-Y_{t^*-1}|G=g,S=1,s]|G=1,S=1]\\
     &=   \E\!\left[
    \E[\Delta Y_l\mid G=g,S=s,X]
    \mid G=1,S=1
\right]\\
&=\int_{\mathcal{X}} \E[\Delta Y_l\mid G=g,S=s,X=x]f_{11}(x)dx \\
&=\int_{\mathcal{X}} \E[\Delta Y_l\mid G=g,S=s,X=x]\frac{f_{11}(x)}{f_{gs}(x)}f_{gs}(x)dx \\
&=\E\left[\left. \E[\Delta Y_l\mid G=g,S=s,X] \frac{f_{11}(X)}{f_{gs}(X)} \right\lvert G=g,S=s  \right] \\
&=\E\left[\left. \E\left[\left. \frac{f_{11}(X)}{f_{gs}(X)} \Delta Y_l  \right\lvert G=g,S=s,X\right]  \right\lvert G=g,S=s  \right]\\
&=\E\left[\left. \frac{f_{11}(X)}{f_{gs}(X)} \Delta Y_l    \right\lvert G=g,S=s  \right]
    \end{align*}
This proves theorem \ref{theorem:identification}.
\end{proof}

\begin{proof}
Proof of theorem \ref{theorem:asymptotic}.
\begin{align*}
    \widehat \tau_l&= \frac{1}{n_{11}}\sum_{i\in\mathcal{I}_{11}}\Delta Y_{il} -\frac{1}{n_{11}}\sum_{i\in\mathcal{I}_{01}}w_i\Delta Y_{il}-\left(\frac{1}{n_{11}}\sum_{i\in\mathcal{I}_{10}}w_i\Delta Y_{il} -\frac{1}{n_{11}}\sum_{i\in\mathcal{I}_{00}}w_i\Delta Y_{il}\right) \\
    &=\frac{1}{n_{11}}\sum_{i\in\mathcal{I}_{11}}\Delta Y_{il} -\frac{1}{n_{11}}\sum_{i\in\mathcal{I}_{01}}w_i\Delta Y_{il} +\left( \frac{1}{n_{11}}\sum_{i\in\mathcal{I}_{11}}\Delta Y_{il} -\frac{1}{n_{11}}\sum_{i\in\mathcal{I}_{10}}w_i\Delta Y_{il}\right) \\
    &-\left(\frac{1}{n_{11}}\sum_{i\in\mathcal{I}_{11}}\Delta Y_{il} -\frac{1}{n_{11}}\sum_{i\in\mathcal{I}_{00}}w_i\Delta Y_{il}\right) \\
    &\xrightarrow{p} \E[\Delta Y_l |G=1,S=1]-\E[\E[\Delta Y_l|G=0,S=1,X]|G=1,S=1] \\
    & +\E[\Delta Y_l |G=1,S=1]-\E[\E[\Delta Y_l|G=1,S=0,X]|G=1,S=1] \\
    &  -\left(\E[\Delta Y_l |G=1,S=1]-\E[\E[\Delta Y_l|G=0,S=0,X]|G=1,S=1] \right) \\
    &= \E[\Delta Y_l |G=1,S=1]-\E[\E[\Delta Y_l|G=0,S=1,X]|G=1,S=1] \\
    &- \E[\E[\Delta Y_l|G=1,S=0,X]|G=1,S=1] +\E[\E[\Delta Y_l|G=0,S=0,X]|G=1,S=1] \\
    &=\tau_l
\end{align*}
    The convergence results follows from lemma \ref{lemma:convergence}. The last equality follows from the identification results.

   For notation compactness define $\mu^{gs}_l(x)=\E[\Delta Y_l|G=g,S=s,X=x]$, $\epsilon_{l}^{gs}(X_i)=\Delta Y_{il} -\mu_l^{G_iS_i}(X_i)$, $\sigma_{gs,l}^2(X_i)=Var[\Delta Y_l|G=g,S=s,X=x]$. To establish the asymptotic normality, notice we can write $\widehat \tau_l$ as    \begin{align*}
\widehat{\tau}_l-B
&=\frac{1}{n_{11}}\isum\left[
\1(i\in \mathcal{I}_{11})-\1(i\in \mathcal{I}_{01})w_i-\1(i\in \mathcal{I}_{10})w_i+\1(i\in \mathcal{I}_{00})w_i
\right]\Delta Y_{il} \\
&=\frac{1}{n_{11}}\isum \1(i\in \mathcal{I}_{11})(\mu_l^{11}(X_i)-\mu_l^{01}(X_i)-\mu_l^{10}(X_i)+\mu_{l}^{00}(X_i))\\
&+\frac{1}{n_{11}}\isum \left[
\1(i\in \mathcal{I}_{11})-\1(i\in \mathcal{I}_{01})w_i-\1(i\in \mathcal{I}_{10})w_i+\1(i\in \mathcal{I}_{00})w_i
\right](\Delta Y_{il} -\mu_l^{G_iS_i}(X_i))
\end{align*}

where, 
\begin{align*}
    B&=\frac{1}{n_{11}}\isum\left[\1(i\in \mathcal{I}_{11})-\1(i\in \mathcal{I}_{01})w_i\right]\mu_l^{01}(X_i)+\frac{1}{n_{11}}\isum\left[\1(i\in \mathcal{I}_{11})-\1(i\in \mathcal{I}_{10})w_i\right]\mu_l^{10}(X_i)\\ 
    &-\frac{1}{n_{11}}\isum\left[\1(i\in \mathcal{I}_{11})-\1(i\in \mathcal{I}_{00})w_i\right]\mu_l^{00}(X_i) 
\end{align*}
Define $V= V_{\tau(X)} +\E\left[V_{E} \right]$, where,
\begin{align*}
    V_{\tau(X)}& =\E\left[\left. \left(\mu_l^{11}(X_i)-\mu_l^{01}(X_i)-\mu_l^{10}(X_i)+\mu_l^{00}(X_i)-\tau_l \right)^2\right\vert G_i=1,S_i=1\right] \\
    V_{E}&= \frac{1}{n_{11}}\left(\sum_{i-n\in\mathcal{I}_{11}}\sigma^2_{11,l}(X_{i-n})+\sum_{i-n\in\mathcal{I}_{01}} w_{i-n}^2\sigma^2_{01,l}(X_{i-n})+\sum_{i-n\in\mathcal{I}_{10}} w_{i-n}^2\sigma^2_{10,l}(X_{i-n}) \right. \\
    & \left. +\sum_{i-n\in\mathcal{I}_{00}} w_{i-n}^2\sigma^2_{00,l}(X_{i-n})\right)
\end{align*}
Then, 
\begin{align*}
    &\frac{\widehat{\tau}_l-\tau_l-B  }{\sqrt{V/n_{11}}}\\
    &=\frac{1}{\sqrt{n_{11}V}}\isum \1(i\in \mathcal{I}_{11})(\mu_l^{11}(X_i)-\mu_l^{01}(X_i)-\mu_l^{10}(X_i)+\mu_{l}^{00}(X_i)-\tau_l) \\
    & +\frac{1}{\sqrt{n_{11}V}}\isum \left[
\1(i\in \mathcal{I}_{11})-\1(i\in \mathcal{I}_{01})w_i-\1(i\in \mathcal{I}_{10})w_i+\1(i\in \mathcal{I}_{00})w_i
\right]\epsilon_{l}^{G_iS_i}(X_i)
\end{align*}
Define
\begin{align*}
Z_{n,i}
=
\begin{cases}
\displaystyle
\frac{1}{\sqrt{n_{11}V}}
\1\{i\in \mathcal{I}_{11}\}
\left(
\mu_l^{11}(X_i)-\mu_l^{01}(X_i)
-\mu_l^{10}(X_i)+\mu_l^{00}(X_i)-\tau_l
\right),
& \text{if } 1\le i\le n,\\[0.8em]
\begin{aligned}
\frac{1}{\sqrt{n_{11}V}}
\Bigl[&
\1\{i-n\in \mathcal{I}_{11}\}
-\1\{i-n\in \mathcal{I}_{01}\}w_{i-n}\\
&-\1\{i-n\in \mathcal{I}_{10}\}w_{i-n}
+\1\{i-n\in \mathcal{I}_{00}\}w_{i-n}
\Bigr]
\epsilon_l^{G_{i-n}S_{i-n}}(X_{i-n}),
\end{aligned}
& \text{if } n+1\le i\le 2n.
\end{cases}
\end{align*}

so that
\[\frac{\widehat{\tau}_l-\tau_l-B  }{\sqrt{V/n_{11}}}=\sum_{i=1}^{2n}Z_{n,i}.\]
Define the $\sigma$-fields:
\begin{align*}
\mathcal{F}_{n,i}&=\begin{cases}
\sigma(\G,\bs,X_1,\ldots,X_i) &\text{if } 1\le i\le n\\
\sigma(\G,\bs,\X,\Delta Y_{1l},\ldots,\Delta{Y}_{i-n,l}) &\text{if } n+1\le i\le 2n.
\end{cases}
\end{align*}
Then, the sequence
\[\left\{\left(\sum_{j=1}^i Z_{n,j},\mathcal{F}_{n,i}\right),1\le i\le 2n\right\}\]
is a martingale. Now, for $1\le i\le n$,
\begin{align*}
&\E[Z_{n,i}^2|\mathcal{F}_{n,i-1}]\\
&=\E\left[\left.\left(\frac{1}{\sqrt{n_{11}V}}
\1(i\in \mathcal{I}_{11})
\left(
\mu_l^{11}(X_i)-\mu_l^{01}(X_i)
-\mu_l^{10}(X_i)+\mu_l^{00}(X_i)-\tau_l
\right)\right)^2\right\vert\G,\bs,X_1,\ldots,X_{i-1}\right]\\
&=\frac{\1(i\in \mathcal{I}_{11})}{n_{11}V}\E\left[\left.\left(\mu_l^{11}(X_i)-\mu_l^{01}(X_i)
-\mu_l^{10}(X_i)+\mu_l^{00}(X_i)-\tau_l \right)^2\right\vert G_i=1,S_i=1\right]
\end{align*}
and for $n+1\le i\le 2n$,
\begin{align*}
&\E[Z_{n,i}^2|\mathcal{F}_{n,i-1}]=\E\left[\left(\frac{1}{\sqrt{n_{11}V}}
\left[\1\{i-n\in \mathcal{I}_{11}\}-\1\{i-n\in \mathcal{I}_{01}\}w_{i-n} \right.\right.\right.\\ 
&\left.\left.\left.\left.  -\1\{i-n\in \mathcal{I}_{10}\}w_{i-n}+\1\{i-n\in \mathcal{I}_{00}\}w_{i-n}\right]\epsilon_l^{G_{i-n}S_{i-n}}(X_{i-n})\right)^2\right\vert\mathcal{F}_{n,i-1}\right]\\
&=\frac{\1(i-n\in \mathcal{I}_{11})}{n_{11}V}\sigma^2(1,1,X_{i-n})+\frac{\1(i-n\in \mathcal{I}_{01})}{n_{11}V}\cdot w_{i-n}^2\sigma^2(0,1,X_{i-n})\\
&+\frac{\1(i-n\in \mathcal{I}_{10})}{n_{11}V}\cdot w_{i-n}^2\sigma^2(1,0,X_{i-n})+\frac{\1(i-n\in \mathcal{I}_{00})}{n_{11}V}\cdot w_{i-n}^2\sigma^2(0,0,X_{i-n})
\end{align*}
It follows that:
\begin{align*}
\sum_{i=1}^{2n}\E[Z_{n,i}^2|\mathcal{F}_{n,i}]&=\frac{1}{V}\left(\E\left[\left.\left(\mu_l^{11}(X_i)-\mu_l^{01}(X_i)
-\mu_l^{10}(X_i)+\mu_l^{00}(X_i)-\tau_l \right)^2\right\vert G_i=1,S_i=1\right]\right.\\
&\left.+\frac{1}{n_{11}}\left(\sum_{i-n\in\mathcal{I}_{11}}\sigma^2(1,1,X_{i-n})+\sum_{i-n\in\mathcal{I}_{01}} w_{i-n}^2\sigma^2(0,1,X_{i-n})+\sum_{i-n\in\mathcal{I}_{10}} w_{i-n}^2\sigma^2(1,0,X_{i-n}) \right. \right.\\
&\left.\left.+\sum_{i-n\in\mathcal{I}_{00}} w_{i-n}^2\sigma^2(0,0,X_{i-n})\right)\right)\\
&\plim 1
\end{align*}
The last line follows from lemma \ref{lemma:varianceconvergence}.

On the other hand, under finite moments, for $1\le i\le n$,
\begin{align*}
\E[Z_{n,i}^4|\mathcal{F}_{n,i}]&=\E\left[\left.\left(\frac{1}{\sqrt{n_{11}V}}
\1(i\in \mathcal{I}_{11})
\left(
\mu_l^{11}(X_i)-\mu_l^{01}(X_i)
-\mu_l^{10}(X_i)+\mu_l^{00}(X_i)-\tau_l
\right)\right)^4\right\vert\G,\bs,X_1,\ldots,X_{i-1}\right]\\
&=\frac{\1(i\in\mathcal{I}_{11})}{n_{11}^2V^2}\E\left[\left.\left(\mu_l^{11}(X_i)-\mu_l^{01}(X_i)
-\mu_l^{10}(X_i)+\mu_l^{00}(X_i)-\tau_l\right)^4\right\vert G_i=1,S_i=1\right]\\
&\le \frac{125\1(i\in\mathcal{I}_{11})}{n_{11}^2V^2}\E\left[\left(\mu_l^{11}(X_i) \right)^4+\left(\mu_l^{01}(X_i) \right)^4  + \left(\mu_l^{10}(X_i) \right)^4+\left(\mu_l^{00}(X_i) \right)^4 +\tau_l^4  \mid G_i=1,S_i=1 \right] \\
&\leq \frac{125\1(i\in\mathcal{I}_{11})}{n_{11}^2\underline{\sigma}^2}C
\end{align*}
The last inequality follows from assumption \ref{assumption:regularity}.\\
For $n+1\le i\le 2n$,
\begin{align*}
&\E[Z_{n,i}^4|\mathcal{F}_{n,i-1}]=\E\left[\left(\frac{1}{\sqrt{n_{11}V}}
\left[\1\{i-n\in \mathcal{I}_{11}\}-\1\{i-n\in \mathcal{I}_{01}\}w_{i-n} \right.\right.\right.\\ 
&\left.\left.\left.\left.  -\1\{i-n\in \mathcal{I}_{10}\}w_{i-n}+\1\{i-n\in \mathcal{I}_{00}\}w_{i-n}\right]\epsilon_l^{G_{i-n}S_{i-n}}(X_{i-n})\right)^4\right\vert\mathcal{F}_{n,i-1}\right]\\
&=\frac{\1(i-n\in \mathcal{I}_{11})}{n_{11}^2V^2}\E\left[\left.\left(\epsilon_l^{11}(X_{i-n})\right)^4\right\vert G_{i-n}=1,S_{i-n}=1,X_{i-n}\right]\\
&+\frac{\1(i-n\in \mathcal{I}_{01})}{n_{11}^2V^2}w_{i-n}^4\E\left[\left.\left(\epsilon_l^{01}(X_{i-n})\right)^4\right\vert G_{i-n}=0,S_{i-n}=1,X_{i-n}\right]\\
&+\frac{\1(i-n\in \mathcal{I}_{10})}{n_{11}^2V^2}w_{i-n}^4\E\left[\left.\left(\epsilon_l^{10}(X_{i-n})\right)^4\right\vert G_{i-n}=1,S_{i-n}=0,X_{i-n}\right]\\
&+\frac{\1(i-n\in \mathcal{I}_{00})}{n_{11}^2V^2}w_{i-n}^4\E\left[\left.\left(\epsilon_l^{00}(X_{i-n})\right)^4\right\vert G_{i-n}=0,S_{i-n}=0,X_{i-n}\right]\\
&\le \frac{C}{\underline{\sigma}^2}\left(\frac{\1(i-n\in \mathcal{I}_{11})}{n_{11}^2}+\frac{\1(i-n\in \mathcal{I}_{01})}{n^2}w_{i-n}^4+\frac{\1(i-n\in \mathcal{I}_{10})}{n_{11}^2}w_{i-n}^4 + \frac{\1(i-n\in \mathcal{I}_{00})}{n_{11}^2}w_{i-n}^4 \right)
\end{align*}
for some finite constant $A_4$. Thus,
\begin{align*}
\sum_{i=1}^{2n}\E[Z_{n,i}^4|\mathcal{F}_{n,i-1}]&\le \frac{125C\1(i\in\mathcal{I}_{11})}{n_{11}\underline{\sigma}^2}\\
&+\frac{C}{n_{11}\underline{\sigma}^2}\left(\frac{1}{n_{11}}\sum_{i=n+1}^{2n}\1(i\in\mathcal{I}_{11}) \right) + \frac{C}{n_{11}\underline{\sigma}^2}\left(\frac{1}{n_{11}}\sum_{i=n+1}^{2n}\1(i\in\mathcal{I}_{01})w_{i-n}^4 \right) \\
&\frac{C}{n_{11}\underline{\sigma}^2}\left(\frac{1}{n_{11}}\sum_{i=n+1}^{2n}\1(i\in\mathcal{I}_{10})w_{i-n}^4 \right) + \frac{C}{n_{11}\underline{\sigma}^2}\left(\frac{1}{n_{11}}\sum_{i=n+1}^{2n}\1(i\in\mathcal{I}_{00})w_{i-n}^4 \right) \\
&\plim 0
\end{align*}
by Lemma \ref{lemma:weightproperty}(iii). Thus, by the martingale central limit theorem 
\begin{align*}
   \sum_{i=1}^{2n} Z_{n,i}\xrightarrow{d}\mathcal{N}(0,1)
\end{align*}
Therefore, by lemma \ref{lemma:varianceconvergence} and the Slutsky theorem,
\begin{align*}
   (V_{\tau(X)}+V_{E})^{-1/2} \sqrt{n_{11}}\left(\widehat{\tau}_l-\tau_l-B \right) = \frac{\sqrt{V}}{\sqrt{V_{\tau(X)}+V_{E}}}\frac{\widehat{\tau}_l-\tau_l-B  }{\sqrt{V/n_{11}}}\xrightarrow{d}\mathcal{N}(0,1)
\end{align*}
which completes the proof.
\end{proof}

\begin{proof}
    Proof of theorem \ref{theorem:varianceestimation}
    \begin{align*}
        \widehat V_{\tau(X)}&= \frac{1}{n_{11}}\sum_{i\in \mathcal{I}_{11}}\left(\Delta Y_{il}-  \widehat {\Delta Y_{il}^{01}}-\widehat {\Delta Y_{il}^{10}} + \widehat {\Delta Y_{il}^{00}}-\widehat{\tau}_l \right)^2  \\
        &=\frac{1}{n_{11}}\sum_{i\in \mathcal{I}_{11}}\left(\Delta Y_{il}-  \widehat {\Delta Y_{il}^{01}}-\widehat {\Delta Y_{il}^{10}} + \widehat {\Delta Y_{il}^{00}}-\tau_l -(\widehat{\tau}_l - \tau_l) \right)^2 \\
        &=\left(\frac{1}{n_{11}}\sum_{i\in \mathcal{I}_{11}}\left(\Delta Y_{il}-  \widehat {\Delta Y_{il}^{01}}-\widehat {\Delta Y_{il}^{10}} + \widehat {\Delta Y_{il}^{00}}-\tau_l \right)^2 \right) + (\tau_l-\widehat{\tau}_l)^2\\
        &-\frac{2(\widehat{\tau}_l-\tau_l)}{n_{11}}\sum_{i\in \mathcal{I}_{11}}\left(\Delta Y_{il}-  \widehat {\Delta Y_{il}^{01}}-\widehat {\Delta Y_{il}^{10}} + \widehat {\Delta Y_{il}^{00}}-\tau_l \right) \\
        &=\left(\frac{1}{n_{11}}\sum_{i\in \mathcal{I}_{11}}\left(\Delta Y_{il}-  \widehat {\Delta Y_{il}^{01}}-\widehat {\Delta Y_{il}^{10}} + \widehat {\Delta Y_{il}^{00}}-\tau_l \right)^2 \right) - (\tau_l-\widehat{\tau}_l)^2\\
        &=\frac{1}{n_{11}}\sum_{i\in \mathcal{I}_{11}}\left(\Delta Y_{il}-  \widehat {\Delta Y_{il}^{01}}-\widehat {\Delta Y_{il}^{10}} + \widehat {\Delta Y_{il}^{00}}-\tau_l \right)^2 +o_p(1)
    \end{align*}
    The third equality follows from lemma \ref{lemma:equivalentestimator}. The last equality follows from the consistency of $\widehat \tau_l$.
    \begin{align*}
        &\frac{1}{n_{11}}\sum_{i\in \mathcal{I}_{11}}\left(\Delta Y_{il}-  \widehat {\Delta Y_{il}^{01}}-\widehat {\Delta Y_{il}^{10}} + \widehat {\Delta Y_{il}^{00}}-\tau_l \right)^2 \\
        &=\frac{1}{n_{11}}\sum_{i\in \mathcal{I}_{11}}\left(\Delta Y_{il} -\sum_{j\in \mathcal{I}_{01}}w_{ij}\Delta Y_{jl} - \sum_{j\in \mathcal{I}_{10}}w_{ij}\Delta Y_{jl} +  \sum_{j\in \mathcal{I}_{00}}w_{ij}\Delta Y_{jl}-\tau_l\right)^2 \\
        &=\frac{1}{n_{11}}\sum_{i\in \mathcal{I}_{11}}\left(\mu_l^{11}(X_i) +\epsilon_{il}^{11} -\sum_{j\in \mathcal{I}_{01}}w_{ij}\left(\mu_{l}^{01}(X_j) +\epsilon_{jl}^{01} \right) - \sum_{j\in \mathcal{I}_{10}}w_{ij}\left(\mu_{l}^{10}(X_j) +\epsilon_{jl}^{10} \right)\right. \\
        &\left.+  \sum_{j\in \mathcal{I}_{00}}w_{ij}\left(\mu_{l}^{00}(X_j) +\epsilon_{jl}^{00} \right)-\tau_l\right)^2 \allowdisplaybreaks \\
        &=\frac{1}{n_{11}}\sum_{i\in \mathcal{I}_{11}} \left(  \mu_l^{11}(X_i)-\sum_{j\in \mathcal{I}_{01}}w_{ij}\mu_l^{01}(X_j) - \sum_{j\in \mathcal{I}_{10}}w_{ij}\mu_l^{10}(X_j) + \sum_{j\in \mathcal{I}_{00}}w_{ij}\mu_l^{00}(X_j) -\tau  \right.    \\
        &\left. \epsilon_{il}^{11} -\sum_{j\in \mathcal{I}_{01}}w_{ij}\epsilon_{jl}^{01} -\sum_{j\in \mathcal{I}_{10}}w_{ij}\epsilon_{jl}^{10} + \sum_{j\in \mathcal{I}_{00}}w_{ij}\epsilon_{jl}^{00}  \right)^2 \\
        &=\frac{1}{n_{11}}\sum_{i\in \mathcal{I}_{11}}\left(  \mu_l^{11}(X_i)-\sum_{j\in \mathcal{I}_{01}}w_{ij}\mu_l^{01}(X_j) - \sum_{j\in \mathcal{I}_{10}}w_{ij}\mu_l^{10}(X_j) + \sum_{j\in \mathcal{I}_{00}}w_{ij}\mu_l^{00}(X_j) -\tau_l  \right)^2 \\
        &+\frac{1}{n_{11}}\sum_{i\in \mathcal{I}_{11}}\left(\epsilon_{il}^{11} -\sum_{j\in \mathcal{I}_{01}}w_{ij}\epsilon_{jl}^{01} -\sum_{j\in \mathcal{I}_{10}}w_{ij}\epsilon_{jl}^{10} + \sum_{j\in \mathcal{I}_{00}}w_{ij}\epsilon_{jl}^{00} \right)^2 \\
        &+\frac{2}{n_{11}}\sum_{i\in \mathcal{I}_{11}}\left(  \mu_l^{11}(X_i)-\sum_{j\in \mathcal{I}_{01}}w_{ij}\mu_l^{01}(X_j) - \sum_{j\in \mathcal{I}_{10}}w_{ij}\mu_l^{10}(X_j) + \sum_{j\in \mathcal{I}_{00}}w_{ij}\mu_l^{00}(X_j) -\tau_l  \right)\\
        &\left(\epsilon_{il}^{11} -\sum_{j\in \mathcal{I}_{01}}w_{ij}\epsilon_{jl}^{01} -\sum_{j\in \mathcal{I}_{10}}w_{ij}\epsilon_{jl}^{10} + \sum_{j\in \mathcal{I}_{00}}w_{ij}\epsilon_{jl}^{00} \right)
    \end{align*}
    \textbf{First term:}
    \begin{align*}
        &\frac{1}{n_{11}}\sum_{i\in \mathcal{I}_{11}}\left(  \mu_l^{11}(X_i)-\sum_{j\in \mathcal{I}_{01}}w_{ij}\mu_l^{01}(X_j) - \sum_{j\in \mathcal{I}_{10}}w_{ij}\mu_l^{10}(X_j) + \sum_{j\in \mathcal{I}_{00}}w_{ij}\mu_l^{00}(X_j) -\tau_l  \right)^2 \\
        &=\frac{1}{n_{11}}\sum_{i\in \mathcal{I}_{11}}\left(\mu_l^{11}(X_i)-\mu_l^{01}(X_i)-\mu_{l}^{10}(X_i)+\mu_{l}^{00}(X_i)-\tau_l \right. \\
        &\left.+\sum_{j\in \mathcal{I}_{01}}w_{ij}\left(\mu_l^{01}(X_i)- \mu_l^{01}(X_j)\right) + \sum_{j\in \mathcal{I}_{10}}w_{ij}\left(\mu_l^{10}(X_i)- \mu_l^{10}(X_j)\right) - \sum_{j\in \mathcal{I}_{00}}w_{ij}\left(\mu_l^{00}(X_i)- \mu_l^{00}(X_j)\right)   \right)^2 \\
        &=\frac{1}{n_{11}}\sum_{i\in \mathcal{I}_{11}}\left( \mu_l^{11}(X_i)-\mu_l^{01}(X_i)-\mu_{l}^{10}(X_i)+\mu_{l}^{00}(X_i)-\tau_l \right)^2\\
        &+\frac{1}{n_{11}}\sum_{i\in \mathcal{I}_{11}}\left( \sum_{j\in \mathcal{I}_{01}}w_{ij}\left(\mu_l^{01}(X_i)- \mu_l^{01}(X_j)\right) + \sum_{j\in \mathcal{I}_{10}}w_{ij}\left(\mu_l^{10}(X_i)- \mu_l^{10}(X_j)\right) - \sum_{j\in \mathcal{I}_{00}}w_{ij}\left(\mu_l^{00}(X_i)- \mu_l^{00}(X_j)\right)  \right)^2 \\
        &+\frac{2}{n_{11}}\sum_{i\in \mathcal{I}_{11}}\left( \mu_l^{11}(X_i)-\mu_l^{01}(X_i)-\mu_{l}^{10}(X_i)+\mu_{l}^{00}(X_i)-\tau_l \right) \\
        &\left( \sum_{j\in \mathcal{I}_{01}}w_{ij}\left(\mu_l^{01}(X_i)- \mu_l^{01}(X_j)\right) + \sum_{j\in \mathcal{I}_{10}}w_{ij}\left(\mu_l^{10}(X_i)- \mu_l^{10}(X_j)\right) - \sum_{j\in \mathcal{I}_{00}}w_{ij}\left(\mu_l^{00}(X_i)- \mu_l^{00}(X_j)\right)  \right)
    \end{align*}
   Notice that:
    \begin{align*}
        &\frac{1}{n_{11}}\sum_{i\in \mathcal{I}_{11}}\left( \sum_{j\in \mathcal{I}_{01}}w_{ij}\left(\mu_l^{01}(X_i)- \mu_l^{01}(X_j)\right) + \sum_{j\in \mathcal{I}_{10}}w_{ij}\left(\mu_l^{10}(X_i)- \mu_l^{10}(X_j)\right) - \sum_{j\in \mathcal{I}_{00}}w_{ij}\left(\mu_l^{00}(X_i)- \mu_l^{00}(X_j)\right)  \right)^2 \\
        &\leq \frac{3}{n_{11}}\sum_{i\in \mathcal{I}_{11}}\left( \sum_{j\in \mathcal{I}_{01}}w_{ij}\left(\mu_l^{01}(X_i)- \mu_l^{01}(X_j)\right) \right)^2 +\left( \sum_{j\in \mathcal{I}_{10}}w_{ij}\left(\mu_l^{10}(X_i)- \mu_l^{10}(X_j)\right)  \right)^2 \\
        &+ \left( \sum_{j\in \mathcal{I}_{00}}w_{ij}\left(\mu_l^{00}(X_i)- \mu_l^{00}(X_j)\right)  \right)^2 \\
        &\leq \frac{3}{n_{11}}\sum_{i\in \mathcal{I}_{11}}\left( \sum_{j\in \mathcal{I}_{01}}w_{ij}\abs{\left(\mu_l^{01}(X_i)- \mu_l^{01}(X_j)\right) }\right)^2 +\left( \sum_{j\in \mathcal{I}_{10}}w_{ij}\abs{\left(\mu_l^{10}(X_i)- \mu_l^{10}(X_j)\right)}  \right)^2 \\
        &+ \left( \sum_{j\in \mathcal{I}_{00}}w_{ij}\abs{\left(\mu_l^{00}(X_i)- \mu_l^{00}(X_j)\right)}  \right)^2 \\
        &\leq \frac{3C}{n_{11}}\sum_{i\in \mathcal{I}_{11}}\left( \sum_{j\in \mathcal{I}_{01}}w_{ij}\norm{X_i-X_j}\right)^2 +\left( \sum_{j\in \mathcal{I}_{10}}w_{ij}\norm{X_i-X_j}  \right)^2 + \left( \sum_{j\in \mathcal{I}_{00}}w_{ij}\norm{X_i-X_j}  \right)^2 \\
        &=o_p(1)
    \end{align*}
    The first inequality follows from $(a+b+c)^2\leq 3(a^2+b^2+c^2)$, the second inequality follows from triangular inequality, the third inequality follows from the Lipschitz continuity of the  conditional mean function. The last equality follows from lemma \ref{lemma:weightproperty}(iv).
    In addition, 
    \begin{align*}
        &\frac{2}{n_{11}}\sum_{i\in \mathcal{I}_{11}}\abs{\left( \mu_l^{11}(X_i)-\mu_l^{01}(X_i)-\mu_{l}^{10}(X_i)+\mu_{l}^{00}(X_i)-\tau_l \right) }\\
        &\abs{\left( \sum_{j\in \mathcal{I}_{01}}w_{ij}\left(\mu_l^{01}(X_i)- \mu_l^{01}(X_j)\right) + \sum_{j\in \mathcal{I}_{10}}w_{ij}\left(\mu_l^{10}(X_i)- \mu_l^{10}(X_j)\right) - \sum_{j\in \mathcal{I}_{00}}w_{ij}\left(\mu_l^{00}(X_i)- \mu_l^{00}(X_j)\right)  \right)}\\
        &\leq \frac{2C}{n_{11}}\sum_{j\in \mathcal{I}_{11}}\sum_{j\in \mathcal{I}_{01}}w_{ij}\abs{\left(\mu_l^{01}(X_i)- \mu_l^{01}(X_j)\right)}   + \sum_{j\in \mathcal{I}_{10}}w_{ij}\abs{\left(\mu_l^{10}(X_i)- \mu_l^{10}(X_j)\right)} \\
        &+\sum_{j\in \mathcal{I}_{00}}w_{ij}\abs{\left(\mu_l^{00}(X_i)- \mu_l^{00}(X_j)\right)}  \\
        &\leq \frac{2C'}{n_{11}}\sum_{j\in \mathcal{I}_{11}}\sum_{j\in \mathcal{I}_{01}}w_{ij}\norm{X_i-X_j}   + \sum_{j\in \mathcal{I}_{10}}w_{ij}\norm{X_i-X_j} +\sum_{j\in \mathcal{I}_{00}}w_{ij}\norm{X_i-X_j} \\
        &=o_p(1)
    \end{align*}
    This implies that
    \begin{align*}
        &\frac{2}{n_{11}}\sum_{i\in \mathcal{I}_{11}}\left( \mu_l^{11}(X_i)-\mu_l^{01}(X_i)-\mu_{l}^{10}(X_i)+\mu_{l}^{00}(X_i)-\tau_l \right) \\
        &\left( \sum_{j\in \mathcal{I}_{01}}w_{ij}\left(\mu_l^{01}(X_i)- \mu_l^{01}(X_j)\right) + \sum_{j\in \mathcal{I}_{10}}w_{ij}\left(\mu_l^{10}(X_i)- \mu_l^{10}(X_j)\right) - \sum_{j\in \mathcal{I}_{00}}w_{ij}\left(\mu_l^{00}(X_i)- \mu_l^{00}(X_j)\right)  \right) =o_p(1)
    \end{align*}
    Therefore, combine the results above,
    \begin{align*}
        &\frac{1}{n_{11}}\sum_{i\in \mathcal{I}_{11}}\left(  \mu_l^{11}(X_i)-\sum_{j\in \mathcal{I}_{01}}w_{ij}\mu_l^{01}(X_j) - \sum_{j\in \mathcal{I}_{10}}w_{ij}\mu_l^{10}(X_j) + \sum_{j\in \mathcal{I}_{00}}w_{ij}\mu_l^{00}(X_j) -\tau_l  \right)^2 \\
        &=\frac{1}{n_{11}}\sum_{i\in \mathcal{I}_{11}}\left(\mu_l^{11}(X_i)-\mu_l^{01}(X_i)-\mu_{l}^{10}(X_i)+\mu_{l}^{00}(X_i)-\tau_l\right)^2 +o_p(1)
    \end{align*}
    Notice that by WLLN,
\begin{align*}
    &\frac{1}{n_{11}}\sum_{i\in \mathcal{I}_{11}}\left(\mu_l^{11}(X_i)-\mu_l^{01}(X_i)-\mu_{l}^{10}(X_i)+\mu_{l}^{00}(X_i)-\tau_l\right)^2 \\
    &\xrightarrow{p} V_{\tau(X)}= \E\left[\left(\mu_l^{11}(X_i)-\mu_l^{01}(X_i)-\mu_{l}^{10}(X_i)+\mu_{l}^{00}(X_i)-\tau_l\right)^2 |G_i=1,S_i=1\right]
\end{align*}
 Therefore,
 \begin{align*}
     &\frac{1}{n_{11}}\sum_{i\in \mathcal{I}_{11}}\left(  \mu_l^{11}(X_i)-\sum_{j\in \mathcal{I}_{01}}w_{ij}\mu_l^{01}(X_j) - \sum_{j\in \mathcal{I}_{10}}w_{ij}\mu_l^{10}(X_j) + \sum_{j\in \mathcal{I}_{00}}w_{ij}\mu_l^{00}(X_j) -\tau_l  \right)^2 =V_{\tau(X)}+o_p(1)
 \end{align*}
    \textbf{Second term:}
    \begin{align*}
    &\frac{1}{n_{11}}\sum_{i\in \mathcal{I}_{11}}\left(\epsilon_{il}^{11} -\sum_{j\in \mathcal{I}_{01}}w_{ij}\epsilon_{jl}^{01} -\sum_{j\in \mathcal{I}_{10}}w_{ij}\epsilon_{jl}^{10} + \sum_{j\in \mathcal{I}_{00}}w_{ij}\epsilon_{jl}^{00} \right)^2 \\
    &=\frac{1}{n_{11}}\sum_{i\in \mathcal{I}_{11}}\left( \left(\epsilon_{il}^{11}\right)^2 + \sum_{j\in \mathcal{I}_{01}}w_{ij}^2\left(\epsilon_{jl}^{01} \right)^2 +\sum_{j\in \mathcal{I}_{10}}w_{ij}^2\left(\epsilon_{jl}^{10} \right)^2 + \sum_{j\in \mathcal{I}_{00}}w_{ij}^2\left(\epsilon_{jl}^{00} \right)^2 \right) \\
    &\quad+\frac{1}{n_{11}}\sum_{i\in \mathcal{I}_{11}}\left(\sum_{j\in \mathcal{I}_{01}}\sum_{\substack{j'\in\mathcal{I}_{01}\\ j'\neq j}}w_{ij}w_{ij'}\epsilon_{jl}^{01}\epsilon_{j'l}^{01} +\sum_{j\in \mathcal{I}_{10}}\sum_{\substack{j'\in\mathcal{I}_{10}\\ j'\neq j}}w_{ij}w_{ij'}\epsilon_{jl}^{10}\epsilon_{j'l}^{10} + \sum_{j\in \mathcal{I}_{00}}\sum_{\substack{j'\in\mathcal{I}_{00}\\ j'\neq j}}w_{ij}w_{ij'}\epsilon_{jl}^{00}\epsilon_{j'l}^{00}  \right)\\
    &\quad+ \frac{2}{n_{11}}\sum_{i\in\mathcal{I}_{11}}\left(\sum_{j\in \mathcal{I}_{01}}\sum_{j'\in\mathcal{I}_{10}} w_{ij}w_{ij'} \epsilon_{jl}^{01}\epsilon_{j'l}^{10} - \sum_{j\in \mathcal{I}_{01}}\sum_{j'\in\mathcal{I}_{00}} w_{ij}w_{ij'} \epsilon_{jl}^{01}\epsilon_{j'l}^{00} - \sum_{j\in \mathcal{I}_{10}}\sum_{j'\in\mathcal{I}_{00}} w_{ij}w_{ij'} \epsilon_{jl}^{10}\epsilon_{j'l}^{00}  \right)   \\
    &\quad-\frac{2}{n_{11}}\sum_{i\in\mathcal{I}_{11}}\left(\sum_{j\in \mathcal{I}_{01}} w_{ij}\epsilon_{il}^{11}\epsilon_{jl}^{01} + \sum_{j\in \mathcal{I}_{10}} w_{ij}\epsilon_{il}^{11}\epsilon_{jl}^{10} - \sum_{j\in \mathcal{I}_{00}} w_{ij}\epsilon_{il}^{11}\epsilon_{jl}^{00}  \right)
\end{align*}
 \begin{align*}
     &\E\left[\left.\frac{1}{n_{11}}\sum_{i\in \mathcal{I}_{11}}\left(\epsilon_{il}^{11}\right)^2-\sigma_{11,l}^2(X_i)\right\lvert\G,\bs,\X \right] \\
     &=\frac{1}{n_{11}}\sum_{i\in \mathcal{I}_{11}} \E\left[\left(\epsilon_{il}^{11}\right)^2-\sigma_{11,l}^2(X_i)\mid G_i,S_i,X_i \right] \\
     &=0
 \end{align*}
 \begin{align*}
     Var\left(\left.\frac{1}{n_{11}}\sum_{i\in \mathcal{I}_{11}}\left(\epsilon_{il}^{11}\right)^2-\sigma_{11,l}^2(X_i)\right\lvert \G,\bs,\X \right) & = \frac{1}{n_{11}}Var\left(\left.\left(\epsilon_{il}^{11}\right)^2-\sigma_{11,l}^2(X_i)\right\lvert \G,\bs,\X,G_i=1,S_i=1 \right) \\
     &=\frac{1}{n_{11}}Var\left(\left.\left(\epsilon_{il}^{11}\right)^2\right\lvert \G,\bs,\X,G_i=1,S_i=1 \right) \\
     &\leq\frac{1}{n_{11}}\E\left[\left.\left(\epsilon_{il}^{11}\right)^4\right\lvert \G,\bs,\X,G_i=1,S_i=1 \right]\\
     &\leq \frac{C}{n_{11}} \\
     &=o(1)
 \end{align*}
 The last inequality follows from lemma \ref{lemma:residualbound}. Then by lemma \ref{lemma:conditionalchebyshev},
 \begin{align*}
     \frac{1}{n_{11}}\sum_{i\in \mathcal{I}_{11}}\left(\epsilon_{il}^{11}\right)^2-\sigma_{11,l}^2(X_i) =o_p(1)
 \end{align*}
     For any $(g,s)\in \{(0,1),(1,0),(0,0)\}$,
\begin{align*}
    &\E\left[\left. \frac{1}{n_{11}}\sum_{i\in \mathcal{I}_{11}}\sum_{j\in \mathcal{I}_{gs}}w_{ij}^2\left[\left(\epsilon_{jl}^{gs} \right)^2-\sigma^2_{gs,l}(X_j) \right] \right\lvert \G,\bs,\X  \right] \\
    &=\frac{1}{n_{11}}\sum_{i\in \mathcal{I}_{11}}\sum_{j\in \mathcal{I}_{gs}}w_{ij}^2\E\left[\left. \left(\epsilon_{jl}^{gs} \right)^2-\sigma^2_{gs,l}(X_j)  \right\lvert G_j,S_j,X_j  \right] \\
    &=0
\end{align*}
    
\begin{align*}
    &Var\left(\left. \frac{1}{n_{11}}\sum_{i\in \mathcal{I}_{11}}\sum_{j\in \mathcal{I}_{gs}}w_{ij}^2\left[\left(\epsilon_{jl}^{gs} \right)^2-\sigma^2_{gs,l}(X_j) \right] \right\lvert \G,\bs,\X \right) \\
    &=Var\left(\left. \frac{1}{n_{11}}\sum_{j\in \mathcal{I}_{gs}}\sum_{i\in \mathcal{I}_{11}}  w_{ij}^2\left[\left(\epsilon_{jl}^{gs} \right)^2-\sigma^2_{gs,l}(X_j) \right] \right\lvert \G,\bs,\X \right) \\
    &= \frac{1}{n_{11}^2}\sum_{j \in \mathcal{I}_{gs}}Var\left(\left.\sum_{i\in \mathcal{I}_{11}}w_{ij}^2\left(\epsilon_{jl}^{gs} \right)^2  \right\lvert \G,\bs,\X \right) \\
    &\leq \frac{1}{n_{11}^2}\sum_{j \in \mathcal{I}_{gs}}\E\left[\left.\left(\sum_{i\in \mathcal{I}_{11}}w_{ij}^2\left(\epsilon_{jl}^{gs} \right)^2\right)^2  \right\lvert \G,\bs,\X \right] \\
    &=\frac{1}{n_{11}^2}\sum_{j \in \mathcal{I}_{gs}}\E\left[\left.\sum_{i\in \mathcal{I}_{11}}\sum_{i'\in \mathcal{I}_{11}}  w_{ij}^2 w_{i'j}^2 \left(\epsilon_{jl}^{gs} \right)^4 \right\lvert \G,\bs,\X \right] \\
    &=\frac{1}{n_{11}^2}\sum_{j \in \mathcal{I}_{gs}}\sum_{i\in \mathcal{I}_{11}}\sum_{i'\in \mathcal{I}_{11}}  w_{ij}^2 w_{i'j}^2  \E\left[\left. \left(\epsilon_{jl}^{gs} \right)^4 \right\lvert \G,\bs,\X \right] \\
    &\leq \frac{C}{n_{11}^2}\sum_{j \in \mathcal{I}_{gs}}\sum_{i\in \mathcal{I}_{11}}\sum_{i'\in \mathcal{I}_{11}}  w_{ij}^2 w_{i'j}^2 \\
    &\leq \frac{C}{n_{11}^2}\sum_{j \in \mathcal{I}_{gs}}\sum_{i\in \mathcal{I}_{11}} w_{ij} \sum_{i'\in \mathcal{I}_{11}}  w_{i'j} \\
    &=\frac{C}{n_{11}^2}\sum_{j \in \mathcal{I}_{gs}} w_j^2 \\
    &=o_p(1)
\end{align*}
This last line follows from lemma \ref{lemma:weightproperty}(i). By lemma \ref{lemma:conditionalchebyshev},
\[\frac{1}{n_{11}}\sum_{i\in \mathcal{I}_{11}}\sum_{j\in \mathcal{I}_{gs}}w_{ij}^2\left[\left(\epsilon_{jl}^{gs} \right)^2-\sigma^2_{gs,l}(X_j) \right]  = o_p(1)   \]
    For $(g,s)\in \{(0,1),(1,0),(0,0)\}$,
    \begin{align*}
        &\E\left[\left. \frac{1}{n_{11}}\sum_{i\in \mathcal{I}_{11}}\sum_{j\in \mathcal{I}_{gs}}\sum_{\substack{j'\in\mathcal{I}_{gs}\\ j'\neq j}}w_{ij}w_{ij'}\epsilon_{jl}^{gs}\epsilon_{j'l}^{gs} \right\lvert \G,\bs,\X  \right] \\
        &=\frac{1}{n_{11}}\sum_{i\in \mathcal{I}_{11}}\sum_{j\in \mathcal{I}_{gs}}\sum_{\substack{j'\in\mathcal{I}_{gs}\\ j'\neq j}}w_{ij}w_{ij'} \E\left[\left. \epsilon_{jl}^{gs}\epsilon_{j'l}^{gs} \right\lvert \G,\bs,\X  \right] \\
        &=\frac{1}{n_{11}}\sum_{i\in \mathcal{I}_{11}}\sum_{j\in \mathcal{I}_{gs}}\sum_{\substack{j'\in\mathcal{I}_{gs}\\ j'\neq j}}w_{ij}w_{ij'} \E\left[\left. \epsilon_{jl}^{gs} \right\lvert \G,\bs,\X  \right]\E\left[\left. \epsilon_{j'l}^{gs} \right\lvert \G,\bs,\X  \right] \\
        &=0
    \end{align*}
    On the other hand,
    \begin{align*}
        &Var\left(\left. \frac{1}{n_{11}}\sum_{i\in \mathcal{I}_{11}}\sum_{j\in \mathcal{I}_{gs}}\sum_{\substack{j'\in\mathcal{I}_{gs}\\ j'\neq j}}w_{ij}w_{ij'}\epsilon_{jl}^{gs}\epsilon_{j'l}^{gs} \right\lvert \G,\bs,\X \right)\\
        &=\frac{1}{n_{11}^2}Var\left(\left. \sum_{j\in \mathcal{I}_{gs}}\sum_{\substack{j'\in\mathcal{I}_{gs}\\ j'\neq j}} \sum_{i\in \mathcal{I}_{11}}w_{ij}w_{ij'} \epsilon_{jl}^{gs}\epsilon_{j'l}^{gs} \right\lvert \G,\bs,\X \right) \\
        &=\frac{1}{n_{11}^2}\sum_{j\in \mathcal{I}_{gs}}\sum_{\substack{j'\in\mathcal{I}_{gs}\\ j'\neq j}} Var\left(\left. \sum_{i\in \mathcal{I}_{11}}w_{ij}w_{ij'} \epsilon_{jl}^{gs}\epsilon_{j'l}^{gs} \right\lvert \G,\bs,\X \right) \\
        &\leq\frac{1}{n_{11}^2}\sum_{j\in \mathcal{I}_{gs}}\sum_{\substack{j'\in\mathcal{I}_{gs}\\ j'\neq j}}\left( \sum_{i\in \mathcal{I}_{11}}w_{ij}w_{ij'}\right)^2\E\left[\left. \left( \epsilon_{jl}^{gs}\right)^2\left(\epsilon_{j'l}^{gs}\right)^2 \right\lvert \G,\bs,\X \right]  \\
        &\leq \frac{C}{n_{11}^2}\sum_{j\in \mathcal{I}_{gs}}\sum_{\substack{j'\in\mathcal{I}_{gs}\\ j'\neq j}}\left( \sum_{i\in \mathcal{I}_{11}}w_{ij}w_{ij'}\right)^2  \allowdisplaybreaks \\
        &\leq \frac{C}{n_{11}^2}\max_{r\in \mathcal{I}_{gs}}w_r\sum_{j\in \mathcal{I}_{gs}}\sum_{\substack{j'\in\mathcal{I}_{gs}\\ j'\neq j}} \sum_{i\in \mathcal{I}_{11}}w_{ij}w_{ij'} \\
        &\leq\frac{C}{n_{11}^2}\max_{r\in \mathcal{I}_{gs}}w_r\sum_{i\in \mathcal{I}_{11}}  \sum_{j\in \mathcal{I}_{gs}}w_{ij}  \sum_{j'\in\mathcal{I}_{gs}} w_{ij'} \\
        &=\frac{C}{n_{11}}O_p(n^{1/2}) \\
        &=o_p(1)
    \end{align*}
    The second equality follows because all the cross product term has zero conditional covariance.  By lemma \ref{lemma:conditionalchebyshev},
\begin{align*}
    \frac{1}{n_{11}}\sum_{i\in \mathcal{I}_{11}}\sum_{j\in \mathcal{I}_{gs}}\sum_{\substack{j'\in\mathcal{I}_{gs}\\ j'\neq j}}w_{ij}w_{ij'}\epsilon_{jl}^{gs}\epsilon_{j'l}^{gs} =o_p(1)
\end{align*}
Therefore,
\begin{align*}
    \frac{1}{n_{11}}\sum_{i\in \mathcal{I}_{11}}\left(\sum_{j\in \mathcal{I}_{01}}\sum_{\substack{j'\in\mathcal{I}_{01}\\ j'\neq j}}w_{ij}w_{ij'}\epsilon_{jl}^{01}\epsilon_{j'l}^{01} +\sum_{j\in \mathcal{I}_{10}}\sum_{\substack{j'\in\mathcal{I}_{10}\\ j'\neq j}}w_{ij}w_{ij'}\epsilon_{jl}^{10}\epsilon_{j'l}^{10} + \sum_{j\in \mathcal{I}_{00}}\sum_{\substack{j'\in\mathcal{I}_{00}\\ j'\neq j}}w_{ij}w_{ij'}\epsilon_{jl}^{00}\epsilon_{j'l}^{00}  \right) =o_p(1)
\end{align*}
For $(g,s)\in \{(0,1),(1,0),(0,0)  \}$ and $(g',s')\in \{(0,1),(1,0),(0,0)\}$
\begin{align*}
    &\E\left[\left. \frac{1}{n_{11}}\sum_{i\in \mathcal{I}_{11}}\sum_{j\in \mathcal{I}_{gs}}\sum_{j\in \mathcal{I}_{g's'}} w_{ij}w_{ij'}\epsilon_{jl}^{gs}\epsilon_{j'l}^{g's'} \right\lvert \G,\bs,\X \right] \\&=\frac{1}{n_{11}}\sum_{i\in \mathcal{I}_{11}}\sum_{j\in \mathcal{gs}}\sum_{j\in \mathcal{g's'}} w_{ij}w_{ij'} \E\left[\left. \epsilon_{jl}^{gs}\epsilon_{j'l}^{g's'} \right\lvert \G,\bs,\X \right] \\
    &=\frac{1}{n_{11}}\sum_{i\in \mathcal{I}_{11}}\sum_{j\in \mathcal{gs}}\sum_{j\in \mathcal{g's'}} w_{ij}w_{ij'} \E\left[\left. \epsilon_{jl}^{gs} \right\lvert \G,\bs,\X \right] \E\left[\left. \epsilon_{j'l}^{g's'} \right\lvert \G,\bs,\X \right]\\
    &=0
\end{align*}
On the other hand,
\begin{align*}
    &Var\left(\left. \frac{1}{n_{11}}\sum_{i\in \mathcal{I}_{11}}\sum_{j\in \mathcal{I}_{gs}}\sum_{j'\in \mathcal{I}_{g's'}} w_{ij}w_{ij'}\epsilon_{jl}^{gs}\epsilon_{j'l}^{g's'} \right\lvert \G,\bs,\X \right) \\
    &=\frac{1}{n_{11}^2}Var\left(\left. \sum_{j\in \mathcal{I}_{gs}}\sum_{j'\in \mathcal{I}_{g's'}}\sum_{i\in \mathcal{I}_{11}} w_{ij}w_{ij'}\epsilon_{jl}^{gs}\epsilon_{j'l}^{g's'} \right\lvert \G,\bs,\X \right) \\
    &=\frac{1}{n_{11}^2}\sum_{j\in \mathcal{I}_{gs}}\sum_{j'\in \mathcal{I}_{g's'}}  Var\left(\left. \sum_{i\in \mathcal{I}_{11}} w_{ij}w_{ij'}\epsilon_{jl}^{gs}\epsilon_{j'l}^{g's'} \right\lvert \G,\bs,\X \right) \\
    &=\frac{1}{n_{11}^2}\sum_{j\in \mathcal{I}_{gs}}\sum_{j'\in \mathcal{I}_{g's'}}\left(\sum_{i\in \mathcal{I}_{11}} w_{ij}w_{ij'} \right)^2 Var\left(\left. \epsilon_{jl}^{gs}\epsilon_{j'l}^{g's'} \right\lvert \G,\bs,\X \right) \\
    &\leq \frac{1}{n_{11}^2}\sum_{j\in \mathcal{I}_{gs}}\sum_{j'\in \mathcal{I}_{g's'}}\left(\sum_{i\in \mathcal{I}_{11}} w_{ij}w_{ij'} \right)^2\E\left[\left.\left( \epsilon_{jl}^{gs}\right)^2\left(\epsilon_{j'l}^{g's'}\right)^2 \right\lvert \G,\bs,\X \right] \allowdisplaybreaks \\ 
    &\leq  \frac{C}{n_{11}^2}\sum_{j\in \mathcal{I}_{gs}}\sum_{j'\in \mathcal{I}_{g's'}}\left(\sum_{i\in \mathcal{I}_{11}} w_{ij}w_{ij'} \right)^2 \\
    &\leq \frac{C}{n_{11}^2}\sum_{j\in \mathcal{I}_{gs}}\sum_{j'\in \mathcal{I}_{g's'}}w_j\sum_{i\in \mathcal{I}_{11}} w_{ij}w_{ij'}  \\
    &\leq \frac{C}{n_{11}^2}\max_{r\in \mathcal{I}_{gs}}w_r\sum_{j\in \mathcal{I}_{gs}}\sum_{j'\in \mathcal{I}_{g's'}}\sum_{i\in \mathcal{I}_{11}} w_{ij}w_{ij'} \\
    &=\frac{C}{n_{11}^2}\max_{r\in \mathcal{I}_{gs}}w_r\sum_{i\in \mathcal{I}_{11}}  \sum_{j\in \mathcal{I}_{gs}}w_{ij}  \sum_{j'\in \mathcal{I}_{g's'}} w_{ij'} \\
    &=\frac{C}{n_{11}}\max_{r\in \mathcal{I}_{gs}}w_r \\
    &=\frac{C}{n_{11}}O_p(n^{1/2}) \\
    &=o_p(1)
\end{align*}
 By lemma \ref{lemma:conditionalchebyshev},
 \begin{align*}
     \frac{1}{n_{11}}\sum_{i\in \mathcal{I}_{11}}\sum_{j\in \mathcal{I}_{gs}}\sum_{j'\in \mathcal{I}_{g's'}} w_{ij}w_{ij'}\epsilon_{jl}^{gs}\epsilon_{j'l}^{g's'} = o_p(1)
 \end{align*}
 This implies that,
 \begin{align*}
     \frac{2}{n_{11}}\sum_{i\in\mathcal{I}_{11}}\left(\sum_{j\in \mathcal{I}_{01}}\sum_{j'\in\mathcal{I}_{10}} w_{ij}w_{ij'} \epsilon_{jl}^{01}\epsilon_{j'l}^{10} - \sum_{j\in \mathcal{I}_{01}}\sum_{j'\in\mathcal{I}_{00}} w_{ij}w_{ij'} \epsilon_{jl}^{01}\epsilon_{j'l}^{00} - \sum_{j\in \mathcal{I}_{10}}\sum_{j'\in\mathcal{I}_{00}} w_{ij}w_{ij'} \epsilon_{jl}^{10}\epsilon_{j'l}^{00}  \right) =o_p(1)
 \end{align*}

 For $(g,s)\in \{(0,1),(1,0),(0,0)  \}$, 
 \begin{align*}
     \E\left[\left. \frac{1}{n_{11}}\sum_{i\in\mathcal{I}_{11}}\sum_{j\in \mathcal{I}_{gs}} w_{ij}\epsilon_{il}^{11}\epsilon_{jl}^{gs}  \right\lvert \G,\bs,\X \right] &= \frac{1}{n_{11}}\sum_{i\in\mathcal{I}_{11}}\sum_{j\in \mathcal{I}_{gs}} w_{ij} \E\left[\left. \epsilon_{il}^{11}\epsilon_{jl}^{gs}  \right\lvert \G,\bs,\X \right] \\
     &=\frac{1}{n_{11}}\sum_{i\in\mathcal{I}_{11}}\sum_{j\in \mathcal{I}_{gs}} w_{ij} \E\left[\left. \epsilon_{il}^{11}  \right\lvert \G,\bs,\X \right] \E\left[\left. \epsilon_{jl}^{gs}  \right\lvert \G,\bs,\X \right] \\
     &=0
 \end{align*}
On the other hand,
\begin{align*}
     &Var\left(\left. \frac{1}{n_{11}}\sum_{i\in\mathcal{I}_{11}}\sum_{j\in \mathcal{I}_{gs}} w_{ij}\epsilon_{il}^{11}\epsilon_{jl}^{gs}  \right\lvert \G,\bs,\X \right) \\
     &=\frac{1}{n_{11}^2}\sum_{i\in\mathcal{I}_{11}} \sum_{j\in \mathcal{I}_{gs}}w_{ij}^2Var\left(\left.  \epsilon_{il}^{11}\epsilon_{jl}^{gs}  \right\lvert \G,\bs,\X \right) \\
     &\leq \frac{1}{n_{11}^2}\sum_{i\in\mathcal{I}_{11}} \sum_{j\in \mathcal{I}_{gs}}w_{ij}^2 \E\left[\left. \left( \epsilon_{il}^{11}\right)^2\left(\epsilon_{jl}^{gs}\right)^2 \right\lvert \G,\bs,\X \right]\\
     &\leq \frac{C}{n_{11}^2}\sum_{i\in\mathcal{I}_{11}} \sum_{j\in \mathcal{I}_{gs}}w_{ij}^2 \\
     &\leq  \frac{C}{n_{11}^2} \sum_{i\in\mathcal{I}_{11}} \sum_{j\in \mathcal{I}_{gs}} w_{ij}\\
     &=\frac{C}{n_{11}} \\
     &=o_p(1)
\end{align*}
 By lemma \ref{lemma:conditionalchebyshev}, 
 \begin{align*}
     \frac{1}{n_{11}}\sum_{i\in\mathcal{I}_{11}}\sum_{j\in \mathcal{I}_{gs}} w_{ij}\epsilon_{il}^{11}\epsilon_{jl}^{gs} =o_p(1)
 \end{align*}
 This implies that:
 \begin{align*}
     \frac{2}{n_{11}}\sum_{i\in\mathcal{I}_{11}}\left(\sum_{j\in \mathcal{I}_{01}} w_{ij}\epsilon_{il}^{11}\epsilon_{jl}^{01} + \sum_{j\in \mathcal{I}_{10}} w_{ij}\epsilon_{il}^{11}\epsilon_{jl}^{10} - \sum_{j\in \mathcal{I}_{00}} w_{ij}\epsilon_{il}^{11}\epsilon_{jl}^{00}  \right) =o_p(1)
 \end{align*}
 Therefore, 
\begin{align*}
     &\frac{1}{n_{11}}\sum_{i\in \mathcal{I}_{11}}\left(\epsilon_{il}^{11} -\sum_{j\in \mathcal{I}_{01}}w_{ij}\epsilon_{jl}^{01} -\sum_{j\in \mathcal{I}_{10}}w_{ij}\epsilon_{jl}^{10} + \sum_{j\in \mathcal{I}_{00}}w_{ij}\epsilon_{jl}^{00} \right)^2 \\
    &=\frac{1}{n_{11}}\sum_{i\in \mathcal{I}_{11}}\left( \sigma^2_{11,l}(X_i) + \sum_{j\in \mathcal{I}_{01}}w_{ij}^2\sigma_{01,l}^2(X_j) +\sum_{j\in \mathcal{I}_{10}}w_{ij}^2\sigma_{10,l}^2(X_j) + \sum_{j\in \mathcal{I}_{00}}w_{ij}^2\sigma_{00,l}^2(X_j) \right) +o_p(1)
\end{align*}
\textbf{Third term:}
\begin{align*}
    &\frac{2}{n_{11}}\sum_{i\in \mathcal{I}_{11}}\left(  \mu_l^{11}(X_i)-\sum_{j\in \mathcal{I}_{01}}w_{ij}\mu_l^{01}(X_j) - \sum_{j\in \mathcal{I}_{10}}w_{ij}\mu_l^{10}(X_j) + \sum_{j\in \mathcal{I}_{00}}w_{ij}\mu_l^{00}(X_j) -\tau_l  \right)\\
        &\left(\epsilon_{il}^{11} -\sum_{j\in \mathcal{I}_{01}}w_{ij}\epsilon_{jl}^{01} -\sum_{j\in \mathcal{I}_{10}}w_{ij}\epsilon_{jl}^{10} + \sum_{j\in \mathcal{I}_{00}}w_{ij}\epsilon_{jl}^{00} \right)
\end{align*}
Define $\Phi_{1,i}= \mu_l^{11}(X_i)-\sum_{j\in \mathcal{I}_{01}}w_{ij}\mu_l^{01}(X_j) - \sum_{j\in \mathcal{I}_{10}}w_{ij}\mu_l^{10}(X_j) + \sum_{j\in \mathcal{I}_{00}}w_{ij}\mu_l^{00}(X_j) -\tau_l $.
First notice that
\begin{align*}
    \abs{\Phi_{1,i}}&=\abs{ \mu_l^{11}(X_i)-\sum_{j\in \mathcal{I}_{01}}w_{ij}\mu_l^{01}(X_j) - \sum_{j\in \mathcal{I}_{10}}w_{ij}\mu_l^{10}(X_j) + \sum_{j\in \mathcal{I}_{00}}w_{ij}\mu_l^{00}(X_j) -\tau_l }\\
    &\leq \abs{\mu_l^{11}(X_i)} +\sum_{j\in \mathcal{I}_{01}}w_{ij}\abs{\mu_l^{01}(X_j)}+
    \sum_{j\in \mathcal{I}_{01}}w_{ij}\abs{\mu_l^{10}(X_j)} +\sum_{j\in \mathcal{I}_{01}}w_{ij}\abs{\mu_l^{00}(X_j)} +\abs{\tau_l}\\
    &\leq C
\end{align*}
\begin{align*}
    \E\left[\left. \frac{1}{n_{11}} \sum_{i\in \mathcal{I}_{11}} \Phi_{1,i} \epsilon_{il}^{11}   \right| \G,\bs,\X \right] =\frac{1}{n_{11}} \sum_{i\in \mathcal{I}_{11}} \Phi_{1,i}\E\left[\left. \epsilon_{il}^{11}\right|\G,\bs,\X \right] =0
\end{align*}
On the other hand,
\begin{align*}
     Var\left(\left. \frac{1}{n_{11}} \sum_{i\in \mathcal{I}_{11}} \Phi_{1,i} \epsilon_{il}^{11}   \right| \G,\bs,\X \right)&=\frac{1}{n_{11}^2}Var\left(\left. \sum_{i\in \mathcal{I}_{11}} \Phi_{1,i} \epsilon_{il}^{11}   \right| \G,\bs,\X \right)\\
     &=\frac{1}{n_{11}^2}\sum_{i\in \mathcal{I}_{11}} \Phi_{1,i}^2 Var\left(\left.   \epsilon_{il}^{11}   \right| \G,\bs,\X \right) \\
     &\leq \frac{C}{n_{11}^2}\sum_{i\in \mathcal{I}_{11}} \sigma^2_{11,l}(X_i)\\
     &\leq \frac{C\overline{\sigma}^2}{n_{11}}\\
     &=o_p(1)
\end{align*}
By lemma \ref{lemma:conditionalchebyshev},
\begin{align*}
    \frac{1}{n_{11}} \sum_{i\in \mathcal{I}_{11}} \Phi_{1,i} \epsilon_{il}^{11}=o_p(1)
\end{align*}
For $(g,s)\in \{(0,1),(1,0),(0,0)\}$,
\begin{align*}
    \E\left[\left. \frac{1}{n_{11}} \sum_{i\in \mathcal{I}_{11}} \Phi_{1,i} \sum_{j\in \mathcal{I}_{gs}}w_{ij}\epsilon_{jl}^{gs}  \right| \G,\bs,\X \right] &= \frac{1}{n_{11}} \sum_{i\in \mathcal{I}_{11}} \Phi_{1,i}\sum_{j\in \mathcal{I}_{gs}}w_{ij} \E\left[\left. \epsilon_{jl}^{gs}\right|\G,\bs,\X \right] =0
\end{align*}
On the other hand
\begin{align*}
    &Var\left(\left. \frac{1}{n_{11}} \sum_{i\in \mathcal{I}_{11}} \Phi_{1,i} \sum_{j\in \mathcal{I}_{gs}}w_{ij}\epsilon_{jl}^{gs}  \right| \G,\bs,\X \right)\\
    &=\frac{1}{n_{11}^2}Var\left(\left.  \sum_{j\in \mathcal{I}_{gs}} \sum_{i\in \mathcal{I}_{11}} \Phi_{1,i} w_{ij}\epsilon_{jl}^{gs}  \right| \G,\bs,\X \right) \\
    &=\frac{1}{n_{11}^2}\sum_{j\in \mathcal{I}_{gs}} Var\left(\left.   \sum_{i\in \mathcal{I}_{11}} \Phi_{1,i} w_{ij}\epsilon_{jl}^{gs}  \right| \G,\bs,\X \right) \\
    &=\frac{1}{n_{11}^2}\sum_{j\in \mathcal{I}_{gs}}\left( \sum_{i\in \mathcal{I}_{11}} \Phi_{1,i} w_{ij}\right)^2   Var\left(\left.   \epsilon_{jl}^{gs}  \right| G_i=g,S_i=s,X_i \right) \allowdisplaybreaks \\
    &=\frac{C\overline{\sigma}^2}{n_{11}^2}\sum_{j\in \mathcal{I}_{gs}}\left( \sum_{i\in \mathcal{I}_{11}} w_{ij}\right)^2 \\
    &=\frac{C\overline{\sigma}^2}{n_{11}^2}\sum_{j\in \mathcal{I}_{gs}}w_j^2 \\
    &=\frac{n_{gs}C\overline{\sigma}^2}{n_{11}^2}\frac{1}{n_{gs}}\sum_{j\in \mathcal{I}_{gs}}w_j^2\\
    &=\frac{1}{n_{11}}\frac{n_{gs}C\overline{\sigma}^2}{n_{11}}O(1)\\
    &=o_p(1)
\end{align*}
By lemma \ref{lemma:conditionalchebyshev},
\begin{align*}
    \frac{1}{n_{11}} \sum_{i\in \mathcal{I}_{11}} \Phi_{1,i} \sum_{j\in \mathcal{I}_{gs}}w_{ij}\epsilon_{jl}^{gs}  =o_p(1)
\end{align*}
Therefore,
\begin{align*}
    &\frac{2}{n_{11}}\sum_{i\in \mathcal{I}_{11}}\left(  \mu_l^{11}(X_i)-\sum_{j\in \mathcal{I}_{01}}w_{ij}\mu_l^{01}(X_j) - \sum_{j\in \mathcal{I}_{10}}w_{ij}\mu_l^{10}(X_j) + \sum_{j\in \mathcal{I}_{00}}w_{ij}\mu_l^{00}(X_j) -\tau_l  \right)\\
        &\left(\epsilon_{il}^{11} -\sum_{j\in \mathcal{I}_{01}}w_{ij}\epsilon_{jl}^{01} -\sum_{j\in \mathcal{I}_{10}}w_{ij}\epsilon_{jl}^{10} + \sum_{j\in \mathcal{I}_{00}}w_{ij}\epsilon_{jl}^{00} \right)\\
        &=\frac{1}{n_{11}} \sum_{i\in \mathcal{I}_{11}} \Phi_{1,i} \left(\epsilon_{il}^{11} -\sum_{j\in \mathcal{I}_{01}}w_{ij}\epsilon_{jl}^{01} -\sum_{j\in \mathcal{I}_{10}}w_{ij}\epsilon_{jl}^{10} + \sum_{j\in \mathcal{I}_{00}}w_{ij}\epsilon_{jl}^{00} \right) \\
        &=o_p(1)
\end{align*}
Combine the results from the three terms, we have:
\begin{align*}
    \widehat{V}_{\tau(X)}&=\frac{1}{n_{11}}\sum_{i\in \mathcal{I}_{11}}\left(\Delta Y_{il}-  \widehat {\Delta Y_{il}^{01}}-\widehat {\Delta Y_{il}^{10}} + \widehat {\Delta Y_{il}^{00}}-\tau_l \right)^2 +o_p(1) \\
    &=V_{\tau(X)}+\frac{1}{n_{11}}\sum_{i\in \mathcal{I}_{11}}\left( \sigma^2_{11,l}(X_i) + \sum_{j\in \mathcal{I}_{01}}w_{ij}^2\sigma_{01,l}^2(X_j) +\sum_{j\in \mathcal{I}_{10}}w_{ij}^2\sigma_{10,l}^2(X_j) + \sum_{j\in \mathcal{I}_{00}}w_{ij}^2\sigma_{00,l}^2(X_j) \right) +o_p(1) 
\end{align*}

Next, for $(g,s)\in \{(0,1),(1,0),(0,0) \} $,
\begin{align*}
    \widehat{\sigma}_{gs,l}^2(X_j)&=\frac{K}{K+1}\left(\Delta Y_{jl}-\frac{1}{K}\sum_{m=1}^K \Delta Y_{k_{m(j)}l} \right)^2 \\
    &=\frac{K}{K+1}\left(\Delta Y_{jl}-\mu_{l}^{gs}(X_j)- \frac{1}{K}\sum_{m=1}^K\left[\Delta Y_{k_{m(j)}l}- \mu_l^{gs}(X_{{k_{m(j)}}})\right] +\frac{1}{K}\sum_{m=1}^K\left[\mu_l^{gs}(X_j) - \mu_l^{gs}(X_{{k_{m(j)}}}) \right] \right)^2 \\
    &=\frac{K}{K+1}\left(\epsilon_{jl}^{gs}- \frac{1}{K}\sum_{m=1}^K\epsilon_{k_{m(j)}l}^{gs} +\psi_{K,j} \right)^2
\end{align*}
where $\psi_{K,j} =\frac{1}{K}\sum_{m=1}^K\left[\mu_l^{gs}(X_j) - \mu_l^{gs}(X_{{k_{m(j)}}})\right]  $
\begin{align*}
   \E\left[\left. \widehat{\sigma}_{gs,l}^2(X_j) \right\lvert \G,\bs,\X \right]  &=\frac{K}{K+1}\E\left[\left. \left(\epsilon_{jl}^{gs}- \frac{1}{K}\sum_{m=1}^K\epsilon_{k_{m(j)}l}^{gs} +\psi_{K,j} \right)^2 \right\lvert \G,\bs,\X \right] \\
   &=\frac{K}{K+1}\left( \E\left[\left. \left(\epsilon_{jl}^{gs} \right)^2 \right\lvert \G,\bs,\X \right] + \frac{1}{K^2}\sum_{m=1}^K \E\left[\left. \left(\epsilon_{k_{m(j)}l}^{gs} \right)^2 \right\lvert \G,\bs,\X \right] +\psi_{K,j}^2  \right) \\
   &=\frac{K}{K+1}\left( \sigma^2_{gs,l}(X_j) + \frac{1}{K^2}\sum_{m=1}^K \sigma^2_{gs,l}(X_{k_m(j)}) +\psi_{K,j}^2  \right) \\
   &=\sigma^2_{gs,l}(X_j)+\frac{1}{K+1}\sum_{m=1}^K \left( \sigma^2_{gs,l}(X_{k_m(j)})-\sigma^2_{gs,l}(X_j) \right) +\frac{K}{K+1}\psi_{K,j}^2
\end{align*}

\begin{align*}
   \abs{\psi_{K,j}} &=\abs{\frac{1}{K}\sum_{m=1}^K\left[\mu_l^{gs}(X_j) - \mu_l^{gs}(X_{{k_{m(j)}}})\right]} \leq\frac{1}{K}\sum_{m=1}^K\abs{\left[\mu_l^{gs}(X_j) - \mu_l^{gs}(X_{{k_{m(j)}}})\right]} \\
   &\leq \frac{C}{K}\sum_{m=1}^K\norm{X_j-X_{k_{m(j)}}} \leq C\cdot \max_{r\in \mathcal{I}_{gs}}\norm{X_r-X_{k_{m(r)}} } =o_p(1)
\end{align*}
The last equality follows from lemma \ref{lemma:maximumdistancelimit}.

\begin{align*}
   & \abs{ \frac{1}{n_{11}}\sum_{j\in \mathcal{I}_{gs}} w_j^2\left(\E\left[\left. \widehat{\sigma}_{gs,l}^2(X_j)\right|\G,\bs,\X\right]-\sigma_{gs,l}^2(X_j) \right) }\\
   &=\abs{ \frac{1}{n_{11}}\sum_{j\in \mathcal{I}_{gs}} w_j^2\left(\frac{1}{K+1}\sum_{m=1}^K \left( \sigma^2_{gs,l}(X_{k_m(j)})-\sigma^2_{gs,l}(X_j) \right) +\frac{K}{K+1}\psi_{K,j}^2\right) } \\
   &=\abs{ \frac{1}{n_{11}}\sum_{j\in \mathcal{I}_{gs}} w_j^2\left(\frac{1}{K+1}\sum_{m=1}^K \left( \sigma^2_{gs,l}(X_{k_m(j)})-\sigma^2_{gs,l}(X_j) \right) \right) }+\abs{\frac{1}{n_{11}}\sum_{j\in \mathcal{I}_{gs}} w_j^2\left( \frac{K}{K+1}\psi_{K,j}^2\right)} \\
   &\leq  \frac{1}{n_{11}}\sum_{j\in \mathcal{I}_{gs}} w_j^2\left(\frac{1}{K+1}\sum_{m=1}^K \abs{ \left( \sigma^2_{gs,l}(X_{k_m(j)})-\sigma^2_{gs,l}(X_j) \right)} \right) +\frac{1}{n_{11}}\frac{K}{K+1}\sum_{j\in\mathcal{I}_{gs}}w_j^2\abs{\psi_{K,j} }^2 \\
   &\leq \frac{C}{n_{11}}\frac{1}{K+1}\sum_{j\in \mathcal{I}_{gs}} w_j^2\left(\sum_{m=1}^K \norm{X_j-X_{k_{m(j)}} } \right) +\frac{CKn_{gs}}{(K+1)n_{11}}\left(\max_{r\in \mathcal{I}_{gs}}\norm{X_r-X_{k_{m(r)}} }\right)^2 \frac{1}{n_{gs}}  \sum_{j\in\mathcal{I}_{gs}}w_j^2   \\
   &=\frac{CKn_{gs}}{(K+1)n_{11}}\max_{r\in \mathcal{I}_{gs}}\norm{X_r-X_{k_{m(r)}} } \frac{1}{n_{gs}} \sum_{j\in \mathcal{I}_{gs}} w_j^2+o_p(1) \\
   &=o_p(1)
\end{align*}
This implies that,
 \begin{align*}
     \frac{1}{n_{11}}\sum_{j\in \mathcal{I}_{gs}} w_j^2\left(\E\left[\left. \widehat{\sigma}_{gs,l}^2(X_j)\right|\G,\bs,\X\right]-\sigma_{gs,l}^2(X_j) \right)=o_p(1)
 \end{align*}
Since the support of $X$ is compact so $\abs{\psi_{K,j}}$ is uniformly bounded.
\begin{align*}
&\E\left[\left.\frac{1}{n_{11}}\sum_{j\in \mathcal{I}_{gs}} w_j^2\left(\widehat{\sigma}_{gs,l}^2(X_j)-\E\left[\left. \widehat{\sigma}_{gs,l}^2(X_j)\right|\G,\bs,\X\right] \right)\right\lvert \G,\bs,\X \right] \\
&=\frac{1}{n_{11}}\sum_{j\in \mathcal{I}_{gs}} w_j^2\left(\E\left[\left. \widehat{\sigma}_{gs,l}^2(X_j)\right|\G,\bs,\X\right]-\E\left[\left. \widehat{\sigma}_{gs,l}^2(X_j)\right|\G,\bs,\X\right] \right)\\
&=0
\end{align*}

On the other hand,
\begin{align*}
   & Var\left(\left.\frac{1}{n_{11}}\sum_{j\in \mathcal{I}_{gs}} w_j^2\left(\widehat{\sigma}_{gs,l}^2(X_j)-\E\left[\left. \widehat{\sigma}_{gs,l}^2(X_j)\right|\G,\bs,\X\right] \right)\right\lvert \G,\bs,\X \right) \\
    &= Var\left(\left.\frac{1}{n_{11}}\sum_{j\in \mathcal{I}_{gs}} w_j^2\widehat{\sigma}_{gs,l}^2(X_j) \right\lvert \G,\bs,\X \right) \\
    &=\frac{K^2}{n_{11}^2(K+1)^2} Var\left(\left.\sum_{j\in \mathcal{I}_{gs}} w_j^2\left(\epsilon_{jl}^{gs}- \frac{1}{K}\sum_{m=1}^K\epsilon_{k_{m(j)}l}^{gs} +\psi_{K,j} \right)^2 \right\lvert \G,\bs,\X \right) \\
    &=\frac{K^2}{n_{11}^2(K+1)^2}\left( \sum_{j\in \mathcal{I}_{gs}}  Var\left(\left. w_j^2\left(\epsilon_{jl}^{gs}- \frac{1}{K}\sum_{m=1}^K\epsilon_{k_{m(j)}l}^{gs} +\psi_{K,j} \right)^2 \right\lvert \G,\bs,\X \right) \right. \\
    &+\left.\sum_{j\in \mathcal{I}_{gs}}\sum_{\substack{j'\in \mathcal{I}_{gs} \\j'\neq j}}Cov\left(\left.  w_j^2\left(\epsilon_{jl}^{gs}- \frac{1}{K}\sum_{m=1}^K\epsilon_{k_{m(j)}l}^{gs} +\psi_{K,j} \right)^2,  w_{j'}^2\left(\epsilon_{j'l}^{gs}- \frac{1}{K}\sum_{m=1}^K\epsilon_{k_{m(j')}l}^{gs} +\psi_{K,j'} \right)^2   \right\lvert \G,\bs,\X \right) \right) 
\end{align*}
 Now for the variance term,
 \begin{align*}
     &\sum_{j\in \mathcal{I}_{gs}}  Var\left(\left. w_j^2\left(\epsilon_{jl}^{gs}- \frac{1}{K}\sum_{m=1}^K\epsilon_{k_{m(j)}l}^{gs} +\psi_{K,j} \right)^2 \right\lvert \G,\bs,\X \right) \\
     &\leq  \sum_{j\in \mathcal{I}_{gs}} w_j^4\E\left[\left. \left(\epsilon_{jl}^{gs}- \frac{1}{K}\sum_{m=1}^K\epsilon_{k_{m(j)}l}^{gs} +\psi_{K,j} \right)^4 \right\lvert \G,\bs,\X \right] \\
     &\leq  27\sum_{j\in \mathcal{I}_{gs}} w_j^4 \left( \E\left[\left. \left(\epsilon_{jl}^{gs} \right)^4 \right\lvert \G,\bs,\X \right] + \E\left[\left. \left(\frac{1}{K}\sum_{m=1}^K\epsilon_{k_{m(j)}l}^{gs}  \right)^4 \right\lvert \G,\bs,\X \right] + \E\left[\left. \left(\psi_{K,j} \right)^4 \right\lvert \G,\bs,\X \right] \right) \\
     &\leq   27\sum_{j\in \mathcal{I}_{gs}} w_j^4 \left( \E\left[\left. \left(\epsilon_{jl}^{gs} \right)^4 \right\lvert \G,\bs,\X \right] +\frac{1}{K}\sum_{m=1}^K \E\left[\left. \left(\epsilon_{k_{m(j)}l}^{gs}  \right)^4 \right\lvert \G,\bs,\X \right] + \E\left[\left. \left(\psi_{K,j} \right)^4 \right\lvert \G,\bs,\X \right] \right)\\
     &\leq 27C \sum_{j\in \mathcal{I}_{gs}} w_j^4
 \end{align*}
 The third inequality follows from Jensen's inequality.
 \begin{align*}
     &\frac{K^2}{n_{11}^2(K+1)^2}\sum_{j\in \mathcal{I}_{gs}}  Var\left(\left. w_j^2\left(\epsilon_{jl}^{gs}- \frac{1}{K}\sum_{m=1}^K\epsilon_{k_{m(j)}l}^{gs} +\psi_{K,j} \right)^2 \right\lvert \G,\bs,\X \right) \\
     &\leq \frac{27n_{gs}C K^2}{n_{11}^2(K+1)^2} \frac{1}{n_{gs}} \sum_{j\in \mathcal{I}_{gs}} w_j^4 \\
     &=o_p(1)
 \end{align*}
 For the covariance term, by Cauchy Schwarz inequality and the derivation above
 \begin{align*}
     &\abs{Cov\left(\left.  w_j^2\left(\epsilon_{jl}^{gs}- \frac{1}{K}\sum_{m=1}^K\epsilon_{k_{m(j)}l}^{gs} +\psi_{K,j} \right)^2,  w_{j'}^2\left(\epsilon_{j'l}^{gs}- \frac{1}{K}\sum_{m=1}^K\epsilon_{k_{m(j')}l}^{gs} +\psi_{K,j'} \right)^2   \right\lvert \G,\bs,\X \right) }\\
     &\leq \sqrt{Var\left(\left. w_j^2\left(\epsilon_{jl}^{gs}- \frac{1}{K}\sum_{m=1}^K\epsilon_{k_{m(j)}l}^{gs} +\psi_{K,j} \right)^2 \right\lvert \G,\bs,\X \right) } \sqrt{Var\left(\left. w_{j'}^2\left(\epsilon_{j'l}^{gs}- \frac{1}{K}\sum_{m=1}^K\epsilon_{k_{m(j')}l}^{gs} +\psi_{K,j'} \right)^2 \right\lvert \G,\bs,\X \right) } \\
     &\leq Cw_j^2w_{j'}^2
 \end{align*}

On the other hand, recall $\mathcal K_K(j)$ and $\mathcal K_K(j')$ are the within-subgroup matched sets of unit $j$, and $j'$, define $\overline{\mathcal{K}}_K(j) = \{j\} \cup K_K(j) $, and $\overline{\mathcal{K}}_K(j') = \{j'\} \cup K_K(j') $ if $ \overline{\mathcal{K}}_K(j)\cap \overline{\mathcal{K}}_K(j') =\phi $, then notice that the two terms in the covariance are independent after conditioning on $(\G,\bs,\X)$, which implies that the covariance is none zero if and only if $\overline{\mathcal{K}}_K(j)\cap \overline{\mathcal{K}}_K(j') \neq\phi  $. Therefore,
\begin{align*}
    &\abs{ \sum_{j\in \mathcal{I}_{gs}}\sum_{\substack{j'\in \mathcal{I}_{gs} \\j'\neq j}}Cov\left(\left.  w_j^2\left(\epsilon_{jl}^{gs}- \frac{1}{K}\sum_{m=1}^K\epsilon_{k_{m(j)}l}^{gs} +\psi_{K,j} \right)^2,  w_{j'}^2\left(\epsilon_{j'l}^{gs}- \frac{1}{K}\sum_{m=1}^K\epsilon_{k_{m(j')}l}^{gs} +\psi_{K,j'} \right)^2   \right\lvert \G,\bs,\X \right) }\\
    &=\left\lvert\sum_{j\in \mathcal{I}_{gs}}\sum_{\substack{j'\in \mathcal{I}_{gs} \\j'\neq j}}Cov\left(\left.  w_j^2\left(\epsilon_{jl}^{gs}- \frac{1}{K}\sum_{m=1}^K\epsilon_{k_{m(j)}l}^{gs} +\psi_{K,j} \right)^2,  w_{j'}^2\left(\epsilon_{j'l}^{gs}- \frac{1}{K}\sum_{m=1}^K\epsilon_{k_{m(j')}l}^{gs} +\psi_{K,j'} \right)^2   \right\lvert \G,\bs,\X \right) \right.\\
    &\left.\1(\overline{\mathcal{K}}_K(j)\cap \overline{\mathcal{K}}_K(j') \neq\phi )\right\lvert \\
    &\leq \sum_{j\in \mathcal{I}_{gs}}\sum_{\substack{j'\in \mathcal{I}_{gs} \\j'\neq j}}\abs{Cov\left(\left.  w_j^2\left(\epsilon_{jl}^{gs}- \frac{1}{K}\sum_{m=1}^K\epsilon_{k_{m(j)}l}^{gs} +\psi_{K,j} \right)^2,  w_{j'}^2\left(\epsilon_{j'l}^{gs}- \frac{1}{K}\sum_{m=1}^K\epsilon_{k_{m(j')}l}^{gs} +\psi_{K,j'} \right)^2   \right\lvert \G,\bs,\X \right)} \\
    &\1(\overline{\mathcal{K}}_K(j)\cap \overline{\mathcal{K}}_K(j') \neq\phi ) \\
    &\leq C\sum_{j\in \mathcal{I}_{gs}}\sum_{\substack{j'\in \mathcal{I}_{gs} \\j'\neq j}} w_j^2w_{j'}^2\1(\overline{\mathcal{K}}_K(j)\cap \overline{\mathcal{K}}_K(j') \neq\phi  ) \\
    &\leq C\sum_{j\in \mathcal{I}_{gs}}\sum_{\substack{j'\in \mathcal{I}_{gs} \\j'\neq j}} \frac{w_j^4+w_{j'}^4}{2}\1(\overline{\mathcal{K}}_K(j)\cap \overline{\mathcal{K}}_K(j') \neq\phi ) \\
    &\leq \frac{C}{2} \sum_{j\in \mathcal{I}_{gs}}\sum_{\substack{j'\in \mathcal{I}_{gs} }} w_j^4\1(\overline{\mathcal{K}}_K(j)\cap \overline{\mathcal{K}}_K(j') \neq\phi ) + \frac{C}{2} \sum_{\substack{j'\in \mathcal{I}_{gs} }}\sum_{j\in \mathcal{I}_{gs}} w_{j'}^4\1(\overline{\mathcal{K}}_K(j')\cap \overline{\mathcal{K}}_K(j) \neq\phi ) \\
    &=C \sum_{j\in \mathcal{I}_{gs}}\sum_{\substack{j'\in \mathcal{I}_{gs} }} w_j^4\1(\overline{\mathcal{K}}_K(j)\cap \overline{\mathcal{K}}_K(j') \neq\phi ) \\
    &=C \sum_{j\in \mathcal{I}_{gs}}w_j^4 \sum_{\substack{j'\in \mathcal{I}_{gs} }} \1(\overline{\mathcal{K}}_K(j)\cap \overline{\mathcal{K}}_K(j') \neq\phi) \\
\end{align*}
   Notice that $\sum_{\substack{j'\in \mathcal{I}_{gs} }} \1(\overline{\mathcal{K}}_K(j)\cap \overline{\mathcal{K}}_K(j') \neq\phi) $ is the total number of the sets that has a nonempty intersection with $\overline{\mathcal{K}}_K(i)$. This is bounded by,
   \begin{align*}
       \sum_{\substack{j'\in \mathcal{I}_{gs} }} \1(\overline{\mathcal{K}}_K(j)\cap \overline{\mathcal{K}}_K(j') \neq\phi) \leq \sum_{r\in \overline{\mathcal{K}}_K(j) } K_K(r) \leq \sum_{r\in \overline{\mathcal{K}}_K(j) } \max_{r\in \mathcal{I}_{gs}} K_K(r) = (J+1) \max_{r\in \mathcal{I}_{gs}} K_K(r)\leq (K+1)K\gamma_q
   \end{align*}
   The last line follows from stone's lemma. \\
   Therefore,
    \begin{align*}
     &\frac{K^2}{n_{11}^2(K+1)^2} \abs{\sum_{j\in \mathcal{I}_{gs}}\sum_{\substack{j'\in \mathcal{I}_{gs} \\j'\neq j}}Cov\left(\left.  w_j^2\left(\epsilon_{jl}^{gs}- \frac{1}{K}\sum_{m=1}^K\epsilon_{k_{m(j)}l}^{gs} +\psi_{K,j} \right)^2,  w_{j'}^2\left(\epsilon_{j'l}^{gs}- \frac{1}{K}\sum_{m=1}^K\epsilon_{k_{m(j')}l}^{gs} +\psi_{K,j'} \right)^2   \right\lvert \G,\bs,\X \right)} \\
     &\leq  \frac{CK^2}{n_{11}^2(K+1)^2}  \sum_{j\in \mathcal{I}_{gs}}w_j^4 \sum_{\substack{j'\in \mathcal{I}_{gs} }} \1(\overline{\mathcal{K}}_K(j)\cap \overline{\mathcal{K}}_K(j') \neq\phi) \\
     &\leq \frac{C\gamma_dK^3n_{gs}}{n_{11}^2(K+1)} \frac{1}{n_{gs}} \sum_{j\in \mathcal{I}_{gs}}w_j^4\\
     &=o_p(1)
    \end{align*}
 So,  
   \begin{align*}
       Var\left(\left.\frac{1}{n_{11}}\sum_{j\in \mathcal{I}_{gs}} w_j^2\left(\widehat{\sigma}_{gs,l}^2(X_j)-\E\left[\left. \widehat{\sigma}_{gs,l}^2(X_j)\right|\G,\bs,\X\right] \right)\right\lvert \G,\bs,\X \right) =o(1)
   \end{align*}
   By lemma \ref{lemma:conditionalchebyshev}, 
   \begin{align*}
       \frac{1}{n_{11}}\sum_{j\in \mathcal{I}_{gs}} w_j^2\left(\widehat{\sigma}_{gs,l}^2(X_j)-\E\left[\left. \widehat{\sigma}_{gs,l}^2(X_j)\right|\G,\bs,\X\right] \right) = o_p(1)
   \end{align*}
   Therefore, 
   \begin{align*}
       & \frac{1}{n_{11}}\sum_{j\in \mathcal{I}_{gs}} w_j^2\widehat{\sigma}_{gs,l}^2(X_j)-\frac{1}{n_{11}}\sum_{j\in \mathcal{I}_{gs}} w_j^2\sigma_{gs,l}^2(X_j) \\
       &=\frac{1}{n_{11}}\sum_{j\in \mathcal{I}_{gs}} w_j^2\left(\widehat{\sigma}_{gs,l}^2(X_j)-\sigma_{gs,l}^2(X_j)\right)\\
       &=\frac{1}{n_{11}}\sum_{j\in \mathcal{I}_{gs}} w_j^2\left(\widehat{\sigma}_{gs,l}^2(X_j)-\E\left[\left. \widehat{\sigma}_{gs,l}^2(X_j)\right|\G,\bs,\X\right] + \E\left[\left. \widehat{\sigma}_{gs,l}^2(X_j)\right|\G,\bs,\X\right] - \sigma_{gs,l}^2(X_j)\right)\\
       &=\frac{1}{n_{11}}\sum_{j\in \mathcal{I}_{gs}} w_j^2\left(\widehat{\sigma}_{gs,l}^2(X_j)-\E\left[\left. \widehat{\sigma}_{gs,l}^2(X_j)\right|\G,\bs,\X\right]\right)+\frac{1}{n_{11}}\sum_{j\in \mathcal{I}_{gs}} w_j^2\left(  \E\left[\left. \widehat{\sigma}_{gs,l}^2(X_j)\right|\G,\bs,\X\right] - \sigma_{gs,l}^2(X_j)\right) \\
       &=o_p(1)
   \end{align*}
   Lastly, denote $\widetilde w_j^2=\sum_{i\in \mathcal{I}_{11}}w_{ij}^2$, so
    \begin{align*}
      \widetilde w_{j}^2=\sum_{i\in \mathcal{I}_{11}}w_{ij}^2 \leq \left(\sum_{i\in \mathcal{I}_{11}}w_{ij} \right)^2=w_j^2
   \end{align*}
   This implies that, 
   \begin{align*}
       \frac{1}{n_{gs}}\sum_{j\in \mathcal{I}_{gs}} \widetilde w_j^2 = O(1), \quad \frac{1}{n_{gs}}\sum_{j\in \mathcal{I}_{gs}} \widetilde w_j^4 = O(1)
   \end{align*}
 Also notice that
   \begin{align*}
        \frac{1}{n_{11}}\sum_{i\in \mathcal{I}_{11}}\sum_{j\in\mathcal{I}_{gs}}w_{ij}^2\widehat\sigma^2_{gs,l}(X_j)&=\frac{1}{n_{11}}\sum_{j\in\mathcal{I}_{gs}} \sum_{i\in \mathcal{I}_{11}}  w_{ij}^2\widehat\sigma^2_{gs,l}(X_j) =\frac{1}{n_{11}}\sum_{j\in\mathcal{I}_{gs}}\widetilde w_j^2\widehat\sigma^2_{gs,l}(X_j) \\
        \frac{1}{n_{11}}\sum_{i\in \mathcal{I}_{11}}\sum_{j\in\mathcal{I}_{gs}}w_{ij}^2\sigma^2_{gs,l}(X_j)&=\frac{1}{n_{11}}\sum_{j\in\mathcal{I}_{gs}} \sum_{i\in \mathcal{I}_{11}}  w_{ij}^2\sigma^2_{gs,l}(X_j) =\frac{1}{n_{11}}\sum_{j\in\mathcal{I}_{gs}}\widetilde w_j^2\sigma^2_{gs,l}(X_j)
        \end{align*}
   Then following the same method as above, we could show that,
   \begin{align*}
       \frac{1}{n_{11}}\sum_{i\in \mathcal{I}_{11}}\sum_{j\in\mathcal{I}_{gs}}w_{ij}^2\widehat\sigma^2_{gs,l}(X_j) -\frac{1}{n_{11}}\sum_{i\in \mathcal{I}_{11}}\sum_{j\in\mathcal{I}_{gs}}w_{ij}^2\sigma^2_{gs,l}(X_j) =o_p(1)
   \end{align*}
   This implies that,
   \begin{align*}
        \widehat{V}_{E}&=\frac{1}{n_{11}}\sum_{j\in  \mathcal{I}_{01}} \left(  w_j^2-\widetilde w_j^2 \right)\widehat{\sigma}_{01,l}^2(X_j) + \frac{1}{n_{11}}\sum_{j\in  \mathcal{I}_{10}} \left(  w_j^2-\widetilde w_j^2 \right)\widehat{\sigma}_{10,l}^2(X_j) + \frac{1}{n_{11}}\sum_{j\in  \mathcal{I}_{00}} \left(  w_j^2-\widetilde w_j^2 \right)\widehat{\sigma}_{00,l}^2(X_j) \\
        &=\frac{1}{n_{11}}\sum_{j\in  \mathcal{I}_{01}} \left(  w_j^2-\widetilde w_j^2 \right)\sigma_{01,l}^2(X_j) + \frac{1}{n_{11}}\sum_{j\in  \mathcal{I}_{10}} \left(  w_j^2-\widetilde w_j^2 \right)\sigma_{10,l}^2(X_j) + \frac{1}{n_{11}}\sum_{j\in  \mathcal{I}_{00}} \left(  w_j^2-\widetilde w_j^2 \right)\sigma_{00,l}^2(X_j) +o_p(1)
   \end{align*}
   Combining all the results above,
   \begin{align*}
      &\widehat{V}_{\tau(X)}+ \widehat{V}_{E}\\
      &=\frac{1}{n_{11}}\sum_{i\in \mathcal{I}_{11}}\left(\Delta Y_{il}-  \widehat {\Delta Y_{il}^{01}}-\widehat {\Delta Y_{il}^{10}} + \widehat {\Delta Y_{il}^{00}}-\tau_l \right)^2 +o_p(1) \\
    &=V_{\tau(X)}+\frac{1}{n_{11}}\sum_{i\in \mathcal{I}_{11}} \sigma^2_{11,l}(X_i) +\frac{1}{n_{11}} \sum_{j\in \mathcal{I}_{01}}\widetilde w_{j}^2\sigma_{01,l}^2(X_j) +\frac{1}{n_{11}}\sum_{j\in \mathcal{I}_{10}}\widetilde w_{j}^2\sigma_{10,l}^2(X_j) +\frac{1}{n_{11}} \sum_{j\in \mathcal{I}_{00}}\widetilde w_{j}^2\sigma_{00,l}^2(X_j) \\
    &+\frac{1}{n_{11}}\sum_{j\in  \mathcal{I}_{01}} \left(  w_j^2-\widetilde w_j^2 \right)\sigma_{01,l}^2(X_j) + \frac{1}{n_{11}}\sum_{j\in  \mathcal{I}_{10}} \left(  w_j^2-\widetilde w_j^2 \right)\sigma_{10,l}^2(X_j) + \frac{1}{n_{11}}\sum_{j\in  \mathcal{I}_{00}} \left(  w_j^2-\widetilde w_j^2 \right)\sigma_{00,l}^2(X_j) +o_p(1)\\
    &=V_{\tau(X)}+V_{E}+o_p(1)
   \end{align*}
   This completes the proof.
\end{proof}

\begin{proof}
    Proof of theorem \ref{theorem:biascorrection}.\\
    For any $(g,s)\in \{(0,1),(1,0),(0,0) \}$,
    \begin{align*}
        &\abs{\frac{1}{n_{11}}\isum\left[\1(i\in \mathcal{I}_{11})-\1(i\in \mathcal{I}_{gs})w_i\right]\left(\mu_l^{gs}(X_i)-\widehat\mu_l^{gs}(X_i) \right)}  \\
        &=\abs{\frac{1}{n_{11}}\sum_{i\in \mathcal{I}_{11}}\sum_{j\in \mathcal{I}_{gs}}w_{ij}\left[\mu_{l}^{gs}(X_i)-\mu_{l}^{gs}(X_j)-\left(\widehat\mu_{l}^{gs}(X_i)-\widehat\mu_{l}^{gs}(X_j) \right) \right]}\\
        &\leq \frac{1}{n_{11}}\sum_{i\in \mathcal{I}_{11}}\sum_{j\in \mathcal{I}_{gs}}w_{ij}\abs{\mu_{l}^{gs}(X_i)-\mu_{l}^{gs}(X_j)-\left(\widehat\mu_{l}^{gs}(X_i)-\widehat\mu_{l}^{gs}(X_j) \right) }
    \end{align*}
    The first equality follows from the same logic as the proof of lemma \ref{lemma:equivalentestimator}.\\
    By Taylor expansion to $k-1$-th order,
    \begin{align*}
        \abs{\mu_{l}^{gs}(X_j)-\mu_{l}^{gs}(X_i)-\sum_{r=1}^{p-1}\frac{1}{r!}\sum_{m\in \Lambda_r}\partial^m\mu^{gs}_l(X_i)(X_j-X_i)^m}&\leq \max_{m\in\Lambda_p}\norm{\partial^m\mu^{gs}_l(x)}_{\infty}\frac{1}{p!}\sum_{m\in \Lambda_p}\norm{X_j-X_i}^p \\
        &= O(1)\norm{X_j-X_i}^p
    \end{align*}
    Similarly,
    \begin{align*}
        \abs{\widehat\mu_{l}^{gs}(X_j)-\widehat\mu_{l}^{gs}(X_i)-\sum_{r=1}^{p-1}\frac{1}{r!}\sum_{m\in \Lambda_r}\partial^m\widehat\mu^{gs}_l(X_i)(X_j-X_i)^m}&\leq \max_{m\in\Lambda_p}\norm{\partial^m\widehat\mu^{gs}_l(x)}_{\infty}\frac{1}{p!}\sum_{m\in \Lambda_p}\norm{X_j-X_i}^p \\
        &= O_p(1)\norm{X_j-X_i}^p
    \end{align*}
    Lastly,
    \begin{align*}
     &\abs{   \sum_{r=1}^{p-1}\frac{1}{r!}\sum_{m\in \Lambda_r}\left(\partial^m\widehat\mu^{gs}_l(X_i)-\partial^m\mu^{gs}_l(X_i) \right)(X_j-X_i)^m} \\
     &\leq \sum_{r=1}^{p-1} \max_{t\in \Lambda_r}\norm{\partial^{m}\widehat{\mu}^{gs}(x)-\partial^{m}\mu^{gs}(x)}_{\infty}\frac{1}{r!}\sum_{m\in \Lambda_r}\norm{X_j-X_i}^r \\
     &\leq \sum_{r=1}^{p-1} O_p(-\gamma_r)\norm{X_j-X_i}^r 
    \end{align*}
    Now,
    \begin{align*}
        &\frac{1}{n_{11}}\sum_{i\in \mathcal{I}_{11}}\sum_{j\in \mathcal{I}_{gs}}w_{ij}\abs{\mu_{l}^{gs}(X_i)-\mu_{l}^{gs}(X_j)-\left(\widehat\mu_{l}^{gs}(X_i)-\widehat\mu_{l}^{gs}(X_j) \right) }\\
        &\leq \frac{1}{n_{11}}\sum_{i\in \mathcal{I}_{11}}\sum_{j\in \mathcal{I}_{gs}}w_{ij}\abs{ \mu_{l}^{gs}(X_j)-\mu_{l}^{gs}(X_i)-\sum_{r=1}^{p-1}\frac{1}{r!}\sum_{m\in \Lambda_r}\partial^m\mu^{gs}_l(X_i)(X_j-X_i)^m }\\
        &+ \frac{1}{n_{11}}\sum_{i\in \mathcal{I}_{11}}\sum_{j\in \mathcal{I}_{gs}}w_{ij}\abs{\widehat\mu_{l}^{gs}(X_j)-\widehat\mu_{l}^{gs}(X_i)-\sum_{r=1}^{p-1}\frac{1}{r!}\sum_{m\in \Lambda_r}\partial^m\widehat\mu^{gs}_l(X_i)(X_j-X_i)^m } \\
        &+\frac{1}{n_{11}}\sum_{i\in \mathcal{I}_{11}}\sum_{j\in \mathcal{I}_{gs}}w_{ij} \abs{   \sum_{r=1}^{p-1}\frac{1}{r!}\sum_{m\in \Lambda_r}\left(\partial^m\widehat\mu^{gs}_l(X_i)-\partial^m\mu^{gs}_l(X_i) \right)(X_j-X_i)^m} \\
        &\leq O_p(1)\frac{1}{n_{11}}\sum_{i\in \mathcal{I}_{11}}\sum_{j\in \mathcal{I}_{gs}}w_{ij} \norm{X_j-X_i}^p +O_p(1) \frac{1}{n_{11}}\sum_{i\in \mathcal{I}_{11}}\sum_{j\in \mathcal{I}_{gs}}w_{ij} \norm{X_j-X_i}^p \\
        &+\sum_{r=1}^{p-1} O_p(-\gamma_r) \left( \frac{1}{n_{11}}\sum_{i\in \mathcal{I}_{11}}\sum_{j\in \mathcal{I}_{gs}}w_{ij}\norm{X_j-X_i}^r \right)\\
        &=o_p(n^{-1/2})+o_p(n^{-1/2})+\sum_{r=1}^{p-1} O_p(-\gamma_r)o_p(n^{-1/2}+\gamma_r)\\
        &=o_p(n^{-1/2})
    \end{align*}
    The first equality follows from assumption \ref{assumption:biascorrection} and the similar argument as the proof of lemma \ref{lemma:weightproperty}(iv).\\
    Therefore,
    \begin{align*}
        \abs{B-\widehat{B}}&\leq \abs{\frac{1}{n_{11}}\isum\left[\1(i\in \mathcal{I}_{11})-\1(i\in \mathcal{I}_{gs})w_i\right]\left(\mu_l^{01}(X_i)-\widehat\mu_l^{01}(X_i) \right)} \\
        &+\abs{\frac{1}{n_{11}}\isum\left[\1(i\in \mathcal{I}_{11})-\1(i\in \mathcal{I}_{gs})w_i\right]\left(\mu_l^{10}(X_i)-\widehat\mu_l^{10}(X_i) \right)} \\
        &+\abs{\frac{1}{n_{11}}\isum\left[\1(i\in \mathcal{I}_{11})-\1(i\in \mathcal{I}_{gs})w_i\right]\left(\mu_l^{00}(X_i)-\widehat\mu_l^{00}(X_i) \right)} \\
        &=o_p(n^{-1/2})
    \end{align*}
    By Slutsky theorem, this shows that $\sqrt{n_{11}}(B-\widehat{B})\xrightarrow{p} 0 $.\\
    Lastly,
    \begin{align*}
         &(V_{\tau(X)}+V_{E})^{-1/2} \sqrt{n_{11}}\left(\widehat{\tau}_l-\tau_l-\widehat B \right)\\
        & = (V_{\tau(X)}+V_{E})^{-1/2} \sqrt{n_{11}}\left(\widehat{\tau}_l-\tau_l- B \right) + (V_{\tau(X)}+V_{E})^{-1/2} \sqrt{n_{11}}(B-\widehat{B}) \\
        &\xrightarrow{d}\mathcal{N}(0,1)
    \end{align*}
    The convergence result follows by the argument above and theorem \ref{theorem:asymptotic}. This completes the proof.
\end{proof}

\section{Proof of Lemmas}\label{Appendix:lemmaproofs}
\begin{proof}
    Proof of lemma \ref{lemma:covariatetrend}. \\
    To prove (I), notice that the covariate trend 3WFE specification can be written as:
        \begin{align*}
    &Y_{it}
    = \sum_{g,s} \alpha_{gs}\1(G_i=g,S_i=s)
    + \sum_{g} \alpha_{gt}\1(G_i=g)
    + \sum_{s} \gamma_{st}\1(S_i=s)
    + \sum_{l\neq -1}
        \tau_l^{CT} D_i\1(t-t^*=l)
    +   X_i'\theta_t
    + u_{it}\\
 & \implies \Delta Y_{il}=(\alpha_{0,t^*+l}-\alpha_{0,t^*-1} + \gamma_{0,t^*+l}-\gamma_{0,t^*-1}) +(\alpha_{0,t^*+l}-\alpha_{0,t^*-1})G_i +(\gamma_{0,t^*+l}-\gamma_{0,t^*-1})S_i+\tau_l^{CT} D_i\\
 &+X_i(\theta_{t^*+l}-\theta_{t^*-1})+(u_{i,t^*+l}-u_{i,t^*-1})
\end{align*}
   So $\widehat \tau_l^{CT}$ is the coefficient of $D$ of the population linear projection of $\Delta Y_l$ on $1,G,S,D,X$. Apply the FWL theorem, we have:
   \begin{align*}
       \tau_l^{CT}&=\frac{\E[(D-L( D|G,S,X))\Delta Y_l]}{\E[(D-L( D|G,S,X))^2]  }\\
       &=\E\left[\left.\frac{p_{11}\left(1-L(D|G,S,X)\right)}{\E[(D-L(D|G,S,X))^2]} \Delta Y_l \right| G=1,S=1\right]-\E\left[\left.\frac{p_{01}L(D|G,S,X)}{\E[(D-L(D|G,S,X))^2]} \Delta Y_l \right| G=0,S=1\right] \\
         &-\E\left[\left.\frac{p_{10}L(D|G,S,X)}{\E[(D-L(D|G,S,X))^2]} \Delta Y_l \right| G=1,S=0\right] +\E\left[\left.\frac{-p_{00}L(D|G,S,X)}{\E[(D-L(D|G,S,X))^2]} \Delta Y_l \right| G=0,S=0\right]
   \end{align*}
   The second equality follows from law of iterated expectation.\\
   To prove (II), we need to prove two cases. We first consider the case where the conditional mean is linear with a common slope. 
   From the proof above,
\begin{align*}
    \tau_l^{CT}&=\frac{\E[(D-L( D|G,S,X))\Delta Y_l]}{\E[(D-L( D|G,S,X))^2]  }
\end{align*}
First consider the numerator,
\begin{align*}
    &\E[(D-L( D|G,S,X))\Delta Y_l]\\
    &=\E[(D-L( D|G,S,X))(\Delta Y_l-L_0(\Delta Y_l|G,S,X))]  \\
    &= \E[D(D-L( D|G,S,X))(\Delta Y_l-L_0(\Delta Y_l|G,S,X))]\\
    &+ \E[(1-D)(D-L( D|G,S,X))(\Delta Y_l-L_0(\Delta Y_l|G,S,X))] \\
    &=\E[D(D-L( D|G,S,X))(\Delta Y_l-L_0(\Delta Y_l|G,S,X))]\\
    &-\E[L( D|G,S,X))(\Delta Y_l-L_0(\Delta Y_l|G,S,X))|D=0](1-p_{11})\\
    &=\E[D(D-L( D|G,S,X))(\Delta Y_l-L_0(\Delta Y_l|G,S,X))] \\
    &=\E[(1-L( D|1,1,X))(\Delta Y_l-L_0(\Delta Y_l|1,1,X))|D=1]p_{11}
\end{align*}
The first equality follows from the fact that $\E[(D-L( D|G,S,X))L_0(\Delta Y_l|G,S,X)]=0$. The fourth equality follows from the fact that $\E[L( D|G,S,X))(\Delta Y_l-L_0(\Delta Y_l|G,S,X)|D=0]=0$.\\
Define $\widetilde w_{11}(X)=\frac{p_{11}(1-L(1,1,X))}{\kappa}$, $\tau_l(X)=\E[Y_{t^*+l}(1)-Y_{t^*+l}(0)|D=1,X]$, then
\begin{align*}
    \tau_l^{CT}&=\E[\widetilde w_{11}(X)(\Delta Y_l-L_0(\Delta Y_l|1,1,X))|D=1 ]\\
    &=\E[\widetilde w_{11}(X)(\E[\Delta Y_l|D=1,X]-L_0(\Delta Y_l|1,1,X))|D=1 ]\\
    &=\E[\widetilde w_{11}(X)(\tau_l(X)+\E[\Delta Y_l(0)|D=1,X]-L_0(\Delta Y_l|1,1,X))|D=1 ]\\
    &=\E[\widetilde w_{11}(X)\tau_l(X)|D=1]-\E[\widetilde w_{11}(X)L_0(\Delta Y_l|1,1,X))|D=1] \\
    &+\E[\widetilde w_{11}(X)(\E[\Delta Y_l(0)|D=1,X]-\E[\Delta Y_l(0)|G=0,S=1,X]\\
    &-\E[\Delta Y_l(0)|G=1,S=0,X]+\E[\Delta Y_l(0)|G=1,S=0,X]|D=1] \\
    &+\E[\widetilde w_{11}(X)(\E[\Delta Y_l(0)|G=0,S=1,X]+\E[\Delta Y_l(0)|G=1,S=0,X]-\E[\Delta Y_l(0)|G=1,S=0,X]|D=1] \\
    &=\E[\widetilde w_{11}(X)\tau_l(X)|D=1] \\
    &+\E[\widetilde w_{11}(X)(\E[\Delta Y_l(0)|G=0,S=1,X]-L_0(\Delta Y|0,1,X))]\\
    &+\E[\widetilde w_{11}(X)(\E[\Delta Y_l(0)|G=1,S=0,X]-L_0(\Delta Y|1,0,X))]\\
    &-\E[\widetilde w_{11}(X)(\E[\Delta Y_l(0)|G=0,S=0,X]-L_0(\Delta Y|0,0,X))]
\end{align*}
The last equality follows from the conditional parallel gaps and the fact that $L_0(\Delta Y|1,1,X)=L_0(\Delta Y|0,1,X)+L_0(\Delta Y|1,0,X)-L_0(\Delta Y|0,0,X) $.\\
Under the assumption (i), the treatment effect is homogeneous with respect to $X$ and $\E[\Delta Y_l|G=g,S=s,X]=L_0(\Delta Y_l|g,s,X)$, it follows directly from the last equation above that $\tau_l^{CT}=\tau_l$.\\
We next consider the case where the generalized propensity scores takes the inverse linear form.
\begin{align*}
    &\E\!\left[\frac{p_{11}\left(1-L(D\mid G,S,X)\right)}{\kappa}\, \Delta Y_l \;\middle|\; G=1,S=1\right]\\
    &=\E\!\left[\frac{p_{11}\left(1-L(D\mid G,S,X)\right)}{\kappa}\, \E[\Delta Y_l|G=1,S=1,X] \;\middle|\; G=1,S=1\right] \\
    &=\E\!\left[\frac{p_{11}\left(1-L(D\mid G,S,X)\right)}{\kappa}\, \left(\tau_l(X) +\E[\Delta Y_l(0)|G=1,S=1,X] \right) \;\middle|\; G=1,S=1\right] \\
    &=\E\!\left[\frac{p_{11}\left(1-L(D\mid G,S,X)\right)}{\kappa}\, \tau_l(X) \;\middle|\; G=1,S=1\right] \\
    &+\E\!\left[\frac{p_{11}\left(1-L(D\mid G,S,X)\right)}{\kappa}\, \E[\Delta Y_l(0)|G=0,S=1,X]  \;\middle|\; G=1,S=1\right] \\
    &+\E\!\left[\frac{p_{11}\left(1-L(D\mid G,S,X)\right)}{\kappa}\, \E[\Delta Y_l(0)|G=1,S=0,X]  \;\middle|\; G=1,S=1\right] \\
    &-\E\!\left[\frac{p_{11}\left(1-L(D\mid G,S,X)\right)}{\kappa}\, \E[\Delta Y_l(0)|G=0,S=0,X]  \;\middle|\; G=1,S=1\right] \allowdisplaybreaks \\
    &=\E\!\left[\frac{p_{11}\left(1-L(D\mid 1,1,X)\right)}{\kappa}\, \tau_l(X) \;\middle|\; G=1,S=1\right] \\
    &+\frac{1}{\kappa}\E\!\left[e_{11}(X)\left(1-L(D\mid 1,1,X)\right)\, \E[\Delta Y_l(0)|G=0,S=1,X]  \right] \\
    &+\frac{1}{\kappa}\E\!\left[e_{11}(X)\left(1-L(D\mid 1,1,X)\right)\, \E[\Delta Y_l(0)|G=1,S=0,X]  \right] \\
    &-\frac{1}{\kappa}\E\!\left[e_{11}(X)\left(1-L(D\mid 1,1,X)\right)\, \E[\Delta Y_l(0)|G=0,S=0,X]  \right] 
\end{align*}
    The third equality follows from the conditional parallel gaps assumption. The last equality follows from the argument below,
    \begin{align*}
        &\E\!\left[\frac{p_{11}\left(1-L(D\mid G,S,X)\right)}{\kappa}\, \E[\Delta Y_l(0)|G=0,S=1,X]  \;\middle|\; G=1,S=1\right]  \\
        &=\frac{1}{\kappa}\E\!\left[\1(G=1,S=1)\left(1-L(D\mid 1,1,X)\right)\, \E[\Delta Y_l(0)|G=g,S=s,X]  \right] \\
        &=\frac{1}{\kappa}\E\!\left[e_{11}(X)\left(1-L(D\mid 1,1,X)\right)\, \E[\Delta Y_l(0)|G=g,S=s,X]  \right]
    \end{align*}
    On the other hand, 
    \begin{align*}
        &\E\!\left[\frac{p_{gs}\,L(D\mid G,S,X)}{\kappa}\, \Delta Y_l \;\middle|\; G=g,S=s\right]\\
        &=\E\!\left[\frac{p_{gs}\,L(D\mid g,s,X)}{\kappa}\, \E[ \Delta Y_l|G=g,S=s,X] \;\middle|\; G=g,S=s\right] \\
        &=\frac{1}{\kappa} \E\!\left[\1(G=g,S=s)L(D\mid g,s,X)\, \E[ \Delta Y_l|G=g,S=s,X] \right]\\
        &=\frac{1}{\kappa} \E\!\left[e_{gs}(X)L(D\mid g,s,X)\, \E[ \Delta Y_l|G=g,S=s,X] \right]
    \end{align*}
    Therefore,
    \begin{align*}
        \tau_l^{CT} &=\E\!\left[\frac{p_{11}\left(1-L(D\mid 1,1,X)\right)}{\kappa}\, \tau_l(X) \;\middle|\; G=1,S=1\right]\\
        &+\frac{1}{\kappa}\E\!\left[\left[e_{11}(X)\left(1-L(D\mid 1,1,X)\right)-e_{01}(X)L(D\mid 0,1,X)\right]\, \E[\Delta Y_l(0)|G=0,S=1,X]  \right] \\
        &+\frac{1}{\kappa}\E\!\left[\left[e_{11}(X)\left(1-L(D\mid 1,1,X)\right)-e_{10}(X)L(D\mid 1,0,X)\right]\, \E[\Delta Y_l(0)|G=1,S=0,X]  \right] \\
        &-\frac{1}{\kappa}\E\!\left[\left[e_{11}(X)\left(1-L(D\mid 1,1,X)\right)+e_{00}(X)L(D\mid 0,0,X)\right]\, \E[\Delta Y_l(0)|G=0,S=0,X]  \right]
    \end{align*}
    Under the assumption (ii), we have:
    \begin{align*}
        &\E\!\left[\frac{p_{11}\left(1-L(D\mid 1,1,X)\right)}{\kappa}\, \tau_l(X) \;\middle|\; G=1,S=1\right] =  \E\!\left[\frac{p_{11}\left(1-L(D\mid 1,1,X)\right)}{\kappa}\, \tau_l \;\middle|\; G=1,S=1\right] =\tau_l \\
    \end{align*}
    In addition, for $\{g,s\}\in\{(0,1),(1,0),(0,0)\}$
    \begin{align*}
    &e_{11}(X)\left(1-L(D\mid 1,1,X)\right)+(-1)^{g+s}e_{gs}(X)L(D\mid g,s,X) \\
     &=\frac{[1-L(D|1,1,X)]^{-1}}{\sum_{g=0}^1\sum_{s=0}^1[gs-(-1)^{g+s}L(D|g,s,x)]^{-1}}\left(1-L(D\mid 1,1,X)\right) \\
     &-\frac{[(-1)^{g+s}L(D\mid g,s,X)]^{-1}}{\sum_{g=0}^1\sum_{s=0}^1[gs-(-1)^{g+s}L(D|g,s,x)]^{-1}}(-1)^{g+s}L(D\mid g,s,X)\\
    &=0
    \end{align*}
    This implies that $\tau_l^{CT}=\tau_l$.
\end{proof}

\begin{proof}
    Proof of lemma \ref{lemma:covariatelevel}. \\
    To prove (I), notice that the covariate-level 3WFE specification can be written as:
        \begin{align*}
    &Y_{it}
    = \sum_{g,s} \alpha_{gs}\1(G_i=g,S_i=s)
    + \sum_{g} \alpha_{gt}\1(G_i=g)
    + \sum_{s} \gamma_{st}\1(S_i=s)
    + \sum_{l\neq -1}
        \tau_l^{CL} D_i\1(t-t^*=l)
    +  X_i'\theta
    + u_{it}\\
 & \implies \Delta Y_{il}=(\alpha_{0,t^*+l-1}-\alpha_{0,t^*} + \gamma_{0,t^*+l-1}-\gamma_{0,t^*}) +(\alpha_{0,t^*+l-1}-\alpha_{0,t^*})G_i +(\gamma_{0,t^*+l-1}-\gamma_{0,t^*})S_i+\tau_l D_i\\
 &+(u_{i,t^*+l-1}-u_{i,t^*})
\end{align*}
   So  $\tau_l^{CL}$ is the coefficient of $D$ of the population linear projection of  $\Delta Y_l$ on $1,G,S,D$, which is a fully saturated regression. Therefore, it is straight forward to solve that,
   \begin{align*}
        \tau_l^{CL} &=\E[\Delta Y_l | G=1,S=1] -\E[\Delta Y_l | G=0,S=1]- \E[\Delta Y_l | G=1,S=0] + \E[\Delta Y_l | G=0,S=0] 
   \end{align*}
   To prove (II) 
   \begin{align*}
       \tau_l^{CL} &=\E[\Delta Y_l | G=1,S=1] -\E[\Delta Y_l | G=0,S=1]- \E[\Delta Y_l | G=1,S=0] + \E[\Delta Y_l | G=0,S=0] \\
       &=\tau_l + \E[\Delta Y_l(0) | G=1,S=1] -\E[\Delta Y_l(0) | G=0,S=1] \\
       &- \E[\Delta Y_l(0) | G=1,S=0] + \E[\Delta Y_l(0) | G=0,S=0]
   \end{align*}
   Then it follows that $ \tau_l^{CL}= \tau_l $ if and only if the parallel gaps holds unconditionally, in other words,
    \begin{align*}
        &\E[\Delta Y_l(0)\mid G=1,\ S=1] - \E[\Delta Y_l(0)\mid G=0,\ S=1]\\
     &=   \E[\Delta Y_l(0)\mid G=1,\ S=0] - \E[\Delta Y_l(0)\mid G=0,\ S=0]
    \end{align*}
\end{proof}

\begin{proof}
    Proof of lemma \ref{lemma:matchingprocedure}.\\
    Proof of part (a): By construction of the nearest neighbor matching algorithm, $w_{ij}^{fix}$ and $w_j^{fix}$ satisfy assumption \ref{assumption:matching} (i). 
    To show (ii), conditional on the sample size $n_{g's'},$ notice that under nearest neighbor matching, 
    \begin{align*}
    &\mathbb E\!\left[\sum_{j\in\mathcal I_{gs}} w_{ij}\left\lVert X_j-X_i\right\rVert\,\middle|\, i\in\mathcal I_{11}\right]\\
    &=\mathbb E\!\left[\sum_{m=1}^M \frac{1}{M}\left\lVert X_i-X_{j_m(i,g,s)}\right\rVert\,\middle|\, i\in\mathcal I_{11}\right] \\
    &\leq \mathbb E\!\left[\left\lVert X_i-X_{j_M(i,g,s)}\right\rVert\,\middle|\, i\in\mathcal I_{11}\right] \\
    &=n_{gs}^{-1/q}\mathbb E\!\left[n_{gs}^{1/q}\left\lVert X_i-X_{j_M(i,g,s)}\right\rVert\,\middle|\, i\in\mathcal I_{11}\right] \\
    &\leq n_{gs}^{-1/q}C\\
    &=o(1)
    \end{align*}
    The last inequality follows lemma 2 in \cite{Abadie-Imbens_2006_ECMA}.\\
    To show (iii), we can also conduct the other pairwise matching by reversing the role of reference group and the comparison group in the original pairwise matching. Then each unit $j\in \mathcal{I}_{11}\cup \mathcal{I}_{gs}$ gets a matching weight $w_j$, then under our assumption, we can invoke lemma 3(i) in \cite{Abadie-Imbens_2006_ECMA} to get that $\E[w_j^{4+\delta}|j\in \mathcal{I}_{11}\cup \mathcal{I}_{gs}]=O(1)$. Therefore,
    \begin{align*}
\E[w_j^{4+\delta}|j\in \mathcal{I}_{gs}]&=\E[w_j^{4+\delta}\1(j\in \mathcal{I}_{gs} )|i\in \mathcal{I}_{11}\cup \mathcal{I}_{gs}]\frac{p_{11}+p_{gs}}{p_{gs}} \\
&\leq \E[w_j^{4+\delta}|i\in \mathcal{I}_{11}\cup \mathcal{I}_{gs}]\frac{p_{11}+p_{gs}}{p_{gs}} \\
 &=O(1)     
    \end{align*}
    (iv) is shown in lemma \ref{lemma:reallocationbound}.
    
    Proof of part (b): By construction of the nearest neighbor matching algorithm, $w_{ij}^{diverge}$ and $w_j^{diverge}$ satisfy assumption \ref{assumption:matching} (i). \\
    To show (ii),  conditional on the sample size $n_{gs}$, we have:
     \begin{align*}
    &\mathbb E\!\left[\sum_{j\in\mathcal I_{gs}} w_{ij}\left\lVert X_j-X_i\right\rVert\,\middle|\, i\in\mathcal I_{11}\right]\\
    &=\mathbb E\!\left[\sum_{m=1}^M \frac{1}{M}\left\lVert X_i-X_{j_m(i,g,s)}\right\rVert\,\middle|\, i\in\mathcal I_{11}\right] \\
    &\leq \mathbb E\!\left[\left\lVert X_i-X_{j_M(i,g,s)}\right\rVert\,\middle|\, i\in\mathcal I_{11}\right] \\
    &\leq C\left(\frac{M}{n_{gs}} \right)^{1/d} \\
    &=o(1)
    \end{align*}
    The second inequality follows from lemma $C.2$ in \cite{Lin-Ding-Han_2023_ECMA}. The last line follows from assumption \ref{assumption:NNMdiverge}.\\
    To show (iii), we still conduct the other pairwise matching by reversing the role of reference group and the comparison group in the original pairwise matching. Then each unit $j\in \mathcal{I}_{11}\cup \mathcal{I}_{gs}$ gets a matching weight $w_j$. Conditional on the sample size $n_{11}$, and $n_{gs}$, by theorem $B.2$ in \cite{Lin-Ding-Han_2023_ECMA}, we have:
    \begin{align*}
        \E\left[\abs{\frac{n_{gs}}{n_{11}}w_j-\frac{f_{11}(X)}{f_{gs}(X)} }^{4+\delta} \right] \rightarrow 0
    \end{align*}
    Use the C-r inequality, we have:
    \begin{align*}
        \E\left[\left(\frac{n_{gs}}{n_{11}}w_j\right)^{4+\delta}  \right]&\leq 2^{3+\delta} \E\left[\abs{\frac{n_{gs}}{n_{11}}w_j-\frac{f_{11}(X)}{f_{gs}(X)} }^{4+\delta} \right] + \E\left[\left(\frac{f_{11}(X)}{f_{gs}(X)}\right)^{4+\delta} \right] \\
        &\leq o(1) + C \\
        &=O(1)
    \end{align*}
    The inequality follows from assumption \ref{assumption:regularity} (i). Therefore, 
    \begin{align*}
        \E[w_j ^{4+\delta}] =\left( \frac{n_{11}}{n_{gs}} \right)^{4+\delta}\E\left[\left(\frac{n_{gs}}{n_{11}}w_j\right)^{4+\delta}  \right]=O(1)
    \end{align*}
    Then by the similar argument as the proof in part (a), we can show that
    \begin{align*}
        \E[w_j ^{4+\delta}|j\in \mathcal{I}_{gs}]=O(1)
    \end{align*}
    (iv) is shown in lemma \ref{lemma:reallocationbound}.
    
    Proof of part (c): By construction of the nearest neighbor matching algorithm, $w_{ij}^{kernel}$ and $w_j^{kernel}$ satisfy assumption \ref{assumption:matching} (i). \\
    To prove (ii), notice that under assumption \ref{assumption:kernel} (i), if $\norm{\frac{X_j-X_i}{h}}>R$
    \begin{align*}
        w_{ij}^{kernel}\norm{X_j-X_i}=\frac{K\!\left(\frac{X_j - X_i}{h}\right)}
       {\sum_{l\in\mathcal I_{g's'}} K\!\left(\frac{X_l - X_i}{h}\right)}\norm{X_j-X_i}=0
    \end{align*}
    if $\norm{\frac{X_j-X_i}{h}}<R$
    \begin{align*}
        w_{ij}^{kernel}\norm{X_j-X_i}\leq w_{ij}^{kernel}Rh
    \end{align*}
    Therefore, 
    \begin{align*}
        \mathbb E\!\left[\sum_{j\in\mathcal I_{g's'}} w_{ij}\left\lVert X_j-X_i\right\rVert\,\middle|\, i\in\mathcal I_{gs}\right] &\leq Rh\cdot \mathbb E\!\left[\sum_{j\in\mathcal I_{g's'}} w_{ij}\middle|\, i\in\mathcal I_{gs}\right] = Rh\rightarrow 0
    \end{align*}
    The equality follows by the fact the pairwise weight $w_{ij}^{kernel}$ satisfies assumption \ref{assumption:matching}(i) and the convergence result follows by assumption \ref{assumption:kernel}.\\
    To prove (iii), it suffices to show that  
    \[\E\!\left[w_j^5\mid j\in\Ical_{g's'}\right]=O(1)\]
    Conditional on the subgroup sample size $n_{gs}$, $n_{g's'}$, and $j\in \mathcal{I}_{g's'}$, by definition,
    \begin{align*}
        w_j^5&=\left(\sum_{i\in \mathcal{I}_{gs}} w_{ij} \right)^5\\
        &=\sum_{i_1\in \mathcal{I}_{gs}}\cdots\sum_{i_5\in \mathcal{I}_{gs}}  w_{i_1j}\cdots w_{i_5j}
    \end{align*}
    Consider a tuple $(i_1,\ldots,i_5)$ containing exactly $k$ distinct indices, with multiplicities $a_1,\ldots,a_k$ satisfying
\[
    a_1+\cdots+a_k=5.
\]
Conditional on all comparison group covariates,  and $j\in \mathcal{I}_{g's'}$, $\{w_{ij}\}_{i\in \mathcal{I}_{gs}}$, are i.i.d random variables, so 
\begin{align*}
    &\E\!\left[w_{i_1j}\cdots w_{i_5j}
        \mid j\in \mathcal{I}_{g's'},  \{X_\ell:\ell\in\Ical_{g's'}\}
    \right]\\
    &=
    \prod_{h=1}^{k}
    \E\!\left[
        w_{ij}^{a_h}
        \mid j\in \mathcal{I}_{g's'},  \{X_\ell:\ell\in\Ical_{g's'}\}
    \right]\\
    &\leq \left(\E\!\left[
        w_{ij}
        \mid j\in \mathcal{I}_{g's'}, \{X_\ell:\ell\in\Ical_{g's'}\}
    \right]\right)^k
\end{align*}
The inequality follows from $0\le w_{ij}\le 1$. On the other hand, the number of order five tuples containing exact $k$ distinct indices is $n_{gs}(n_{gs}-1)\cdots (n_{gs}-k+1)C(5,k)$ which is bounded by $C(5,k)n_{gs}^k$, where $C(5,k)$ is the Stirling number of the second kind (a constant). It follows that
\begin{align*}
    \E\!\left[w_j^5\mid j\in \mathcal{I}_{g's'},  \{X_\ell:\ell\in\Ical_{g's'}\}\right]\leq \sum_{k=1}^5 C(5,k)n_{gs}^k\left(\E\!\left[
        w_{ij}
        \mid j\in \mathcal{I}_{g's'},  \{X_\ell:\ell\in\Ical_{g's'}\}
    \right]\right)^k
\end{align*}
By law of iterated expectation
\begin{align*}
    \E\!\left[w_j^5\mid   j\in \mathcal{I}_{g's'}\right]\leq \sum_{k=1}^5 C(5,k)n_{gs}^k\E\left[\left(\E\!\left[
        w_{ij}
        \mid \{ X_j:j\in\Ical_{g's'} \}
    \right]\right)^k\mid j\in \mathcal{I}_{g's'}  \right]
\end{align*}
Now the key is to put an appropriate bound on $\E\left[\left(\E\!\left[w_{ij}\mid \{X_j:j\in\Ical_{g's'}\}\right]\right)^k\mid j\in \mathcal{I}_{g's'} \right]$

Fix $x\in\Xcal$, and define 
\[
    N_x
    =
    \sum_{\substack{\ell\in\Ical_{g's'}\\ \ell\ne j}}
    \1\!\left\{\norm{X_\ell-x}\le rh\right\}.
\]
which is the total number of the comparison units other than unit $j$ that falls into $B(x,rh)$. Since the sample is i.i.d
\[
    N_x\sim\operatorname{Binomial}(n_{g's'}-1,p_h(x)),
\]
where
\[
    p_h(x)
    =
    \Pr\!\left(\norm{X_\ell-x}\le rh
    \mid \ell\in\Ical_{g's'}\right).
\]
By the lower bound of the density function and convex support condition imposed assumption \ref{assumption:regularity} (i),
\begin{align*}
    p_h(x)
    &=
    \int_{\Xcal\cap B(x,rh)} f_{g's'}(z)\,dz\\
    &\ge
    \underline f\,\operatorname{Leb}\!\left(\Xcal\cap B(x,rh)\right)\\
    &\ge
    \underline f\,c_{\Xcal}Leb(B(x,rh))\\
    &\ge  \underline f\,c_{\Xcal} \frac{\pi^{q/2}}{\Gamma\left(\frac{q}{2}+1\right)}r^qh^q
\end{align*}
The second inequality follows from theorem $5$ (pp. 12) of \cite{Viel_2025_wp}. $\Gamma(\cdot)$ is the Gamma function.
Thus, for some constant $C>0$,
\[
    p_h(x)\ge Ch^q
\]
uniformly over $x\in\Xcal$.

On the other hand, 
\begin{align*}
    \E[N_{x}\mid j\in \mathcal{I}_{g's'}]=(n_{g's'}-1)p_h(x) \geq \frac{Cn_{g's'}h^q}{2}
\end{align*}
uniformly in $x$. Since $n_{g's'}h^q\to\infty$, $\E[N_{x}\mid j\in \mathcal{I}_{g's'} ] \to\infty$ uniformly in $x$.

For the same positive integer $k$ as above, 
\begin{align*}
    &\E\!\left[(N_x+1)^{-k}\mid j\in \mathcal{I}_{g's'} \right]\\&= \E\!\left[\left.(N_x+1)^{-k}\1\left(N_x> \frac{\E[N_x\mid j\in \mathcal{I}_{g's'}  ]}{2}\right)\right\lvert j\in \mathcal{I}_{g's'} \right]\\
    &+ \E\!\left[\left.(N_x+1)^{-k}\1\left(N_x\leq \frac{\E[N_x\mid j\in \mathcal{I}_{g's'} ]}{2}\right)\right\lvert j\in \mathcal{I}_{g's'} \right]  \\
    &\le
    \left(\frac{2}{\E[N_x\mid j\in \mathcal{I}_{g's'}  ]}\right)^k
    +
    \Pr\!\left(\left. N_x<\frac{\E[N_x\mid j\in \mathcal{I}_{g's'}  ]}{2}\right \lvert j\in \mathcal{I}_{g's'} \right).
\end{align*}
By the Chernoff bound for binomial random variable,
\[
     \Pr\!\left(\left. N_x<\frac{\E[N_x\mid j\in \mathcal{I}_{g's'}  ]}{2}\right \lvert j\in \mathcal{I}_{g's'} \right)
    \le
    \exp\!\left(-\frac{\E[N_x\lvert j\in \mathcal{I}_{g's'}  ]}{8}\right).
\]
Since $\E[N_{x} \lvert j\in \mathcal{I}_{g's'} ] \to\infty$ uniformly in $x$ and exponential decay dominates every fixed polynomial, for sufficiently large $n_{g's'}$,
\[
    \exp\!\left(-\frac{\E[N_x \lvert j\in \mathcal{I}_{g's'} ]}{8}\right)\leq
    \left(\frac{2}{\E[N_x\lvert j\in \mathcal{I}_{g's'}  ]}\right)^k.
\]
Therefore, for sufficiently large $n_{g's'}$
\begin{align*}
    \E\!\left[(N_x+1)^{-k}\lvert j\in \mathcal{I}_{g's'}  \right]
    \le
    \frac{2^{k+1}}{(\E[N_x \lvert j\in \mathcal{I}_{g's'} ])^k } \leq \frac{C_k}{(n_{g's'}h^q)^k}.
    \label{eq:inverse-count}
\end{align*}

Now use the definition of the pairwise weight $w_{ij}$ under the kernel matching
\[
    w_{ij}
    =
    \frac{
        K\!\left(\dfrac{X_j-x}{h}\right)
    }{
        \displaystyle\sum_{\ell\in\Ical_{g's'}}
        K\!\left(\dfrac{X_\ell-x}{h}\right)
    }
    =
    \frac{
        K\!\left(\dfrac{X_j-x}{h}\right)
    }{
      K\!\left(\dfrac{X_j-x}{h}\right)+  \displaystyle\sum_{\substack{\ell\in\Ical_{g's'}\\ \ell\ne j}}
        K\!\left(\dfrac{X_\ell-x}{h}\right)
    }
\]
By assumption \ref{assumption:kernel}(i), every comparison observation satisfying $\norm{X_\ell-x}\le rh$ contributes at least $\underline K$ to the denominator. Hence,
\[
    \sum_{\ell\in\Ical_{g's'}}
    K\!\left(\frac{X_\ell-x}{h}\right)
    \ge
    K\!\left(\frac{X_j-x}{h}\right)+\underline K N_x.
\]
When $N_x\ge1$, using the fact that the support of the kernel $K$ is bounded above by $\overline{K}$, we have
\[
    w_{ij}\le \frac{\overline K}{\underline K N_x}.
\]
In addition, combining the fact that $w_{ij}\leq1$, so
\begin{align*}
    w_{ij}\le \frac{2\overline{K}}{\underline{K}(N_x+1)}.
\end{align*}    
For fixed $x$, $N_x$ depends only on comparison observations other than $j$, and is therefore independent of $X_j$. As a result, for $k=1,\ldots,5$,
\begin{align*}
    &\E\!\left[
        \1\!\left\{\norm{X_j-x}\le Rh\right\}w_{ij}^k\lvert j\in \mathcal{I}_{g's'}
    \right]\\&\leq\E\!\left[\left.
        \1\!\left\{\norm{X_j-x}\le Rh\right\}\left(\frac{2\overline{K}}{\underline{K}(N_x+1)} \right)^k\right \lvert j\in \mathcal{I}_{g's'}
    \right]  \\
    &\le
    \left(\frac{2\overline{K}}{\underline{K}} \right)^k\Pr\!\left(\norm{X_j-x}\le Rh \lvert j\in \mathcal{I}_{g's'} \right)
    \E\!\left[(N_x+1)^{-k}\lvert j\in \mathcal{I}_{g's'}  \right].
\end{align*}
The upper density bound and the convex support condition imposed by assumption \ref{assumption:regularity} yields
\begin{align*}
    \Pr\!\left(\norm{X_j-x}\le Rh \lvert j\in \mathcal{I}_{g's'} \right)
    &=
    \int_{\Xcal\cap B(x,Rh)} f_{g's'}(z)\,dz\\
    &\le
    \overline f\,\operatorname{Leb}(B(x,Rh))\\
    &\le \overline{f}\frac{\pi^{q/2}}{\Gamma\left(\frac{q}{2}+1\right)}R^qh^q.
\end{align*}
Consequently, for sufficiently large $n_{g's'}$
\begin{align*}
    \E\!\left[
        \1\!\left\{\norm{X_j-x}\le Rh\right\}w_{ij}^k\lvert j\in \mathcal{I}_{g's'}
    \right] &\leq \left(\frac{2\overline{K}}{\underline{K}} \right)^k\overline{f}\frac{\pi^{q/2}}{\Gamma\left(\frac{q}{2}+1\right)}R^qh^q\frac{C_k}{(n_{g's'}h^q)^k}\\
    &=\left(\frac{2\overline{K}}{\underline{K}} \right)^k\overline{f}\frac{\pi^{q/2}}{\Gamma\left(\frac{q}{2}+1\right)}R^qC_kh^q(n_{g's'}h^q)^{-k}
\end{align*}
Lastly, 
\begin{align*}
    \E\!\left[
        w_{ij}
        \lvert j\in \mathcal{I}_{g's'},   \{X_\ell:\ell\in\Ical_{g's'}\}
    \right]&=\int_{\Xcal} f_{gs}(x)w_{ij}\,dx\\
    &=\int_{\Xcal} \1\!\left\{\norm{x-X_j}\le Rh\right\} f_{gs} (x)w_{ij}\,dx \\
    &\leq \left(\int_{\Xcal} \1\!\left\{\norm{x-X_j}\le Rh\right\} \,dx \right)^{\frac{k-1}{k}}\left(\int_{\Xcal} f_{gs}^k (x)w_{ij}^k \,dx \right)^{\frac{1}{k}}\\
    &=\bigl[\operatorname{Leb}(B(X_j,Rh))\bigr]^{\frac{k-1}{k}}\left(\int_{\Xcal} f_{gs}^k (x)w_{ij}^k \,dx \right)^{\frac{1}{k}}
\end{align*}
The inequality follows from the H\"older's inequality. Thus,
\begin{align*}
   \left(\E\!\left[w_{ij}\lvert j\in \mathcal{I}_{g's'}, \{X_\ell:\ell\in\Ical_{g's'}\}\right]\right)^k&\leq \bigl[\operatorname{Leb}(B(X_j,Rh))\bigr]^{k-1}\times
    \int_{\Xcal}
    f_{gs}(x)^kw_{ij}^k\,dx \\
   &= \bigl[\operatorname{Leb}(B(X_j,Rh))\bigr]^{k-1}\times
    \int_{\Xcal}
    \1\!\left\{\norm{x-X_j}\le Rh\right\}
    f_{gs}(x)^kw_{ij}^k\,dx \\
    &\leq \overline{f}\left(\frac{\pi^{q/2}}{\Gamma\left(\frac{q}{2}+1\right)}R^qh^q \right)^{k-1}\int_{\Xcal}
    \1\!\left\{\norm{x-X_j}\le Rh\right\}
    w_{ij}^k\,dx
\end{align*}
Take expectation on both sides with respect to all the comparison group covariates, and the event $j\in\mathcal{I}_{g's'}$,
\begin{align*}
    &\E\left[\left(\E\!\left[w_{ij}\lvert j\in \mathcal{I}_{g's'}, \{X_\ell:\ell\in\Ical_{g's'}\}\right]\right)^k\right]\\
    &\leq \overline{f}\left(\frac{\pi^{q/2}}{\Gamma\left(\frac{q}{2}+1\right)}R^qh^q \right)^{k-1}\E\left[\left.\int_{\Xcal}
    \1\!\left\{\norm{x-X_j}\le Rh\right\}
    w_{ij}^k\,dx \right \lvert j\in \mathcal{I}_{g's'} \right] \\
    &=\overline{f}\left(\frac{\pi^{q/2}}{\Gamma\left(\frac{q}{2}+1\right)}R^q \right)^{k-1}h^{q(k-1)}\int_{\Xcal}\E\left[
    \1\!\left\{\norm{x-X_j}\le Rh\right\}
    w_{ij}^k\lvert j\in \mathcal{I}_{g's'} \right]\,dx \\
    &\leq C_k h^{q(k-1)}\int_{\Xcal}h^q(n_{g's'}h^q)^{-k}\,dx \\
    &\leq \frac{C_k}{n_{g's'}^k}
\end{align*}
The first equality follows by the fact that the integrand is non-negative so we can apply the Tonelli's theorem to exchange the expectation and integration.

Therefore,
\begin{align*}
     \E\!\left[w_j^5\lvert j\in \mathcal{I}_{g's'}  \right]&\leq \sum_{k=1}^5 C(5,k)n_{gs}^k\E\left[\left(\E\!\left[
        w_{ij}
        \lvert j\in \mathcal{I}_{g's'}, \{ X_j:j\in\Ical_{g's'} \}
    \right]\right)^k\right] \\
    &\leq \sum_{k=1}^5 C(5,k)C_k\left(\frac{n_{gs}}{n_{g's'}}\right)^k
\end{align*}
Notice that $\frac{n_{gs}}{n_{g's'}}\rightarrow \frac{p_{gs}}{p_{g's'}}>0$, so $\left(\frac{n_{gs}}{n_{g's'}}\right)^k=O(1)$. This finally shows that,
\[\E\!\left[w_j^5 \mid j\in \mathcal{I}_{g's'}\right] =O(1)  \]
(iv) is shown in lemma \ref{lemma:reallocationbound}.
\end{proof}

\begin{proof}
    Proof of lemma \ref{lemma:regressionweight}. To show the first equality, apply the law of total probability, we have:
    \begin{align*}
        &\E\left[\left.\frac{(1-p_{11})L(D|G,S,X)}{\E[(D-L(D|G,S,X))^2]} X \right| D=0\right] =\E\left[\left.\frac{(1-p_{11})L(D|G,S,X)}{\E[(D-L(D|G,S,X))^2]} X \right| G=0,S=1\right]\frac{p_{01}}{1-p_{11}}\\
    &+ \E\left[\left.\frac{(1-p_{11})L(D|G,S,X)}{\E[(D-L(D|G,S,X))^2]} X \right| G=1,S=0\right]\frac{p_{10}}{1-p_{11}}\\ 
    &+\E\left[\left.\frac{(1-p_{11})L(D|G,S,X)}{\E[(D-L(D|G,S,X))^2]} X \right| G=0,S=0\right]\frac{p_{00}}{1-p_{11}} \\
    &=\E\left[\left.\frac{p_{01}L(D|G,S,X)}{\E[(D-L(D|G,S,X))^2]} X \right| G=0,S=1\right] + \E\left[\left.\frac{p_{10}L(D|G,S,X)}{\E[(D-L(D|G,S,X))^2]} X \right| G=1,S=0\right] \\ 
    &+\E\left[\left.\frac{p_{00}L(D|G,S,X)}{\E[(D-L(D|G,S,X))^2]} X \right| G=0,S=0\right]\\
    \end{align*}
    To prove the second equality, it suffices to show that
    \begin{align*}
        \E[p_{11}\left(1-L(D|G,S,X)\right)X | G=1,S=1 ] =\E[(1-p_{11})L(D|G,S,X)X| D=0  ]
    \end{align*}
    Note that 
    \begin{align*}
         &\E[p_{11}\left(1-L(D|G,S,X)\right) | G=1,S=1 ] -\E[(1-p_{11})L(D|G,S,X)| D=0  ]\\
         &=\E[(D(1-L(D|G,S,X)) - (1-D)L(D|G,S,X))X ]\\
         &=\E[(D-L(D|G,S,X))X] \\
         &=0
    \end{align*}
    The last equality follows from the fact that $X$ is orthogonal to the regression residual $D-L(D|G,S,X)$
\end{proof}

\begin{proof} Proof of lemma \ref{lemma:convergence}.\\
Let $\G=(G_1,\ldots,G_n)$, $\bs=(S_1,\ldots,S_n)$, and $\X=(X_1',\ldots,X_n')'$, 
Now,
\begin{align*}
T_n=\frac{n}{n_{11}}&\left\{\frac{1}{n}\isum\left[\1(i\in \mathcal{I}_{11})-\1(i\in \mathcal{I}_{gs})w_i\right](Z_i-\mu_Z(G_i,S_i,X_i))\right.\\
&+\frac{1}{n}\sum_{i\in \mathcal{I}_{11}}(\mu_Z(1,1,X_i)-\mu_Z(g,s,X_i))\\
&+\left.\frac{1}{n}\isum\left[\1(i\in \mathcal{I}_{11})-\1(i\in \mathcal{I}_{gs})w_i\right]\mu_Z(g,s,X_i)\right\}\\
:&=\frac{n}{n_{11}}\left(T_{1,n}+T_{2,n}+T_{3,n}\right).
\end{align*}
We look at these three terms separately. First, $\E[T_{1,n}|\G,\bs,\X]=0$ and
\begin{align*}
\V[T_{1,n}|\G,\bs,\X]&=\frac{1}{n^2}\isum \left[\1(i\in \mathcal{I}_{11})-\1(i\in \mathcal{I}_{gs})w_i\right]^2\sigma^2_Z(G_i,S_i,X_i)
\end{align*}
from which 
\begin{align*}
\V[T_{1,n}]&=\frac{1}{n}\E\left[\left[\1(i\in \mathcal{I}_{11})-\1(i\in \mathcal{I}_{gs})w_i\right]^2\sigma^2_Z(G_i,S_i,X_i)\right]\\
&\le \frac{C}{n}\left(1+p_{gs}\E\left[w_i^2|i\in \mathcal{I}_{gs}\right]\right)\\
&\to 0
\end{align*}
the last line follows by assumption \ref{assumption:matching}. Thus, $T_{1,n}=o_p(1)$. Next,
\begin{align*}
   &Var[\1(i\in \mathcal{I}_{gs})(\mu_Z(1,1,X_i)-\mu_Z(g,s,X_i))  ]  \\
   &\leq \E[(\mu_Z(1,1,X_i)-\mu_Z(g,s,X_i))^2]\\
   &\leq  \E\left[ \left(\abs{\mu_Z(1,1,X_i)} + \abs{\mu_Z(g,s,X_i)} \right)^2 \right]\\
   &<\infty 
\end{align*}
The second to last inequality follows by the fact that the conditional mean function is Lipschitz continuous and the support of $X_i$ is bounded. Apply the WLLN,
\begin{align*}
T_{2,n}&=\frac{1}{n}\sum_{i\in \mathcal{I}_{11}}(\mu_Z(1,1,X_i)-\mu_Z(g,s,X_i))\\
&\plim \E\left[\mu_Z(1,1,X)-\mu_Z(g,s,X)|G=1,S=1\right]p_{11}\\
&=\left(\mu_{Z}(1,1)-\E[\mu_Z(g,s,X)|G=1,S=1]\right)p_{11}.
\end{align*}
Finally, for the third term,
\begin{align*}
T_{3,n}&=\frac{1}{n}\isum\left[\1(i\in \mathcal{I}_{11})-\1(i\in \mathcal{I}_{gs})w_i\right]\mu_Z(g,s,X_i)\\
&=\frac{1}{n}\sum_{i\in \mathcal{I}_{11}}\mu_Z(g,s,X_i)-\frac{1}{n}\sum_{j\in\mathcal{I}_{gs}}w_i\mu_Z(g,s,X_j) \\
&=\frac{1}{n}\sum_{i\in \mathcal{I}_{11}}\mu_Z(g,s,X_i)- \frac{1}{n}\sum_{j\in\mathcal{I}_{gs}}\sum_{i\in\mathcal{I}_{11}}w_{ij}\mu_Z(g,s,X_j) \\
&=\frac{1}{n}\sum_{i\in \mathcal{I}_{11}}\sum_{j\in\mathcal{I}_{gs}}w_{ij} \mu_Z(g,s,X_i)- \frac{1}{n}\sum_{i\in\mathcal{I}_{11}}  \sum_{j\in\mathcal{I}_{gs}}w_{ij}\mu_Z(g,s,X_j) \\
&=\frac{1}{n}\sum_{i\in \mathcal{I}_{11}}\sum_{j\in\mathcal{I}_{gs}}w_{ij}\left[ \mu_Z(g,s,X_i)-\mu_Z(g,s,X_j)\right] \\
&\leq \frac{C_2}{n}\sum_{i\in \mathcal{I}_{11}}\sum_{j\in\mathcal{I}_{gs}}w_{ij}\norm{X_i-X_j}
\end{align*}
The third equality follows from $\sum_{j\in\mathcal{I}_{gs}}w_{ij}=1$. Notice the last term is non-negative because $w_{ij}\geq 0$ By Markov inequality,
\begin{align*}
   P\left( \frac{C}{n}\sum_{i\in \mathcal{I}_{11}}\sum_{j\in\mathcal{I}_{gs}}w_{ij}\norm{X_i-X_j}>\epsilon\right)&\leq \frac{\E\left[\frac{C}{n}\sum_{i\in \mathcal{I}_{11}}\sum_{j\in\mathcal{I}_{gs}}w_{ij}\norm{X_i-X_j}  \right]}{\epsilon} \\
   &= \frac{C\E\left[\1(i\in \mathcal{I}_{11})\sum_{j\in\mathcal{I}_{gs}}w_{ij}\norm{X_i-X_j}  \right]}{\epsilon} \\
   &= \frac{Cp_{11}\E\left[\left.\sum_{j\in\mathcal{I}_{gs}}w_{ij}\norm{X_i-X_j}\right | G_i=1, S_i=1  \right]}{\epsilon} \\
   &\rightarrow 0
\end{align*}
The first inequality follows from the fact that the sample is i.i.d. and the pairwise weight is permutation equivariant.
The last line follows from assumption \ref{assumption:matching}(ii). Thus $T_{3,n}=o_p(1)$. Therefore, 
\begin{align*}
T_n&=\frac{n}{n_{11}}\left(T_{1,n}+T_{2,n}+T_{3,n}\right)\plim  \mu_{Z}(1,1)-\E[\mu_Z(g,s,X)|G=1,S=1]  
\end{align*}
as required.
\end{proof}

\begin{proof}
    Proof of lemma \ref{lemma:varianceconvergence}.\\
   For the first set of claim, by Cheyshev's inequality, 
    \begin{align*}
     &P\left( \abs{  \frac{1}{n}\sum_{i\in\mathcal{I}_{11}}\sigma^2_{11,l}(X_{i})- \E\left[\frac{1}{n}\sum_{i\in\mathcal{I}_{11}}\sigma^2_{11,l}(X_{i}) \right] }>\epsilon \right) \\
     &\leq \frac{Var\left(  \frac{1}{n}\sum_{i\in\mathcal{I}_{11}}\sigma^2_{11,l}(X_{i}) \right)}{\epsilon^2} \\
     &=\frac{n_{11}Var\left( \sigma^2_{11,l}(X_{i}) \right)}{n^2\epsilon^2} \\
     &\leq \frac{n_{11}\E\left[\left( \sigma^2_{11,l}(X_{i}) \right)^2\right]}{n^2\epsilon^2} \\
     &\leq \frac{n_{11}\overline{\sigma}^2}{n^2\epsilon^2} \\
     &\rightarrow o_p(1)
    \end{align*}
    The last inequality follows from assumption \ref{assumption:regularity}(iii). The last line follows from $n_{11} \asymp n$.\\
    This shows that 
    \begin{align*}
        \frac{1}{n}\sum_{i\in\mathcal{I}_{11}}\sigma^2_{11,l}(X_{i})- \E\left[\frac{1}{n}\sum_{i\in\mathcal{I}_{11}}\sigma^2_{11,l}(X_{i}) \right] &=o_p(1)
    \end{align*}
    On the other hand,
    \begin{align*}
        \E\left[ \frac{1}{n}\sum_{i\in \mathcal{I}_{11}}\sigma^2_{11,l}(X_i) \right] \leq p_{11}\overline{\sigma}^2 \implies \E\left[ \frac{1}{n}\sum_{i\in \mathcal{I}_{11}}\sigma^2_{11,l}(X_i) \right] =O(1)
    \end{align*}
    The argument above shows that 
    \begin{align*}
        Var\left(  \frac{1}{n}\sum_{i\in\mathcal{I}_{11}}\sigma^2_{11,l}(X_{i}) \right) =o(1)
    \end{align*}
    Therefore,
     \begin{align*}
        \E\left[\left(  \frac{1}{n}\sum_{i\in\mathcal{I}_{11}}\sigma^2_{11,l}(X_{i}) \right)^2\right] =O(1)
    \end{align*}
    Now,
    \begin{align*}
        &\E\left[\abs{\frac{1}{n_{11}}\sum_{i\in\mathcal{I}_{11}}\sigma^2_{11,l}(X_{i})-\frac{1}{p_{11}}\E\left[\frac{1}{n}\sum_{i\in\mathcal{I}_{11}}\sigma^2_{11,l}(X_{i})\right]}\right]\\
        &=\E\left[\abs{\left(\frac{n}{n_{11}}-\frac{1}{p_{11}} \right)\frac{1}{n}\sum_{i\in\mathcal{I}_{11}}\sigma^2_{11,l}(X_{i}) + \frac{1}{p_{11}}\left( \frac{1}{n}\sum_{i\in\mathcal{I}_{11}}\sigma^2_{11,l}(X_{i}) -\E\left[\frac{1}{n}\sum_{i\in\mathcal{I}_{11}}\sigma^2_{11,l}(X_{i})\right]  \right)}\right] \\
        &\leq \E\left[\abs{\left(\frac{n}{n_{11}}-\frac{1}{p_{11}} \right)\frac{1}{n}\sum_{i\in\mathcal{I}_{11}}\sigma^2_{11,l}(X_{i}) }\right] +\E\left[\abs{ \frac{1}{p_{11}}\left( \frac{1}{n}\sum_{i\in\mathcal{I}_{11}}\sigma^2_{11,l}(X_{i}) -\E\left[\frac{1}{n}\sum_{i\in\mathcal{I}_{11}}\sigma^2_{11,l}(X_{i})\right]  \right)}\right] \\
        &\leq \E\left[\abs{\frac{n}{n_{11}}-\frac{1}{p_{11}} }^2\right]^{1/2}\E\left[\abs{\frac{1}{n}\sum_{i\in\mathcal{I}_{11}}\sigma^2_{11,l}(X_{i}) }^2\right]^{1/2}+\frac{1}{p_{11}}\left(Var\left(\frac{1}{n}\sum_{i\in\mathcal{I}_{11}}\sigma^2_{11,l}(X_{i}) \right) \right)^{1/2} \\
        &=o(1)O(1) + o(1) \\
        &=o(1)
    \end{align*}
    Then by Jensen's inequality,
    \begin{align*}
        &\abs{ \E\left[ \frac{1}{n_{11}}\sum_{i\in\mathcal{I}_{11}}\sigma^2_{11,l}(X_{i})\right]-\frac{1}{p_{11}}\E\left[\frac{1}{n}\sum_{i\in\mathcal{I}_{11}}\sigma^2_{11,l}(X_{i})\right] }\\
        &= \abs{ \E\left[ \frac{1}{n_{11}}\sum_{i\in\mathcal{I}_{11}}\sigma^2_{11,l}(X_{i})-\frac{1}{p_{11}}\E\left[\frac{1}{n}\sum_{i\in\mathcal{I}_{11}}\sigma^2_{11,l}(X_{i})\right]  \right] } \\
        &\leq \E\left[\abs{\frac{1}{n_{11}}\sum_{i\in\mathcal{I}_{11}}\sigma^2_{11,l}(X_{i})-\frac{1}{p_{11}}\E\left[\frac{1}{n}\sum_{i\in\mathcal{I}_{11}}\sigma^2_{11,l}(X_{i})\right]}\right]\\
        &=o_p(1)
    \end{align*}
    Then by Markov inequality, for any $\epsilon>0$,
    \begin{align*}
       & P\left(\abs{\frac{1}{n_{11}}\sum_{i\in\mathcal{I}_{11}}\sigma^2_{11,l}(X_{i})-\frac{1}{p_{11}}\E\left[\frac{1}{n}\sum_{i\in\mathcal{I}_{11}}\sigma^2_{11,l}(X_{i})\right]  } \right)\\
       &\leq \frac{1}{\epsilon} \E\left[\abs{\frac{1}{n_{11}}\sum_{i\in\mathcal{I}_{11}}\sigma^2_{11,l}(X_{i})-\frac{1}{p_{11}}\E\left[\frac{1}{n}\sum_{i\in\mathcal{I}_{11}}\sigma^2_{11,l}(X_{i})\right]}\right] \rightarrow 0
    \end{align*}
    So, 
    \begin{align*}
        \frac{1}{n_{11}}\sum_{i\in\mathcal{I}_{11}}\sigma^2_{11,l}(X_{i})-\frac{1}{p_{11}}\E\left[\frac{1}{n}\sum_{i\in\mathcal{I}_{11}}\sigma^2_{11,l}(X_{i})\right]  =o_p(1)
    \end{align*}
    Combining the two results above,
    \begin{align*}
        &\frac{1}{n_{11}}\sum_{i\in\mathcal{I}_{11}}\sigma^2_{11,l}(X_{i})-\E\left[\frac{1}{n_{11}}\sum_{i\in\mathcal{I}_{11}}\sigma^2_{11,l}(X_{i})\right] \\
        &= \frac{1}{n_{11}}\sum_{i\in\mathcal{I}_{11}}\sigma^2_{11,l}(X_{i})-\frac{1}{p_{11}}\E\left[\frac{1}{n}\sum_{i\in\mathcal{I}_{11}}\sigma^2_{11,l}(X_{i})\right] + \frac{1}{p_{11}}\E\left[\frac{1}{n}\sum_{i\in\mathcal{I}_{11}}\sigma^2_{11,l}(X_{i})\right]  - \E\left[ \frac{1}{n_{11}}\sum_{i\in\mathcal{I}_{11}}\sigma^2_{11,l}(X_{i})\right]\\
        &=o_p(1)
    \end{align*}
    This finishes the proof of the first set of claims.
    
    For the second claim, let $\Omega_n= \sum_{i\in\mathcal{I}_{gs}} w_{i}^2\sigma^2_{gs,l}(X_{i}) $, and $ \Omega_n^{(j)}=\1(j\in \mathcal{I}_{gs})\left( w_j^{(j)}\right)^2\sigma^2_{gs,l}(X_{j}') + \sum_{i\in\mathcal{I}_{gs}\setminus \{j\}} \left(w_{i}^{(j)}\right)^2\sigma^2_{gs,l}(X_{i})  $ be the counterpart of $\Omega_n$ with $(G_j,S_j,X_j)$ being replaced by $(G_j',S_j',X_j')$. 
    \begin{align*}
       \abs{ \Omega_n -\Omega_n^{(j)}} &=\abs{\1(j\in \mathcal{I}_{gs} ) w_j^2\sigma^2_{gs,l}(X_{j})  -\1(j\in \mathcal{I}_{gs}  ) \left(w_j^{(j)}\right)^2\sigma^2_{gs,l}(X_{j}') + \sum_{i\in\mathcal{I}_{gs}\setminus \{1\}} \left( w_i^2-\left(w_{i}^{(j)}\right)^2\right)\sigma^2_{gs,l}(X_{i}) }\\
       &\leq \overline{\sigma}^2\left(\1(j\in \mathcal{I}_{gs}) (w_j^2)+ \1'(j\in \mathcal{I}_{gs} )\left(w_j^{(j)}\right)^2 + \sum_{i\in\mathcal{I}_{gs}\setminus \{j\}} \abs{ w_i^2-\left(w_{i}^{(j)}\right)^2} \right)
    \end{align*}
    By Efron Stein inequality,
    \begin{align*}
        Var(\Omega_n)&\leq \frac{1}{2}\sum_{j=1}^n \E\left[ \left(\Omega_n -\Omega_n^{(j)}\right)^2 \right] \\
        &\leq  \frac{1}{2}\sum_{j=1}^n\E\left[ \overline{\sigma}^2\left(\1(j\in \mathcal{I}_{gs} ) (w_j^2)+ \1'(j\in \mathcal{I}_{gs} ) \left(w_j^{(j)}\right)^2 + \sum_{i\in\mathcal{I}_{gs}\setminus \{j\}} \abs{ w_i^2-\left(w_{i}^{(j)}\right)^2} \right)^2\right] \\
        &\leq \frac{1}{2}\sum_{j=1}^n\E\left[ 3\overline{\sigma}^2\left(\1(j\in \mathcal{I}_{gs} ) (w_j^4)+\1'(j\in \mathcal{I}_{gs} ) \left(w_j^{(j)}\right)^4 + \left( \sum_{i\in\mathcal{I}_{gs}\setminus \{j\}} \abs{ w_i^2-\left(w_{i}^{(j)}\right)^2}\right)^2 \right)\right] \\
        &=\frac{3n\overline{\sigma}^2}{2}\left(\E[\1(j\in \mathcal{I}_{gs} ) w_j^4] +\E\left[\1(j\in \mathcal{I}_{gs} ) \left(w_j^{(j)}\right)^4 \right]+\E\left[ \left( \sum_{i\in\mathcal{I}_{gs}\setminus \{j\}} \abs{ w_i^2-\left(w_{i}^{(j)}\right)^2}\right)^2  \right] \right)\\
        &=\frac{3n\overline{\sigma}^2}{2}(O(1)+O(1)+O(n^{1/2})\\
        &=O(n^{3/2})
    \end{align*}
    The second inequality follows from $(a+b+c)^2\leq 3(a^2+b^2+c^2)$. The second equality follows from assumption \ref{assumption:matching}(iii) and lemma \ref{lemma:weightproperty}(ii).
    
    By Cheyshev's inequality, 
     \begin{align*}
     &P\left( \abs{  \frac{1}{n}\sum_{i\in\mathcal{I}_{gs}} w_{i}^2\sigma^2_{gs,l}(X_{i})-\E\left[ \frac{1}{n}\sum_{i\in\mathcal{I}_{gs}} w_{i}^2\sigma^2_{gs,l}(X_{i})\right] }>\epsilon \right) \\
     &\leq \frac{\Omega_n}{n^2\epsilon^2} =\frac{n^{3/2}}{n^2\epsilon^2} \rightarrow 0
    \end{align*}
    This shows that 
    \begin{align*}
         \frac{1}{n}\sum_{i\in\mathcal{I}_{gs}} w_{i}^2\sigma^2_{gs,l}(X_{i})-\E\left[ \frac{1}{n}\sum_{i\in\mathcal{I}_{gs}} w_{i}^2\sigma^2_{gs,l}(X_{i})\right] & =o_p(1)
    \end{align*}
    On the other hand, 
    \begin{align*}
        \E\left[\frac{1}{n}\sum_{i\in\mathcal{I}_{gs}} w_{i}^2\sigma^2_{gs,l}(X_{i}) \right]&=\frac{1}{n} \isum\E\left[\1(j\in \mathcal{I}_{gs}) w_{i}^2\sigma^2_{gs,l}(X_{i}) \right] \\
        &\leq \overline{\sigma}^2\E\left[ w_{i}^2|i\in \mathcal{I}_{gs} \right]\\
        &=O(1)
    \end{align*}
    The above argument implies that
    \begin{align*}
        Var\left(\frac{1}{n}\sum_{i\in\mathcal{I}_{gs}} w_{i}^2\sigma^2_{gs,l}(X_{i})  \right) =\frac{1}{n^2}Var(\Omega_n)=\frac{1}{n^2}O(n^{3/2}) = o(1)
    \end{align*}
    Therefore,
    \begin{align*}
        \E\left[\left(\frac{1}{n}\sum_{i\in\mathcal{I}_{gs}} w_{i}^2\sigma^2_{gs,l}(X_{i})  \right)^2 \right] = O(1)
    \end{align*}
    Then following the same proof strategy as the proof of the first set of claim, we have:
    \begin{align*}
        \frac{1}{n_{11}}\sum_{i\in\mathcal{I}_{gs}} w_{i}^2\sigma^2_{gs,l}(X_{i})-\E\left[ \frac{1}{n_{11}}\sum_{i\in\mathcal{I}_{gs}} w_{i}^2\sigma^2_{gs,l}(X_{i})\right] & =o_p(1)
    \end{align*}
    This proves the second set of claims.
\end{proof}

\begin{proof}
    Proof of lemma \ref{lemma:weightproperty}. \\
    (i) \begin{align*}
        \E\left[\frac{1}{n^2}\sum_{j\in\mathcal{I}_{gs}}w_{j}^2 \right]&=\frac{1}{n}\E[w_j^2\1(i\in \mathcal{I}_{gs})] \\
        &=\frac{p_{gs}}{n}\E[w_j^2\mid i\in \mathcal{I}_{gs}] \\
        &\leq \frac{1}{n} \left(\E[w_j^{4+\delta}\mid i\in \mathcal{I}_{gs}]\right)^{\frac{2}{4+\delta}} \\
        &= O(n^{-1})
    \end{align*}
    The third line follows from Lyapunov inequality and the last line follows from assumption \ref{assumption:matching}(iii).\\
    So there exists a constant $C$ and $N$, such that, for all $n\geq N$,
    \[ \E\left[\frac{1}{n^2}\sum_{j\in\mathcal{I}_{gs}}w_{j}^2 \right]\leq \frac{C}{N} \].
    For any $\epsilon>0$, pick $M=\frac{C}{\epsilon}$, then for all $n\geq N$, 
    \begin{align*}
        P\left(\frac{1}{n^2}\sum_{j\in\mathcal{I}_{gs}}w_{j}^2 \geq \frac{M}{N}  \right)&\leq \frac{N}{M}\E\left[\frac{1}{n^2}\sum_{j\in\mathcal{I}_{gs}}w_{j}^2 \right] \\
        &\leq \frac{C}{M}=\epsilon
    \end{align*}

    \begin{align*}
        \sum_{i\in \mathcal{I}_{11}}w_{ij}w_{ij'}\leq  \sum_{i\in \mathcal{I}_{11}}w_{ij} \leq \max_{j\in \mathcal{I}_{gs}} w_j
    \end{align*}

    \begin{align*}
        &\max_{j\in \mathcal{I}_{gs}} w_j^2 \leq \sum_{j\in \mathcal{I}_{gs}} w_{gs}^2 \\
        &\implies \left(\max_{j\in \mathcal{I}_{gs}} w_j\right)^{2}\leq  \sum_{j\in \mathcal{I}_{gs}} w_{gs}^2\\
        &\implies \max_{j\in \mathcal{I}_{gs}} w_j \leq \left( \sum_{j\in \mathcal{I}_{gs}} w_{gs}^2 \right)^{1/2} =O_p(n^{1/2})
    \end{align*}
    (ii) Notice that:
    \begin{align*}
      \sum_{j\in \mathcal{I}_{g's'\setminus\{1\}}}\abs{w_j^2-\left(w_j^{(1)}\right)^2}&=\sum_{j\in \mathcal{I}_{g's'\setminus\{1\}}}\left( w_j+w_j^{(1)}\right)\abs{w_j-w_j^{(1)}}  \\
      &\leq \left(\max_{j\in \mathcal{I}_{g's'\setminus\{1\}}} w_j + \max_{j\in \mathcal{I}_{g's'\setminus\{1\}}} w_j^{(1)} \right)\sum_{j\in \mathcal{I}_{g's'\setminus\{1\}}}  \abs{w_j-w_j^{(1)}}
    \end{align*}
    Therefore,
    \begin{align*}
        &\E\left[\left(\sum_{j\in \mathcal{I}_{g's'\setminus\{1\}}}\abs{w_j^2-\left(w_j^{(1)}\right)^2}\right)^2 \right] \\
        &\leq \E\left[ \left(\max_{j\in \mathcal{I}_{g's'\setminus\{1\}}} w_j + \max_{j\in \mathcal{I}_{g's'\setminus\{1\}}} w_j^{(1)} \right)^2 \left(\sum_{j\in \mathcal{I}_{g's'\setminus\{1\}}} \abs{w_j-w_j^{(1)}}\right)^2\right] \\
        & \leq \E\left[ \left(\max_{j\in \mathcal{I}_{g's'\setminus\{1\}}} w_j + \max_{j\in \mathcal{I}_{g's'\setminus\{1\}}} w_j^{(1)} \right)^4\right]^{1/2}  \E\left[ \left(\sum_{j\in \mathcal{I}_{g's'\setminus\{1\}}} \abs{w_j-w_j^{(1)}}\right)^4\right]^{1/2} \\
        &\leq \left( 8\E\left[ \max_{j\in \mathcal{I}_{g's'\setminus\{1\}}} w_j^4 + \max_{j\in \mathcal{I}_{g's'\setminus\{1\}}} \left(w_j^{(1)}\right)^4 \right]\right)^{1/2}  \E\left[ \left(\sum_{j\in \mathcal{I}_{g's'\setminus\{1\}}} \abs{w_j-w_j^{(1)}}\right)^4\right]^{1/2} \\
       & \leq \left( 8\E\left[ \sum_{j\in \mathcal{I}_{g's'\setminus\{1\}}} w_j^4 + \sum_{j\in \mathcal{I}_{g's'\setminus\{1\}}} \left(w_j^{(1)}\right)^4 \right]\right)^{1/2}  \E\left[ \left(\sum_{j\in \mathcal{I}_{g's'\setminus\{1\}}} \abs{w_j-w_j^{(1)}}\right)^4\right]^{1/2} \\
       &=\left( 8\sum_{j\in \mathcal{I}_{g's'\setminus\{1\}}}  \E\left[ w_j^4 + \left(w_j^{(1)}\right)^4 \right]\right)^{1/2}  \E\left[ \left(\sum_{j\in \mathcal{I}_{g's'\setminus\{1\}}} \abs{w_j-w_j^{(1)}}\right)^4\right]^{1/2} \\
       &=O(n^{1/2})O(1) \\
       &=O(n^{1/2})
    \end{align*}
    The second equality follows from Cauchy Schwarz inequality. The third inequality follows from $C_r$ inequality. The second last equality follows from assumption \ref{assumption:matching}(ii) and (iii).

    (iii) By Lyapunov inequality,
    \begin{align*}
        \E[w_j^{r}|j\in \mathcal{I}_{gs}]\leq \E[w_j^{4+\delta}|j\in \mathcal{I}_{gs}]^{\frac{r}{4+\delta}}=O(1)
    \end{align*}
    This proves the first claim. So for any $\epsilon>0$ and $M>0$  there exists $N$ such that, for any $n\geq N$, $ \E[w_j^{r}|j\in \mathcal{I}_{gs}]\leq \frac{M}{\epsilon p_{gs}}$
    By Markov's inequality, for any $n\geq N$
    \begin{align*}
     P\left(\frac{1}{n}\sum_{i\in \mathcal{I}_{gs}}w_j^r  >M \right)& \leq \frac{\E\left[\frac{1}{n}\sum_{i\in \mathcal{I}_{gs}}w_j^r  \right]}{M}\\
     &=\frac{p_{gs}\E\left[ w_j^r | j\in \mathcal{I}_{gs}\right]}{M} \\
     &\leq \epsilon
    \end{align*}
    This proves the second claim. 

    (iv) By Markov inequality,
    \begin{align*}
      P\left(  \frac{1}{n}\sum_{i\in \mathcal{I}_{11}}\sum_{j\in \mathcal{I}_{gs}}w_{ij}\norm{X_i-X_j} \geq \epsilon \right) & \leq \frac{\E\left[ \frac{1}{n}\sum_{i\in \mathcal{I}_{11}}\sum_{j\in \mathcal{I}_{gs}}w_{ij}\norm{X_i-X_j}  \right]}{\epsilon} \\
      &=\frac{\E\left[ \sum_{j\in \mathcal{I}_{gs}}w_{ij}\norm{X_i-X_j}|i\in \mathcal{I}_{11}  \right]p_{11}}{\epsilon} \\
      &\rightarrow 0
    \end{align*}
    The last line follows from assumption \ref{assumption:matching}(ii).\\
    By Slutsky theorem, 
    \begin{align*}
        \frac{1}{n_{11}}\sum_{i\in \mathcal{I}_{11}}\sum_{j\in \mathcal{I}_{gs}}w_{ij}\norm{X_i-X_j} =\frac{n}{n_{11}}\frac{1}{n}\sum_{i\in \mathcal{I}_{11}}\sum_{j\in \mathcal{I}_{gs}}w_{ij}\norm{X_i-X_j} =o_p(1)
    \end{align*}
    On the other hand, since the support is compact,
    \begin{align*}
        \left(\sum_{j\in \mathcal{I}_{gs}}w_{ij}\norm{X_i-X_j}\right)^2&\leq C\cdot \sum_{j\in \mathcal{I}_{gs}}w_{ij}\cdot\left(\sum_{j\in \mathcal{I}_{gs}}w_{ij}\norm{X_i-X_j}\right) \\
        &=C\left(\sum_{j\in \mathcal{I}_{gs}}w_{ij}\norm{X_i-X_j}\right)
    \end{align*}
    Therefore,
    \begin{align*}
        \frac{1}{n_{11}}\sum_{i\in \mathcal{I}_{11}}\left(\sum_{j\in \mathcal{I}_{gs}}w_{ij}\norm{X_i-X_j}\right)^2  \leq \frac{C}{n_{11}}\sum_{i\in \mathcal{I}_{11}}\sum_{j\in \mathcal{I}_{gs}}w_{ij}\norm{X_i-X_j} =o_p(1)
    \end{align*}
\end{proof}

\begin{proof}
     Proof of lemma \ref{lemma:subgroupproportion}.

    Define the event $E_n=\{n_{gs}\geq np_{gs}/2\} $. By Hoeffding's Inequality,
    \begin{align*}
        P(E_n^c)&=P\left(n_{11}-np_{11}<-\frac{np_{11}}{2} \right)\leq exp\left(-\frac{np_{11}^2}{2} \right)
    \end{align*}
    On the event $E_n$,
    \begin{align*}
        \abs{\frac{n}{n_{11}}-\frac{1}{p_{11}}} =\frac{\abs{n_{11}-np_{11}}}{n_{11}p_{11}}\leq \frac{2\abs{n_{11}-np_{11}}}{np_{11}^2}
    \end{align*}
    So that 
    \begin{align*}
        \E\left[\1(E_n)\left(\frac{n}{n_{gs}}-\frac{1}{p_{gs}} \right)^2 \right]&\leq \E\left[\1(E_n) \frac{4(n_{gs}-\E[n_{gs}])^2}{n^2p_{gs}^4} \right] \\
        &\leq \frac{4Var(n_{gs})}{n^2p_{gs}^4} =\frac{4(1-p_{gs})}{np_{gs}^3} = O(n^{-1})
    \end{align*}
    On the event $E_n^c$, $\frac{n}{n_{11}}\leq n$, so
       \begin{align*}
           \E\left[\1(E_n^c)\left(\frac{n}{n_{gs}}-\frac{1}{p_{gs}} \right)^2 \right]& \leq \E\left[\1(E_n^c)\left(n+\frac{1}{p_{gs}} \right)^2 \right] \\
           & =\left(n+\frac{1}{p_{gs}} \right)^2  P(E_n^c) \\
           & \leq \left(n+\frac{1}{p_{gs}} \right)^2 exp\left(-\frac{np_{11}^2}{2} \right) \\
           &=o(n^{-1})
       \end{align*}
       Combine the two results we prove lemma \ref{lemma:subgroupproportion}.
 \end{proof}

 \begin{proof}
     Proof of lemma \ref{lemma:conditionalchebyshev}.

     By conditional Chebyshev's inequality, for any $\epsilon>0$,
     \begin{align*}
         P(\abs{T_n}\geq\epsilon |\G,\bs,\X)\leq \frac{Var(T_n|\G,\bs,\X) }{\epsilon^2}
     \end{align*}
     By law of iterated expectation and the property of probability,
     \begin{align*}
          P(\abs{T_n}\geq\epsilon )=\E[P(\abs{T_n}\geq\epsilon |\G,\bs,\X) ]\leq \min\left\{1, \frac{\E[Var(T_n|\G,\bs,\X) ]}{\epsilon^2} \right\}
     \end{align*}
     Let $Z_n=\min\left\{1, \frac{Var(T_n|\G,\bs,\X) }{\epsilon^2} \right\} $, now for any $0<\delta<1$,
     \begin{align*}
       \E\left[  Z_n \right]&=\E\left[ Z_n\1(Z_n\leq \delta)\right] +\E\left[ Z_n\1(Z_n> \delta)\right] \\
       &\leq \delta + P(Z_n>\delta)\\
       &=\delta +P\left(Var(T_n|\G,\bs,\X) >\delta\epsilon^2\right)\\
     \end{align*}
     This implies that for any $\delta>0$,
     \begin{align*}
         \limsup_{n\rightarrow\infty}\E[Z_n]\leq \delta
     \end{align*}
     Therefore,
     \begin{align*}
        \lim_{n\rightarrow\infty} \E[Z_n]=\lim_{n\rightarrow\infty} \min\left\{1, \frac{\E[Var(T_n|\G,\bs,\X) ]}{\epsilon^2} \right\}=0
     \end{align*}
     Combine the results above, we have: $T_n\xrightarrow{p}0$
 \end{proof}

 \begin{proof}
     Proof of lemma \ref{lemma:reallocationbound}.

     Consider an intermediate status where only unit $1$ is deleted and let $\widetilde w_{ij}^{(1)}$ and $\widetilde w_{j}^{(1)}$ denote the corresponding pairwise weight and matching weight under this status. Then by triangular inequality,
     \begin{align*}
          \sum_{j\in \mathcal{I}_{gs}\setminus\{1\}}\abs{w_j-w_j^{(1)}} =\sum_{j\in \mathcal{I}_{gs}\setminus\{1\}}\abs{w_j-\widetilde w_j^{(1)}+\widetilde w_j^{(1)} -w_j^{(1)}} \leq \sum_{j\in \mathcal{I}_{gs}\setminus\{1\}}\abs{w_j-\widetilde w_j^{(1)}} +\sum_{j\in \mathcal{I}_{gs}\setminus\{1\}}\abs{\widetilde w_j^{(1)} -w_j^{(1)}}
     \end{align*}
     For the first term, there are three cases.

     Case 1: $1\notin \mathcal{I}_{11}\cup \mathcal{I}_{gs}$, then in this case nothing changes, so $\sum_{j\in \mathcal{I}_{gs}\setminus\{1\}}\abs{w_j-\widetilde w_j^{(1)}}=0 $.

     Case 2: $1 \in \mathcal{I}_{11}$, then for all three matching methods, since the pairwise matching weight $w_{ij}$ only uses the information from the reference unit $i$ and the comparison group, so in this case it removes the pairwise matching weight of unit $1$ and leaves all other pairwise matching weight unchanged. Therefore,
     \begin{align*}
         \sum_{j\in \mathcal{I}_{gs}\setminus\{1\}}\abs{w_j-\widetilde w_j^{(1)}}= \sum_{j\in \mathcal{I}_{gs}\setminus\{1\}} w_{ij} =1
     \end{align*}

     Case 3: $1 \in \mathcal{I}_{gs}$, by triangular inequality
     \begin{align*}
         \sum_{j\in \mathcal{I}_{gs}\setminus\{1\}}\abs{w_j-\widetilde w_j^{(1)}}&=\sum_{j\in \mathcal{I}_{gs}\setminus\{1\}}\abs{\sum_{i\in \mathcal{I}_{11}}w_{ij}-\widetilde w_{ij}^{(1)}} \\
         &\leq \sum_{j\in \mathcal{I}_{gs}\setminus\{1\}}\sum_{i\in \mathcal{I}_{11}}\abs{w_{ij}-\widetilde w_{ij}^{(1)} } \\
         &= \sum_{i\in \mathcal{I}_{11}}\sum_{j\in \mathcal{I}_{gs}\setminus\{1\}}  \abs{w_{ij}-\widetilde w_{ij}^{(1)} }
     \end{align*}
     Subcase 1: Nearest neighbor matching with fixed and diverging neighbor. If $1\notin \mathcal{M}(i)$, then $\sum_{j\in \mathcal{I}_{gs}\setminus\{1\}} \abs{w_{ij}-\widetilde w_{ij}^{(1)} }=0 $. If $1\in \mathcal{M}(i)$, then $\sum_{j\in \mathcal{I}_{gs}\setminus\{1\}} \abs{w_{ij}-\widetilde w_{ij}^{(1)} }=\frac{1}{M}=w_{i1}$ , so
     \begin{align*}
         \sum_{i\in \mathcal{I}_{11}}\sum_{j\in \mathcal{I}_{gs}\setminus\{1\}}  \abs{w_{ij}-\widetilde w_{ij}^{(1)} } \leq \sum_{i\in \mathcal{I}_{11}} w_{i1} = w_1
     \end{align*}
     Subcase 2, Kernel matching, 
     \begin{align*}
         &\sum_{i\in \mathcal{I}_{11}}\sum_{j\in \mathcal{I}_{gs}\setminus\{1\}}  \abs{w_{ij}-\widetilde w_{ij}^{(1)} }\\
         &=\sum_{i\in \mathcal{I}_{11}}\sum_{j\in \mathcal{I}_{gs}\setminus\{1\}}  \abs{\frac{K((X_i-X_j)/h)}{\sum_{r\in \mathcal{I}_{gs}}K((X_i-X_r)/h) }- \frac{K((X_i-X_j)/h)}{\sum_{r\in \mathcal{I}_{gs}}K((X_i-X_r)/h)-K((X_i-X_1)/h) }  } \\
         &=\sum_{i\in \mathcal{I}_{11}}\sum_{j\in \mathcal{I}_{gs}\setminus\{1\}}  \frac{ K((X_i-X_j)/h)K((X_i-X_1)/h) }{\sum_{r\in \mathcal{I}_{gs}}K((X_i-X_r)/h) \sum_{r\in \mathcal{I}_{gs}}K((X_i-X_r)/h)-K((X_i-X_1)/h) } \\
         &=\sum_{i\in \mathcal{I}_{11}} \frac{ K((X_i-X_1)/h) }{\sum_{r\in \mathcal{I}_{gs}}K((X_i-X_r)/h)  } \\
         &=\sum_{i\in \mathcal{I}_{11}} w_{i1} \\
         &=w_1
     \end{align*}
     Combine the results from the three cases, 
     \begin{align*}
         \sum_{j\in \mathcal{I}_{gs}\setminus\{1\}}\abs{w_j-\widetilde w_j^{(1)}} \leq w_1 +\1(1\in \mathcal{I}_{11})
     \end{align*}
     For the second term, notice that this case can be viewed as the symmetry case as the first term where we delete the observation $1'$ from the replaced sample, so by the same argument above, we have,
     \begin{align*}
         \sum_{j\in \mathcal{I}_{gs}\setminus\{1\}}\abs{\widetilde w_j^{(1)} -w_j^{(1)}}=\sum_{j\in \mathcal{I}_{gs}\setminus\{1\}}\abs{ w_j^{(1)} - \widetilde w_j^{(1)}} \leq w_{1'} +\1(1'\in \mathcal{I}_{11})
     \end{align*}
     Therefore, 
     \begin{align*}
         \sum_{j\in \mathcal{I}_{gs}\setminus\{1\}}\abs{w_j-w_j^{(1)}} \leq w_1+w_{1'}+\1(1\in \mathcal{I}_{11}) +\1(1'\in \mathcal{I}_{11})
     \end{align*}
     This proves the first claim.
     \begin{align*}
         \E\left[\left(\sum_{j\in \mathcal{I}_{gs}\setminus\{1\}}\abs{w_j-w_j^{(1)}}\right)^4 \right]&\leq \E\left[\left(w_1+w_{1'}+\1(1\in \mathcal{I}_{11}) +\1(1'\in \mathcal{I}_{11})\right)^4 \right] \\
         &\leq 64 \E\left[\left(w_1^4+w_{1'}^4+\1(1\in \mathcal{I}_{11}) +\1(1'\in \mathcal{I}_{11})\right) \right] \\
         &=128 (\E[w_1^4]+\E[\1(1\in \mathcal{I}_{11})  ])\\
         &=128 ( \E[w_1^4|1\in \mathcal{I}_{gs}]p_{gs} +p_{11} )\\
         &=O(1)
     \end{align*}
     
 \end{proof}

   \begin{proof}\label{proof:maximumdistancelimit}
Proof of lemma \ref{lemma:maximumdistancelimit}.\\
Let $\xi>0$ be any fixed constant, we then pick a positive integer $r$ (so $r$ depends on $\xi$), such that $r\xi>1$, then
\begin{align*}
    &P(n_{gs}^{-\xi}\max_{i\in\mathcal{I}_{gs}}n_{gs}^{1/q}\norm{X_i-X_{k_m(i)}}>\epsilon|\G,\bs)\\
    &\leq \sum_{i\in\mathcal{I}_{gs} } P(n_{gs}^{-\xi}n_{gs}^{1/q}\norm{X_i-X_{k_m(i')}}>\epsilon|\G,\bs)\\
    &=\isum \1(i\in \mathcal{I}_{gs})P(n_{gs}^{-\xi}n_{gs}^{1/q}\norm{X_i-X_{k_m(i)}}>\epsilon|\G,\bs)\\
    &=n_{gs} P(n_{gs}^{-r\xi}n_{gs}^{r/q}\norm{X_i-X_{k_m(i)}}^r>\epsilon^r|\G,\bs,G_i=g,S_i=s) \\
    &\leq \frac{\E[n_{gs}^{r/q}\norm{X_i-X_{k_m(i)}}^r| \G,\bs,G_i=g,S_i=s   ]}{\epsilon^r}\frac{n_{gs}}{n_{gs}^{r\xi}}\\
    &\leq \frac{1}{n_{gs}^{r\xi-1}}\frac{C}{\epsilon^r}    
\end{align*}
The first inequality follows from the Bonferroni inequality. The second inequality follows from the Markov's inequality. The third inequality follows from lemma 2 in \cite{Abadie-Imbens_2006_ECMA}. Notice under our assumption, when $n\rightarrow \infty$,  $n_{gs}^{r\xi-1}\xrightarrow{a.s.}\infty$, , so $\frac{1}{n_{gs}^{r\xi-1}}\frac{C}{\epsilon^r}\xrightarrow{a.s.}0$, as $n\rightarrow\infty$. This shows that: 
\[P(n_{gs}^{-\xi}\max_{i\in\mathcal{I}_{gs}}n_{gs}^{1/q}\norm{X_i-X_{k_m(i)}}>\epsilon|\G,\bs) \xrightarrow{a.s.}0\]
On the other hand, by law of iterated expectation
\begin{align*}
        &P(n_{gs}^{-\xi}\max_{i\in\mathcal{I}_{gs}}n_{gs}^{1/q}\norm{X_i-X_{k_m(i)}}>\epsilon)\\
   &=\E\left[ P(n_{gs}^{-\xi}\max_{i\in\mathcal{I}_{gs}}n_{gs}^{1/q}\norm{X_i-X_{k_m(i)}}>\epsilon|\G,\bs)\right]
\end{align*}
and combine with the fact that the conditional expectation is between $[0,1]$, so by dominated convergence theorem,
\[n_{gs}^{-\xi}\max_{i\in\mathcal{I}_{gs}}n_{gs}^{1/q}\norm{X_i-X_{k_m(i)}}=o_p(1)\]
The above equality holds for any fixed $\xi>0$, specifically, pick $\xi=1/q$, then:
\[\max_{i\in \mathcal{I}_{gs}}\left\|X_i - X_{k_m(i)}\right\|=o_p(1)\]
\end{proof}

\begin{proof}
    Proof of lemma \ref{lemma:residualbound}.

    \begin{align*}
        \E\left[\left.\left(\epsilon_{il}^{gs}\right)^4\right\lvert \G,\bs,\X,G_i=g,S_i=s \right] &= \E\left[\left.\left(\Delta Y_{il}- \E[\Delta Y_{il}|G_i=g,S_i=s,X_i]\right)^4\right\lvert G_i=g,S_i=s,X_i \right] \\
        &\leq 8\left(\E[\Delta Y_{il}^4 |G_i=g,S_i=s,X_i] +\left(\E[\Delta Y_{il}|G_i=g,S_i=s,X_i] \right)^4   \right)\\
        &<C
    \end{align*}
    The last inequality follows from assumption \ref{assumption:regularity}.
\end{proof}

\begin{proof}
    Proof of lemma \ref{lemma:equivalentestimator}.

    By definition
    \begin{align*}
        \widehat \tau_l&= \frac{1}{n_{11}}\sum_{i\in\mathcal{I}_{11}}\Delta Y_{il} -\frac{1}{n_{11}}\sum_{i\in\mathcal{I}_{01}}w_i\Delta Y_{il}-\left(\frac{1}{n_{11}}\sum_{i\in\mathcal{I}_{10}}w_i\Delta Y_{il} -\frac{1}{n_{11}}\sum_{i\in\mathcal{I}_{00}}w_i\Delta Y_{il}\right) \\
        &=\frac{1}{n_{11}}\sum_{i\in\mathcal{I}_{11}}\Delta Y_{il} -\frac{1}{n_{11}}\sum_{j\in\mathcal{I}_{01}}w_j\Delta Y_{jl}-\left(\frac{1}{n_{11}}\sum_{j\in\mathcal{I}_{10}}w_j\Delta Y_{jl} -\frac{1}{n_{11}}\sum_{i\in\mathcal{I}_{00}}w_j\Delta Y_{jl}\right) \\
        &=\frac{1}{n_{11}}\sum_{i\in\mathcal{I}_{11}}\Delta Y_{il} -\frac{1}{n_{11}}\sum_{j\in \mathcal{I}_{01}}\sum_{i\in \mathcal{I}_{11}}w_{ij}\Delta Y_{jl}-\frac{1}{n_{11}}\sum_{j\in \mathcal{I}_{10}}\sum_{i\in \mathcal{I}_{11}}w_{ij}\Delta Y_{jl} + \frac{1}{n_{11}}\sum_{j\in \mathcal{I}_{00}}\sum_{i\in \mathcal{I}_{11}}w_{ij}\Delta Y_{jl} \\
        &=\frac{1}{n_{11}}\sum_{i\in\mathcal{I}_{11}}\Delta Y_{il}- \frac{1}{n_{11}}\sum_{i\in \mathcal{I}_{11}}  \sum_{j\in \mathcal{I}_{01}}w_{ij}\Delta Y_{jl}-\frac{1}{n_{11}}\sum_{i\in \mathcal{I}_{11}}  \sum_{j\in \mathcal{I}_{10}}w_{ij}\Delta Y_{jl} + \frac{1}{n_{11}}\sum_{i\in \mathcal{I}_{11}}  \sum_{j\in \mathcal{I}_{00}}w_{ij}\Delta Y_{jl} \\
        &=\frac{1}{n_{11}}\sum_{i\in \mathcal{I}_{11}}\left(\Delta Y_{il}-\sum_{j\in \mathcal{I}_{01}}w_{ij}\Delta Y_{jl} - \sum_{j\in \mathcal{I}_{10}}w_{ij}\Delta Y_{jl} +  \sum_{j\in \mathcal{I}_{00}}w_{ij}\Delta Y_{jl}\right) \\
        &=\frac{1}{n_{11}}\sum_{i\in \mathcal{I}_{11}}\left(\Delta Y_{il}-  \widehat {\Delta Y_{il}^{01}}-\widehat {\Delta Y_{il}^{10}} + \widehat {\Delta Y_{il}^{00}} \right)
    \end{align*}
\end{proof}

\end{document}